\documentclass[aps,twocolumn,prx,superscriptaddress,nofootinbib]{revtex4-2}
\usepackage{bm}
\usepackage[T1]{fontenc}
\usepackage{amsmath,amssymb,mathtools,dsfont}
\usepackage{braket}
\usepackage{graphicx}
\usepackage{xcolor}
\usepackage[hidelinks]{hyperref}
\usepackage{quantikz}
\usepackage{float}
\usepackage[capitalize]{cleveref}

\usepackage{amsthm}
\newtheorem{proposition}{Proposition}
\newtheorem{lemma}{Lemma}
\newtheorem{theorem}{Theorem}
\theoremstyle{definition}

\begin{document}

\title{Adiabatic Error Cancellation in Berry Phase Estimation}

\author{Chusei Kiumi}
\email{c.kiumi.qiqb@osaka-u.ac.jp}
\affiliation{Center for Quantum Information and Quantum Biology, The University of Osaka,
  1-2 Machikaneyama, Toyonaka, Osaka,  560-0043, Japan}
\affiliation{Mathematical Institute, University of Oxford, Woodstock Road, Oxford OX2 6GG, United Kingdom}
\date{\today}

\begin{abstract}
  The Berry phase encodes the geometry of a closed Hamiltonian path, complementing the dynamical phase determined by the accumulated energy. We uncover an adiabatic error-cancellation principle that arises from this geometric character and has no counterpart for energy-based quantities. To extract the Berry phase while canceling the dynamical phase, we combine finite-runtime evolutions generated by $\pm H$ along the loop. This construction exactly cancels the leading $O(T^{-1})$ and all higher odd-order nonoscillatory phase errors. Richardson extrapolation further reduces the leading residual error to an oscillatory contribution whose amplitude is set by the endpoints of the loop alone. Beyond this deterministic cancellation, runtime randomization suppresses the remaining oscillatory contribution, reducing the bias after $r$ Richardson levels to $O(T^{-(2r+2)})$ for fixed $r$. The same average also suppresses the bias from an imperfectly prepared guiding state, so that the admissible preparation infidelity no longer depends on the target accuracy \(\varepsilon_B\). Combining these principles, we obtain a randomized Hadamard-test algorithm for Berry phase estimation that reduces the coherent evolution time to \(O(\varepsilon_B^{-1/(2r+2)})\) while retaining the standard \(O(\varepsilon_B^{-2})\) sample complexity. This turns the geometric origin of the cancellation into an algorithmic advantage and makes Berry phase estimation a promising task for early fault-tolerant quantum computers.
  \end{abstract}
\maketitle

\section{Introduction}

The Berry phase is one of the most fundamental concepts in modern quantum physics~\cite{Simon1983,Berry1984}. While quantum systems were traditionally understood through their energy spectra, the Berry phase revealed physical information in eigenstate geometry inaccessible from energy alone. This insight has profoundly reshaped our understanding of quantum matter, providing the geometric foundation for topological phases characterized by invariants such as the Chern number~\cite{TKNN1982,HasanKane2010,QiZhang2011}, the modern theory of electric polarization in crystalline solids~\cite{KingSmith1993,Resta2000,Vanderbilt2018}, and the anomalous Hall effect~\cite{Nagaosa2010,XiaoChangNiu2010}. Building on these foundations, quantum geometry has become a major frontier in contemporary quantum science, where the Berry curvature and quantum metric shape nonlinear transport, flat-band superconductivity, and interaction-driven phases such as fractional Chern insulators~\cite{Peotta2015,Torma2022,Ahn2022,Roy2014,Liu2025}. Beyond condensed matter, geometric phases serve as resources for quantum sensing~\cite{MartinMartinez2013,Ledbetter2012,Arai2018,Johnsson2020} and, through the Pancharatnam--Berry phase, for wavefront and polarization control in photonics~\cite{Cohen2019,Cisowski2022}.

Geometric quantities are often expected to offer enhanced robustness against noise and control errors compared with energy-based quantities.
This expectation has motivated quantum-information applications of geometric phases, most notably geometric and holonomic quantum computation
~\cite{Zanardi1999,Jones2000,Sjoqvist2012,Leibfried2003}.
Whether geometricity itself guarantees robustness, however, remains debated
~\cite{Colmenar2022,Liu2021}.
This raises a fundamental question: can geometricity lead to intrinsic error-suppression mechanisms absent from energy-based quantities?

This question is particularly relevant to the pursuit of practical quantum advantage.
Quantum computers offer exponential speedups for certain tasks,
with many-body simulation among their most compelling applications
~\cite{Feynman1982,Lloyd1996,Georgescu2014}.
Following the experimental demonstration of quantum advantage
~\cite{Arute2019},
a central challenge is to extend this advantage beyond artificially constructed benchmark problems
to physically meaningful tasks.
Such tasks should be robust against errors arising from quantum processes,
thereby reducing the coherent evolution time and circuit depth required before full fault tolerance~\cite{Preskill2018,Kim2023}.
The growing importance of quantum geometry makes Berry phase estimation an important practical computational task, and demonstrating robustness arising from its geometric nature would establish it as a promising candidate for practical quantum advantage under limited quantum resources.

Classical methods reconstruct Berry phases from overlaps of eigenstates along a parameter path
~\cite{FukuiHatsugaiSuzuki2005,Yu2011,Soluyanov2011},
but generally become intractable for strongly correlated many-body systems.
Recently, Berry phase estimation with a known guiding state has been shown to be BQP-complete~\cite{HayakawaSakamotoKiumi2025}, establishing a provable complexity-theoretic quantum advantage.
In contrast to the classical spectral approach,
several quantum algorithms have been proposed
~\cite{ZhangWei2010,MurtaCatarinaFernandezRossier2020,
TamiyaKohNakagawa2021,HayakawaSakamotoKiumi2025,MootzYao2026},
and they typically access the Berry phase through adiabatic evolution.
As in adiabatic geometric and holonomic gates,
suppressing nonadiabatic errors requires long evolution times,
making finite-runtime adiabatic error a fundamental bottleneck.
Reducing this error and the associated runtime is therefore crucial for practical quantum computation
and faster implementations of holonomic gates.
Yet the structure of the adiabatic phase error remains unresolved,
and no intrinsic error-cancellation mechanism has been identified.

\begin{table*}[t]
  \centering
  \renewcommand{\arraystretch}{1.4}
  \begin{tabular}{lcccc}
    \hline\hline
    \textbf{Method} & \textbf{Leading Phase Error} & \textbf{Runtime Scaling} & \textbf{Admissible \(p_{\mathrm{prep}}\)} & \textbf{Reference} \\
    \hline
    Runtime scaling method~\cite{HayakawaSakamotoKiumi2025}
    &
    $O(T^{-1})$, non-osc.
    &
    $\displaystyle O\!\left(\frac{\dot{H}_{\max}^2}{\Delta_{\min}^3\,\varepsilon_B}\right)$
    &
    $\lesssim\varepsilon_B$
    &
    \cref{lem:phase_error}
    \\[12pt]
    \shortstack[l]{+ Forward--reverse}
    &
    $O(T^{-2})$, osc.\ + non-osc.
    &
    \cref{app:bound}
    &
    $\lesssim\varepsilon_B$
    &
    \cref{thm:cancellation}
    \\[12pt]
    \shortstack[l]{+ Richardson extrapolation}
    &
    $O(T^{-2})$, osc.
    &
    $\displaystyle O\!\left(\frac{\|\dot{H}(0)\|}{\Delta(0)^2\,\varepsilon_B^{1/2}}\right)$
    &
    $\lesssim\varepsilon_B$
    &
    \cref{thm:richardson_m}
    \\[12pt]
    \shortstack[l]{+ Runtime randomization\\ (\(M=2r-1\))}
    &
    $O(T^{-(2r+1)})$, bias
    &
$\displaystyle O\!\left(
  \left[\frac{\|\dot H(0)\|^{2}}
  {\Delta(0)^{4}\,\Delta_{\min}^{2r-1}\,\varepsilon_B}\right]^{1/(2r+1)}
  \right)$
&
    $\lesssim\varepsilon_B^{2/(2r+1)}$
&
    \shortstack{\cref{thm:randomized_richardson}\\ \cref{thm:prep_supp}}
    \\[12pt]
    \shortstack[l]{+ Runtime randomization\\ (\(M=2r+2\))}
    &
    $O(T^{-(2r+2)})$, non-osc.
    &
    $O\!\left(\varepsilon_B^{-1/(2r+2)}\right)$
    &
    \(\varepsilon_B\)-independent
    &
    \shortstack{\cref{thm:randomized_richardson}\\ \cref{thm:prep_supp}}
    \\[12pt]
    \hline\hline
  \end{tabular}
  \caption{Comparison of Berry phase estimation algorithms for target additive error \(\varepsilon_B\).
  \(r\ge1\) counts Richardson levels, and \(M\) denotes the decay order of the weighted characteristic functions.
  Runtime-randomized rows use the Hadamard-test implementation of \cref{sec:hadamard}, with sample
  complexity \(O(\varepsilon_B^{-2})\); runtime counts only coherent evolution time. We assume
  \(\dot H(0)=\dot H(1)\) and omit the branch-resolution step of
  Ref.~\cite{HayakawaSakamotoKiumi2025}, whose cost is independent of
  \(\varepsilon_B\). Each row adds one ingredient: the forward--reverse protocol cancels
  the \(O(T^{-1})\) term (\cref{thm:cancellation}), Richardson extrapolation removes
  the non-oscillatory \(O(T^{-2})\) residual (\cref{thm:richardson_m}), and runtime
  randomization suppresses the remaining oscillatory term
  (\cref{thm:randomized_richardson}), either with \(M=2r-1\), giving the explicit
  gap-dependent bound, or with \(M=2r+2\), suppressing it below the
  non-oscillatory remainder \(O(T^{-(2r+2)})\), whose constants are not explicit. The fourth column lists the admissible guiding-state infidelity (\cref{sec:imperfect}): without randomization, \(p_{\rm prep}\lesssim\varepsilon_B\) is required, and at \(M=2r+2\) the admissible infidelity becomes independent of \(\varepsilon_B\).}
  \label{tab:comparison}
\end{table*}

In this work, we uncover an adiabatic error-cancellation principle rooted in the geometric nature of the Berry phase.
Using adiabatic perturbation theory~\cite{RigolinOrtizPonce2008}, we derive a systematic large-$T$ expansion of the finite-runtime phase error.
Such an order-by-order expansion, rather than the norm-based adiabatic bounds~\cite{AvronSeilerYaffe1987,Nenciu1993,JansenRuskaiSeiler2007}, is what makes the cancellation structure visible; accordingly, we work within this asymptotic framework and do not pursue nonasymptotic finite-$T$ bounds.
We also analyze leakage and note that the geometric nature of the Berry phase allows boundary cancellation~\cite{LidarRezakhaniHamma2009,CamposVenutiLidar2018} to suppress leakage and endpoint-induced phase errors without altering the target phase.  

Combining evolutions generated by $\pm H$ along the same loop cancels the dynamical phase, the leading $O(T^{-1})$ phase error, and all higher odd-order nonoscillatory phase errors exactly.
Richardson extrapolation then removes successive even-order nonoscillatory contributions, leaving a leading oscillatory contribution whose amplitude is controlled by endpoint data.
By further randomizing the runtime with sufficiently smooth distributions, we suppress these oscillatory terms so that, after $r$ levels of extrapolation, the bias scales as $O(T^{-(2r+2)})$ for any fixed $r$. This randomization also suppresses the zeroth-order oscillatory bias from an imperfectly prepared guiding state, so that higher estimation accuracy no longer requires more accurate initial-state preparation.

Building on this cancellation cascade, we construct a randomized Hadamard-test algorithm for Berry phase estimation that tolerates a fixed, sufficiently small initial-state preparation infidelity independent of the target accuracy.
Although the forward--reverse construction initially determines the phase only modulo $\pi$, an additional branch-lifting step recovers the full range $[0,2\pi)$ without changing the asymptotic scaling.
The resulting algorithm reduces the required coherent evolution time without increasing the standard asymptotic sample complexity.

Our results therefore establish a form of geometry-induced robustness against finite-speed adiabatic errors: because the Berry phase depends only on the loop, not on how it is traversed, the error can be cancelled order by order.
With the simple Hadamard-test circuit, which requires no quantum phase estimation, the reduced coherent evolution time, unchanged sample complexity, and tolerance to imperfect initial states position Berry phase estimation as a promising task for practical quantum computation before full fault tolerance.
More broadly, the cancellation relies only on the geometric character of the target, providing evidence that quantum-geometric quantities can be estimated with intrinsic robustness to errors in the underlying quantum process.

The paper is organized as follows.
\Cref{sec:adiabatic_error} derives the large-$T$ expansion of the phase error
and the leakage estimate (\cref{lem:phase_error,prop:leakage_error}).
\Cref{sec:forward_reverse} establishes the forward--reverse cancellation
(\cref{thm:cancellation}), \cref{sec:richardson} the further cancellation by
Richardson extrapolation (\cref{thm:richardson_m}), and \cref{sec:randomization}
the suppression of the oscillatory term by runtime randomization
(\cref{thm:randomized_richardson}).
\Cref{sec:imperfect} shows the suppression of the imperfect-guiding-state bias.
\Cref{sec:numerics} presents numerical tests, and \cref{sec:algorithm} the
Hadamard-test-based estimation algorithm.
\Cref{sec:discussion} discusses extensions and implications; the appendices
collect the technical derivations.
\Cref{tab:comparison} summarizes our main results.

\section{Adiabatic Phase Error}\label{sec:adiabatic_error}
\subsection{Adiabatic Evolution Along the Loop}

Let \(H(s)\) be a smooth family of Hamiltonians on a finite-dimensional Hilbert space for \(s\in[0,1]\) satisfying \(H(0)=H(1)\). For each \(s\in[0,1]\), let \(E_n(s)\) and \(\ket{n(s)}\) denote the \(n\)-th instantaneous eigenvalue and corresponding eigenstate of \(H(s)\), with the ground state denoted by \(\ket{\psi(s)}=\ket{0(s)}\). Assume moreover that the ground state is nondegenerate for every \(s\in[0,1]\), and choose the ground-state eigenvector in a smooth gauge satisfying \(\ket{\psi(1)}=\ket{\psi(0)}\). The minimum spectral gap along the loop is then
\[
\Delta_{\min}:=\min_{s\in[0,1]}\bigl(E_1(s)-E_0(s)\bigr).
\]
For runtime \(T\), the Schr\"odinger equation
\[
i\frac{d}{dt}\ket{\Psi(t)}
=
H\!\left(\frac{t}{T}\right)\ket{\Psi(t)},
\qquad t\in[0,T],
\]
in the rescaled variable \(s:=t/T\in[0,1]\) becomes
\[
  i\frac{d}{ds}\ket{\Psi(s)} = T H(s)\ket{\Psi(s)}.
\]
The time-evolution operator up to rescaled time \(s\) is
\[
U_T(s)=\mathcal{T}\exp\!\left(-iT\int_{0}^{s} H(\sigma)\,d\sigma\right),
\]
where \(\mathcal{T}\) denotes time ordering, so that
\[
  \ket{\Psi(s)} = U_T(s)\ket{\Psi(0)}.
\]
We define the accumulated dynamical phase by
\[
\theta_D(s) := T\int_{0}^{s} E_0(\sigma)\,d\sigma
\]
and the accumulated Berry phase by
\[
\theta_B(s) := \int_{0}^{s} i\braket{\psi(\sigma)|\dot{\psi}(\sigma)}\,d\sigma.
\]
At the end of the loop, $s=1$, \[\theta_D:=\theta_D(1),\quad \theta_B:=\theta_B(1)\] are the dynamical and Berry phases, respectively.

The adiabatic theorem states that if the initial state is the ground state \(\ket{\psi(0)}\), then the evolved state \(U_T(s)\ket{\psi(0)}\) remains close to the instantaneous ground state \(\ket{\psi(s)}\) for all \(s\in[0,1]\), and acquires the dynamical and Berry phases up to an error of order \(O(T^{-1})\):
\[
U_T(s)\ket{\psi(0)}
=
e^{-i\theta_D(s)}e^{i\theta_B(s)}\ket{\psi(s)}+O(T^{-1}).
\]
In particular, at the end of the loop,
\[
U_T(1)\ket{\psi(0)}
=
e^{-i\theta_D}e^{i\theta_B}\ket{\psi(0)}+O(T^{-1}).
\]

At finite \(T\), we quantify the deviation from ideal adiabatic transport by the scalar amplitude
\[
  z(s):=e^{i\theta_D(s)}e^{-i\theta_B(s)}
  \braket{\psi(s)|U_T(s)|\psi(0)},
\]
obtained by removing the accumulated dynamical and Berry phases from the instantaneous ground-state overlap. Its modulus gives the ground-state survival amplitude, while its argument gives the residual phase error. We therefore define
\begin{equation}\label{eq:error}
  p_{\mathrm{leak}}(s):=1-|z(s)|^2,
  \qquad
  \varphi(s):=\arg z(s),
\end{equation}
where \(p_{\mathrm{leak}}(s)\) is the probability of leaving the instantaneous ground-state subspace. We assume \(z(s)\neq 0\) for all \(s\in[0,1]\), so that \(\varphi(s)\) is well defined, and write \(\varphi:=\varphi(1)\) for the final phase error.

Writing \(z(s)=1+\delta(s)\), the adiabatic theorem implies that the true evolution converges to ideal transport as \(T\to\infty\), so \(\delta(s)=O(T^{-1})\) and, for sufficiently large \(T\), \(|\delta(s)|<1\). The phase error is therefore given by the imaginary part of the principal logarithm, which admits the convergent expansion
\begin{equation*}
  \varphi(s)
  =\Im\log\bigl(1+\delta(s)\bigr)
  =\sum_{k=1}^{\infty}\frac{(-1)^{k-1}}{k}\,
  \Im\bigl(\delta(s)^k\bigr).
\end{equation*}
Since \(\delta(s)=O(T^{-1})\), the \(k\)-th term in this expansion is \(O(T^{-k})\). The asymptotic form of the phase error is derived in the next subsection.

\subsection{Phase error expansion}
We derive the phase-error expansion for a single adiabatic evolution, making its order-by-order structure explicit. The derivation uses the adiabatic perturbation theory~\cite{RigolinOrtizPonce2008} summarized in \cref{app:APT}.

For notational simplicity, we further assume throughout this paper that the entire instantaneous spectrum is non-degenerate for all \(s\in[0,1]\). We choose the eigenvectors in a smooth gauge satisfying \(\ket{n(1)}=\ket{n(0)}\) for every \(n\). Since the initial state is the non-degenerate ground state, the argument naturally extends to the more general setting in which only the ground state remains gapped from the excited-state manifold, even if degeneracies or level crossings occur among the excited states, as noted in Ref.~\cite{RigolinOrtizPonce2008}. We nevertheless restrict ourselves to the fully non-degenerate case because treating the more general setting would substantially complicate the notation without changing the main ideas.

We denote the maximum operator norm of the Hamiltonian by
$H_{\max} := \max_s \|H(s)\|$, 
and the maximum norms of its first and second 
derivatives by
\[
  \dot{H}_{\max} := \max_s \|\dot{H}(s)\|, 
  \qquad
  \ddot{H}_{\max} := \max_s \|\ddot{H}(s)\|.
\]
Defining the adiabatic coupling \(M_{nm}(s):=\braket{n(s)|\dot m(s)}\) and the spectral gap \(\Delta_n(s):=E_n(s)-E_0(s)\), differentiation of the instantaneous eigenvalue equation gives
\[
  M_n(s):=M_{n0}(s)
  =-\frac{\braket{n(s)|\dot H(s)|\psi(s)}}{\Delta_n(s)},
  \qquad n\neq 0.
\]
This implies the uniform bound \( |M_n(s)|\leq \dot H_{\max}/\Delta_{\min} \). We denote the Berry phase of the \(n\)-th eigenstate by \(\theta_B^{(n)}(s):=\int_0^s i\braket{n(\sigma)|\dot n(\sigma)}\,d\sigma\), and the difference from the ground-state Berry phase by
\(\beta_n(s):=\theta_B^{(n)}(s)-\theta_B(s)\).

\begin{lemma}[Phase error expansion]\label{lem:phase_error}
  Under the setting of \cref{sec:adiabatic_error}, for any fixed integer
  \(K\geq2\), the final phase error \(\varphi\) admits the asymptotic expansion
  \begin{equation}\label{eq:phi_expansion}
    \varphi
    =
    \frac{\varphi_1}{T}
    +
    \sum_{k=2}^{K}
    \frac{\varphi_k+\varphi_k^{(T)}}{T^k}
    +
    O(T^{-K-1}),
  \end{equation}
  where \(\varphi_k\) are \(T\)-independent non-oscillatory loop
  contributions, while, for each \(k\geq2\),
  \begin{equation*}
    \varphi_k^{(T)}
    =
    \sum_{\nu\in I_k}
    \left[
      a_{k,\nu}\cos(\Omega_{k,\nu}T)
      +
      b_{k,\nu}\sin(\Omega_{k,\nu}T)
    \right]
  \end{equation*}
  is a finite linear combination of oscillatory boundary modes. Here
  \(I_k\) is finite, \(a_{k,\nu}\) and \(b_{k,\nu}\) are independent of
  \(T\), and each \(\Omega_{k,\nu}\) is a positive sum of the integrated
  gaps \(\omega_n:=\int_0^1\Delta_n(\sigma)\,d\sigma\). Hence
  \(\Omega_{k,\nu}\geq\Delta_{\min}>0\), so no zero-frequency component
  appears in the oscillatory sector. The leading coefficients are
  \begin{align*}
    \varphi_1
    &=
    \int_0^1\sum_{n\neq0}
    \frac{|M_n|^2}{\Delta_n}\,ds, \\
    \varphi_2^{(T)}
    &=
    \sum_{n\neq0}
    \frac{|M_n(0)|\,|M_n(1)|}{\Delta_n(0)^2}
    \sin\bigl(\gamma_n-\omega_nT\bigr)
  \end{align*}
  where
  \(\gamma_n:=\beta_n(1)-\arg M_n(1)+\arg M_n(0)\). Moreover,
  \[
    \varphi_1\leq\frac{\dot H_{\max}^2}{\Delta_{\min}^3},
    \qquad
    \bigl|\varphi_2^{(T)}\bigr|
    \leq
    \frac{\|\dot H(0)\|\,\|\dot H(1)\|}{\Delta(0)^4},
  \]
  where \(\Delta(0):=\Delta_1(0)=E_1(0)-E_0(0)\).
  \end{lemma}

  The proof of \cref{lem:phase_error}, including the explicit form of \(\varphi_2\), is given in \cref{app:phase_exact}; a bound on \(\varphi_2\) is given in \cref{app:bound}, while the bounds on \(\varphi_1\) and \(\varphi_2^{(T)}\) follow from completeness of the eigenbasis. Unlike the \(\Delta_{\min}^{-3}\) bound on \(\varphi_1\), the bound on \(\varphi_2^{(T)}\) depends only on endpoint data.

If, in addition, \(\dot H(0)=\dot H(1)\), then \(M_n(1)=M_n(0)\) and the bound reduces to \(|\varphi_2^{(T)}|\le\|\dot H(0)\|^2/\Delta(0)^4\).

The phase error separates into a non-oscillatory bulk contribution and an oscillatory boundary contribution. The coefficients \(\varphi_k\) are \(T\)-independent and encode information accumulated along the loop, whereas \(\varphi_k^{(T)}\) depend on \(T\) through dynamical phases and arise from endpoint terms that survive the closed-loop condition.

\subsection{Leakage error expansion}
\label{sec:leakage}

In addition to the phase error within the adiabatically transported
ground-state sector, it is important to quantify nonadiabatic
transitions out of that sector. We therefore consider the leakage
probability \(p_{\mathrm{leak}}(s)\) defined in \cref{eq:error}.
It measures the loss of weight from the transported ground-state
sector, whereas the phase error \(\varphi(s)=\arg z(s)\) measures the
phase of the same survival amplitude. Unlike the phase error, the leading-order leakage probability follows
directly from the first-order excited-state amplitudes derived in
\cref{app:APT}.

The \(O(T^{-2})\) leakage scaling and its leading endpoint-interference structure are established results~\cite{JansenRuskaiSeiler2007,CohenOh2025}, which we recover here using adiabatic perturbation theory~\cite{RigolinOrtizPonce2008} in the same notation as our phase-error analysis.

\begin{proposition}[Leakage error expansion]\label{prop:leakage_error}
  Under the setting of \cref{sec:adiabatic_error}, the leakage probability satisfies
  \begin{align}
    p_{\mathrm{leak}}(1)
    &=
    \frac{1}{T^2}\sum_{n\neq 0}
    \left|
      \frac{M_{n}(1)}{\Delta_{n}(1)}
      -
      e^{-iT\omega_{n}}e^{i\beta_{n}(1)}
      \frac{M_{n}(0)}{\Delta_{n}(0)}
    \right|^2 \notag\\
    &\quad + O(T^{-3}).
    \label{eq:pleak_expansion}
  \end{align}
  In particular,
  \begin{align*}
    p_{\mathrm{leak}}(1)
    &\le
    \frac{C_{\mathrm{leak}}}{T^2}
    + O(T^{-3}), \\
    C_{\mathrm{leak}}
    &=
    \frac{2\bigl(
      \|\dot H(0)\|^2+\|\dot H(1)\|^2
    \bigr)}
    {\Delta(0)^4},
  \end{align*}
  where \(\Delta(0)\) is defined in \cref{lem:phase_error}.
\end{proposition}
The proof of \cref{prop:leakage_error} is given at
the end of \cref{app:APT}.

The two terms in \cref{eq:pleak_expansion} describe contributions from the endpoints of the loop. Setting the first \(k\) derivatives of \(H\) to zero at both endpoints gives \(p_{\mathrm{leak}}(1)=O(T^{-2(k+1)})\). This is the standard boundary-cancellation technique~\cite{LidarRezakhaniHamma2009,CamposVenutiLidar2018}, which also suppresses endpoint-induced phase errors. The geometric nature of the Berry phase makes it invariant under schedule reparametrization, allowing boundary cancellation to improve error scaling without changing the target phase. The dynamical phase, by contrast, generally changes with the schedule. We do not develop boundary cancellation further here. In the Hadamard-test implementation, the leakage-induced reduction in overlap magnitude increases statistical noise only by a factor \(1+O(T^{-2})\), without changing the asymptotic scaling (\cref{sec:hadamard}). We instead focus on canceling the bulk phase errors and suppressing oscillatory errors through runtime randomization, which also suppresses the preparation bias that boundary cancellation does not remove (\cref{sec:imperfect}).

\section{Adiabatic Error Cancellation}\label{sec:cancellation}

\subsection{Odd-order non-oscillatory cancellation}
\label{sec:forward_reverse}
To isolate the Berry phase while canceling the dynamical
phase, we introduce the sign-reverse evolution generated by
$-H(s)$:
\begin{equation*}
  \hat{U}_T(s) := \mathcal{T}\exp\!\left(
  +iT\int_0^s H(\sigma)\,d\sigma\right).
\end{equation*}
Since \(-H(s)\) has eigenvalues \(-E_n(s)\) with the same
eigenstates \(|n(s)\rangle\), the same instantaneous eigenstate
\(|\psi(s)\rangle\) is followed adiabatically. Hence the adiabatic phase for the reverse evolution is obtained by
changing the sign of the dynamical phase while leaving the Berry phase
unchanged:
\[
  \hat U_T(s)|\psi(0)\rangle
  =
  e^{+i\theta_D(s)}e^{i\theta_B(s)}|\psi(s)\rangle
  +O(T^{-1}).
\]
We therefore define the reverse phase error by
\begin{equation*}
  \hat\varphi(s)
  :=
  \arg\left( e^{-i\theta_D(s)}e^{-i\theta_B(s)}
  \langle\psi(s)|\hat U_T(s)|\psi(0)\rangle \right).
\end{equation*}
Since \(-H(s)\) satisfies the same loop conditions as \(H(s)\),
\cref{lem:phase_error} applies with
\(\Delta_n\to-\Delta_n\) and \(\omega_n\to-\omega_n\), while
\(M_n\), \(\beta_n\), and \(\gamma_n\) remain unchanged. Equivalently,
the reverse evolution is formally obtained by \(T\to -T\), giving, for
any fixed \(K\geq2\),
\begin{equation}\label{eq:phihat_expansion}
  \hat{\varphi}
  =
  -\frac{\varphi_1}{T}
  +
  \sum_{k=2}^{K}
  \frac{(-1)^k\bigl(
    \varphi_k+\varphi_k^{(-T)}
  \bigr)}{T^k}
  +
  O(T^{-K-1}).
\end{equation}
Accordingly, all odd-order non-oscillatory terms change sign, whereas
the even-order ones remain unchanged, and each oscillatory contribution
\(\varphi_k^{(T)}\) is replaced by
\((-1)^k\varphi_k^{(-T)}\).

When the phases of the overlap amplitudes
\(\langle\psi(0)|U_T(1)|\psi(0)\rangle\) and
\(\langle\psi(0)|\hat U_T(1)|\psi(0)\rangle\) are obtained
independently, they are
\(-\theta_D+\theta_B+\varphi\) and
\(+\theta_D+\theta_B+\hat\varphi\), respectively. Their average cancels
the dynamical phase exactly and defines the Berry-phase estimator
\begin{equation}\label{eq:theta_est}
  \tilde{\theta}_B(T)
  :=
  \theta_B+\frac{\varphi+\hat\varphi}{2}.
\end{equation}
Since the sum of the two overlap phases is known only modulo \(2\pi\),
division by two determines \(\tilde{\theta}_B(T)\) only modulo \(\pi\).
This ambiguity is resolved by the branch-resolution procedure in
\cref{sec:algorithm}.

The following theorem shows that all odd-order non-oscillatory errors
cancel. Thus \(\tilde{\theta}_B(T)-\theta_B=O(T^{-2})\), improving on the
\(O(T^{-1})\) error of a single evolution.

\begin{theorem}[Odd-order non-oscillatory adiabatic error cancellation]%
  \label{thm:cancellation}
  Under the setting of \cref{sec:adiabatic_error}, the forward--reverse
  symmetrized estimator cancels the non-oscillatory terms
  \[
    \frac{\varphi_1}{T},
    \frac{\varphi_3}{T^3},
    \frac{\varphi_5}{T^5},
    \ldots
  \]
  exactly. More precisely, for any fixed integer $K\geq 2$, the estimated
  Berry phase satisfies
\begin{equation*}
  \tilde{\theta}_B(T)-\theta_B
  =
  \sum_{k=1}^{\lfloor K/2\rfloor}\frac{\Phi_{2k}}{T^{2k}}
  +
  \sum_{k=2}^{K}\frac{\Phi_k^{(T)}}{T^k}
  +
  O(T^{-K-1}),
\end{equation*}
where
$\Phi_{2k}=\varphi_{2k}$ are non-oscillatory coefficients, while
$\Phi_k^{(T)}
=\bigl(\varphi_k^{(T)}+(-1)^k\varphi_k^{(-T)}\bigr)/2$
are oscillatory coefficients depending on $T$. In particular,
  \begin{align}\label{eq:Phi2T}
    \Phi_2^{(T)}
    &= \sum_{n\neq 0} B_n \cos(\omega_{n}T),
  \end{align}
  where
  \begin{equation}\label{eq:Bn_def}
    B_n
    :=
    \frac{|M_n(0)|\,|M_n(1)|}
         {\Delta_n(0)^2}
         \sin\gamma_n.
  \end{equation}
\end{theorem}
\begin{proof}
  Substituting the phase-error expansion of \cref{lem:phase_error}
  and its reverse-evolution counterpart \cref{eq:phihat_expansion} into
  the symmetrized estimator \cref{eq:theta_est}, the non-oscillatory coefficient \(\varphi_k\) is multiplied
  by \(\bigl(1+(-1)^k\bigr)/2\). Hence all odd-order non-oscillatory terms
  vanish, while the even-order terms remain. The oscillatory terms combine
  as stated.
  \end{proof}

The remaining error consists of non-oscillatory even-order terms and oscillatory terms. In the next subsection, we apply Richardson extrapolation to remove the leading non-oscillatory \(O(T^{-2})\) term.

\subsection{Even-order non-oscillatory cancellation}
\label{sec:richardson}

The residual error established in \cref{thm:cancellation} has the
finite-order asymptotic form
\[
  \tilde{\theta}_B(T)-\theta_B
  =
  \sum_{k=1}^{\lfloor K/2\rfloor}
  \frac{\Phi_{2k}}{T^{2k}}
  +
  \sum_{k=2}^{K}
  \frac{\Phi_k^{(T)}}{T^k}
  +
  O(T^{-K-1}),
\]
where \(\Phi_{2k}\) are \(T\)-independent non-oscillatory terms and
\(\Phi_k^{(T)}=O(1)\) are oscillatory boundary terms. We now use
Richardson extrapolation to remove the even-order non-oscillatory
contributions.

For a fixed integer \(r\ge1\), define
\[
  \tilde{\theta}_{B,R}^{(r)}(T)
  :=
  \sum_{\ell=0}^{r}
  w_{r,\ell}\tilde{\theta}_B(\alpha^\ell T),
  \qquad
  \alpha>1.
\]
The weights are chosen so that
\[
  \sum_{\ell=0}^{r}w_{r,\ell}=1,
  \quad
  \sum_{\ell=0}^{r}w_{r,\ell}\alpha^{-2j\ell}=0,
  \quad
  j=1,\ldots,r.
\]
\begin{theorem}[Even-order non-oscillatory adiabatic error cancellation]
  \label{thm:richardson_m}
  Under the setting of \cref{sec:adiabatic_error}, the \(r\)-fold
  Richardson extrapolant cancels the non-oscillatory terms
  \[
    \frac{\Phi_2}{T^2},
    \frac{\Phi_4}{T^4},
    \ldots,
    \frac{\Phi_{2r}}{T^{2r}}
  \]
  exactly. For any fixed integer \(K\geq 2r+2\),
  \[
    \tilde{\theta}_{B,R}^{(r)}(T)-\theta_B
    =
    \sum_{j=r+1}^{\lfloor K/2\rfloor}
    \frac{C_{j,r}\Phi_{2j}}{T^{2j}}
    +
    \sum_{k=2}^{K}
    \frac{\Phi_{k,r}^{(T)}}{T^k}
    +
    O(T^{-K-1}),
  \]
  where
  \[
    C_{j,r}
    :=
    \sum_{\ell=0}^{r}
    w_{r,\ell}\alpha^{-2j\ell},\quad
    \Phi_{k,r}^{(T)}
    :=
    \sum_{\ell=0}^{r}
    w_{r,\ell}\alpha^{-k\ell}
    \Phi_k^{(\alpha^\ell T)}.
  \]
  Here \(\Phi_{k,r}^{(T)}\) is oscillatory and \(O(1)\), while the first remaining \(T\)-independent contribution occurs at order
  \(T^{-(2r+2)}\).
\end{theorem}

The proof of \cref{thm:richardson_m} is given in \cref{app:richardson_details}.

Richardson extrapolation cancels the first \(r\)
non-oscillatory even-order terms but leaves a bounded oscillatory
boundary contribution. The leading oscillatory coefficient obeys
\[
  |\Phi_{2,r}^{(T)}|
  \leq
  L_r(\alpha)
  \frac{\|\dot H(0)\|\,\|\dot H(1)\|}
       {\Delta(0)^4},
\]
where
\[
  L_r(\alpha)
  :=
  \sum_{\ell=0}^{r}
  |w_{r,\ell}|\,\alpha^{-2\ell}.
\]
Throughout this work, the asymptotic estimates are
taken at fixed \(r\). The explicit form of \(\Phi_{k,r}^{(T)}\) and the bound on \(L_r(\alpha)\) are given in \cref{app:richardson_details}. Further properties of the factor $L_r(\alpha)$---its monotone decrease in $r$ for $\alpha\ge\sqrt3$ and the higher-sector factors $L_{r,k}(\alpha)$---are collected in \cref{app:richardson_factor}.

For \(r=1\), the extrapolant reduces to
\[
  \tilde{\theta}_{B,R}^{(1)}(T)
  =
  \frac{\alpha^2\tilde{\theta}_B(\alpha T)-\tilde{\theta}_B(T)}
       {\alpha^2-1}.
\]
Since
\[
  \Phi_2^{(T)}
  =
  \sum_{n\neq 0}B_n\cos(\omega_nT),
\]
the leading oscillatory coefficient becomes
\[
  \Phi_{2,1}^{(T)}
  =
  \frac{1}{\alpha^2-1}
  \sum_{n\neq 0}B_n
  \left[
    \cos(\alpha\omega_nT)-\cos(\omega_nT)
  \right].
\]
Thus,
\[
  \tilde{\theta}_{B,R}^{(1)}(T)-\theta_B
  =
  \frac{\Phi_{2,1}^{(T)}}{T^2}
  +
  \frac{\Phi_{3,1}^{(T)}}{T^3}
  +
  \frac{\Phi_{4,1}^{(T)}-\alpha^{-2}\Phi_4}{T^4}
  +
  O(T^{-5}),
\]
where \(\Phi_{3,1}^{(T)}\) and \(\Phi_{4,1}^{(T)}\) are oscillatory and
\(O(1)\).

\subsection{Oscillatory-error suppression}
\label{sec:randomization}

Like boundary cancellation discussed in \cref{sec:leakage}, runtime randomization exploits the geometric nature of the Berry phase: averaging estimates at different runtimes suppresses oscillatory errors without changing \(\theta_B\). We develop this approach further because it also suppresses initial-state preparation bias, as shown in \cref{sec:imperfect}. We combine it with Richardson extrapolation as follows. Let \(\mu\) be a probability distribution supported on
\([1-\lambda,1+\lambda]\), where \(0<\lambda<1\), with density \(\rho\), and write
\(c_\mu:=2\lambda\,\|\rho\|_\infty\) (so \(c_\mu=1\) for the uniform distribution; \(c_\mu\)
is finite whenever some \(\chi_{\mu,k}\) is integrable, in particular under the decay
hypothesis of \cref{thm:randomized_richardson} with \(M>1\)). Partition the support into
\(N\) intervals \(I_1,\dots,I_N\) of equal width \(2\lambda/N\), let
\(p_j:=\mu(I_j)\), and draw the runtime multipliers by stratified sampling,
\(X_j\sim\mu(\,\cdot\mid I_j)\), independently sampled from \(\mu\) conditioned on \(I_j\), for \(j=1,\dots,N\); stratification
rather than independent sampling is used to reduce the variance of the estimator.
Set
\[
  T_j=TX_j.
\]
For each randomized runtime, we form the \(r\)-fold Richardson estimator
\[
  \tilde{\theta}_{B,R}^{(r)}(T_j)
  :=
  \sum_{\ell=0}^{r}
  w_{r,\ell}\tilde{\theta}_B(\alpha^\ell T_j),
\]
and average over the \(N\) strata with the weights \(p_j\):
\[
  \hat{\theta}_{B,R}^{(N,r)}(T)
  :=
  \sum_{j=1}^{N}p_j\,
  \tilde{\theta}_{B,R}^{(r)}(T_j).
\]
Since \(\sum_jp_j\,\mathbb E[f(X_j)]=\int f\,d\mu\) for every bounded \(f\), the
expectation of \(\hat{\theta}_{B,R}^{(N,r)}(T)\) is the same as for \(N\) independent
samples from \(\mu\); stratification changes only the fluctuation.

Richardson extrapolation removes the non-oscillatory terms
\(T^{-2},T^{-4},\ldots,T^{-2r}\). The leading remaining contribution is
oscillatory. Runtime randomization suppresses this contribution in
expectation through the decay of the weighted characteristic functions
\[
  \chi_{\mu,k}(\xi)
  :=
  \int_{1-\lambda}^{1+\lambda}
  x^{-k}e^{i\xi x}\,d\mu(x).
\]
\begin{theorem}[Runtime-randomized Richardson estimator]
  \label{thm:randomized_richardson}
  Under the setting of \cref{sec:adiabatic_error}, suppose that, for some
  \(M>0\) with \(M\ge2r-1\), the weighted characteristic functions satisfy
  \[
    |\chi_{\mu,k}(\xi)|
    \leq
    C_{\mu,K}(1+|\xi|)^{-M},
    \qquad
    2\leq k\leq K.
  \]
  Then, for any fixed integer \(K\geq 2r+2\),
  \begin{align}
    \mathbb E\left[
      \hat{\theta}_{B,R}^{(N,r)}(T)
    \right]-\theta_B
    &=
    O\left(
      \frac{
        C_{\mu,K}L_r(\alpha)\,
        \|\dot H(0)\|\,\|\dot H(1)\|
      }{
        \Delta(0)^4
        \Delta_{\min}^{M}
        T^{M+2}
      }
    \right)
    \nonumber\\
    &\quad
    +
    O(T^{-(2r+2)}).
    \label{eq:randomized_richardson_theorem}
  \end{align}
  Moreover, with \(L_{r,k}(\alpha)\) as in \cref{app:richardson_factor}, the finite-sample fluctuation satisfies
  \begin{align*}
    \sqrt{
      \operatorname{Var}\left(
        \hat{\theta}_{B,R}^{(N,r)}(T)
      \right)
    }
    &=
    O\Bigl(
      \frac{
        \sqrt{c_\mu}\,\lambda\,
        \|\dot H(0)\|\,\|\dot H(1)\|
      }{
        \Delta(0)^4\,N^{3/2}
      }\\
    &\quad\times
      \Bigl[
        \frac{H_{\max}L_{r,1}(\alpha)}{(1-\lambda)^{2}\,T}
        +
        \frac{L_r(\alpha)}{(1-\lambda)^{3}\,T^{2}}
      \Bigr]
    \Bigr)\\
    &\quad{}+O\!\left(\frac{1}{T^2N^{3/2}}\right).
  \end{align*}
\end{theorem}
The proof of \cref{thm:randomized_richardson} is given in \cref{app:runtime_randomization}.

The first term in \cref{eq:randomized_richardson_theorem} is the randomized oscillatory bias, while the second is the remaining non-oscillatory Richardson error: runtime randomization controls the oscillatory sector, whereas Richardson extrapolation controls the non-oscillatory one. To make the gap dependence of the leading error explicit, we can choose \(M=2r-1\), so that the oscillatory contribution is \(O(T^{-2r-1})\) while the non-oscillatory remainder is \(O(T^{-(2r+2)})\).

The variance statement shows that randomization suppresses the oscillatory
bias but does not, by itself, suppress the finite-sample fluctuation.  The
remaining randomized fluctuation is of order \(T^{-1}N^{-3/2}\), governed
by the leading \(k=2\) oscillatory sector. The factor \(N^{-3/2}\), rather than the \(N^{-1/2}\) of independent sampling,
comes from stratification: within a stratum the fluctuation is bounded by the width
times the slope of the oscillation, which is why \(H_{\max}T\) appears.

Three representative choices of \(\mu\) are worth noting. It is enough to establish the stated decay for large \(|\xi|\), since
\(\chi_{\mu,k}\) is bounded on compact \(\xi\)-intervals and the constant
\(C_{\mu,K}\) can be enlarged accordingly.

\paragraph{Example 1: Uniform randomization.}
For the uniform distribution on \([1-\lambda,1+\lambda]\), \(c_\mu=1\). Integration
by parts gives
\begin{align*}
  |\chi_{\mu,k}(\xi)|
  &\leq
  \frac{C_{\mu,k}}{|\xi|}
  +
  O(|\xi|^{-2}),
  \\
  C_{\mu,k}
  &:=
  \frac{
    (1-\lambda)^{-k}
    +(1+\lambda)^{-k}
  }{2\lambda}.
\end{align*}
Thus uniform randomization corresponds to \(M=1\), and the leading
oscillatory bias scales as \(O(T^{-3})\).

\paragraph{Example 2: Triangular randomization.}
For the triangular distribution
\[
  \frac{d\mu}{dx}
  =
  \frac{1}{\lambda}
  \left(
    1-\frac{|x-1|}{\lambda}
  \right),
  \qquad
  x\in[1-\lambda,1+\lambda],
\]
the density vanishes at the endpoints. Here, \(c_\mu=2\). Piecewise integration by parts
gives
\begin{align*}
  |\chi_{\mu,k}(\xi)|
  &\leq
  \frac{C_{\mu,k}^{\triangle}}{|\xi|^2}
  +
  O(|\xi|^{-3}),
  \\
  C_{\mu,k}^{\triangle}
  &:=
  \frac{
    (1-\lambda)^{-k}
    +2
    +(1+\lambda)^{-k}
  }{\lambda^2}.
\end{align*}
Thus triangular randomization corresponds to \(M=2\), and the leading
oscillatory bias scales as \(O(T^{-4})\).

\paragraph{Example 3: Smooth bump randomization.}
We use the standard smooth compactly supported bump function
\cite{HormanderAnalysisI}. For the normalized density
\[
  \frac{d\mu}{dx}
  =
  \frac{1}{Z_\lambda}
  \exp\left(
    -\frac{1}{
      1-\left(\frac{x-1}{\lambda}\right)^2
    }
  \right),
  \quad
  x\in(1-\lambda,1+\lambda),
\]
with \(d\mu/dx=0\) outside this interval, the density and all of its
derivatives vanish at the endpoints. Here, \(c_\mu=2\lambda/(eZ_\lambda)\approx1.657\), independent of \(\lambda\). Consequently, for every fixed
integer \(M\geq1\),
\[
  |\chi_{\mu,k}(\xi)|
  \leq
  \frac{C_{\mu,k}^{(M)}}{|\xi|^M},
\]
where \(C_{\mu,k}^{(M)}\) depends on \(M\). Thus smooth bump
randomization suppresses the oscillatory bias to arbitrary fixed
inverse-polynomial order, up to the remaining non-oscillatory
Richardson error \(O(T^{-(2r+2)})\). Sampling from \(\mu\), including the stratified conditional draws, is a one-dimensional classical task and can be done by numerical inversion of the distribution function at negligible cost.

For a fixed runtime, increasing \(M\) does not necessarily give a tighter
numerical bound, because the constant \(C_{\mu,k}^{(M)}\) also depends on
\(M\).  In the asymptotic regime \(\Delta_{\min}T\gg1\), however, smoother
runtime distributions give higher-order suppression of the oscillatory
sector.

  \subsection{Error Cancellation for an Imperfect Initial State}\label{sec:imperfect}
  The analysis so far has assumed that the protocol is supplied with the exact
ground state \(\ket{\psi(0)}\) of \(H(0)\). In practice the guiding state is prepared to
finite accuracy: instead of it the protocol is supplied with a normalized
\begin{equation*}
  \ket{\psi_{\mathrm{prep}}}=\sum_{n\ge0}c_n\ket{n(0)},
  \qquad
  p_{\mathrm{prep}}:=1-|c_0|^2 ,
\end{equation*}
with \(c_0\neq0\), so that \(p_{\mathrm{prep}}\) is the preparation infidelity,
in parallel with the leakage probability \(p_{\mathrm{leak}}\) of \cref{eq:error}.
In analogy with \(z(s)\) we set
\begin{equation*}
  \tilde z:=e^{i\theta_D}e^{-i\theta_B}
  \braket{\psi_{\mathrm{prep}}|U_T(1)|\psi_{\mathrm{prep}}},
\end{equation*}
and write \(\varphi^{\mathrm{prep}}:=\arg\tilde z-\varphi\), so that
\(\arg\tilde z=\varphi+\varphi^{\mathrm{prep}}\) with \(\varphi\) exactly the quantity
expanded in \cref{lem:phase_error}: the whole effect of the imperfect preparation is
carried by \(\varphi^{\mathrm{prep}}\).

\begin{lemma}[Preparation bias]\label{lem:prep_bias}
  Let \(r_n:=|c_n|^2/|c_0|^2\) and assume \(p_{\mathrm{prep}}<1/2\), so that
  \(\sum_{n\ge1}r_n<1\). For any fixed integer \(K\ge2\),
  \begin{equation}
    \varphi^{\mathrm{prep}}
    =\varphi_0^{\mathrm{prep}(T)}
    +\sum_{k=1}^{K}
    \frac{\varphi_k^{\mathrm{prep}}+\varphi_k^{\mathrm{prep}(T)}}{T^{k}}
    +O(T^{-K-1}),
    \label{eq:prep_expansion}
  \end{equation}
  where the \(\varphi_k^{\mathrm{prep}}\) are \(T\)-independent non-oscillatory
  coefficients, each \(\varphi_k^{\mathrm{prep}(T)}\) is an absolutely convergent
  combination of \(\cos(\Omega T)\) and \(\sin(\Omega T)\) with \(T\)-independent
  coefficients and frequencies \(\Omega\ge\Delta_{\min}\), and
  \(\varphi_k^{\mathrm{prep}},\varphi_k^{\mathrm{prep}(T)}=O(p_{\mathrm{prep}}^{1/2})\)
  for \(k\ge1\), with
  \begin{align}
    \varphi_0^{\mathrm{prep}(T)}
    &=\Im\log\Bigl(1+\sum_{n\ge1}r_ne^{i(\beta_n-\omega_nT)}\Bigr) \notag\\
    &=\sum_{\nu}\bigl[a_\nu\cos(\Omega_\nu T)
      +b_\nu\sin(\Omega_\nu T)\bigr].
    \label{eq:prep_k0}
  \end{align}
  The latter sum is countable, with frequencies \(\Omega_\nu\) that are sums of one or more of the \(\omega_n\), hence \(\Omega_\nu\ge\omega_1\), and
  \(\sum_\nu(|a_\nu|+|b_\nu|)=O(p_{\mathrm{prep}})\).
\end{lemma}

The proof is given in \cref{app:prep_expansion}; the assumption \(p_{\mathrm{prep}}<1/2\)
is exactly what makes the expansion of the logarithm converge. \Cref{eq:prep_expansion} has the same structure as \cref{eq:phi_expansion}, except that it begins at \(k=0\). Consequently, \cref{thm:cancellation,thm:richardson_m} remove the non-oscillatory terms \(\varphi_k^{\mathrm{prep}}/T^k\) as before, but neither removes the zeroth-order oscillatory term \(\varphi_0^{\mathrm{prep}(T)}\), which carries no power of \(1/T\). The proof of \cref{thm:randomized_richardson}, by contrast, relies only on oscillation at frequencies bounded away from zero and therefore carries over without modification.

\begin{theorem}[Suppression of the preparation bias]\label{thm:prep_supp}
  Suppose that, for some \(M>0\),
  \begin{equation*}
    |\chi_{\mu,k}(\xi)|\le C_{\mu,K}(1+|\xi|)^{-M},
    \qquad 0\le k\le K.
  \end{equation*}
  Let \(\widehat\varphi^{\mathrm{prep}}(t)\) denote the reverse-evolution preparation error, defined analogously to \(\varphi^{\mathrm{prep}}(t)\), and define the preparation contribution to the full estimator by
  \begin{align*}
    \hat\varphi^{(N,r)}_{\mathrm{prep}}(T)
    &:=\sum_{j=1}^{N}p_j\sum_{\ell=0}^{r}w_{r,\ell}
    \frac{\varphi^{\mathrm{prep}}(t_{j,\ell})
    +\widehat\varphi^{\mathrm{prep}}(t_{j,\ell})}{2},\\
    t_{j,\ell}&:=\alpha^\ell TX_j.
  \end{align*}
  Then, for any fixed integer \(K\ge2r+2\), \(\hat\varphi^{(N,r)}_{\mathrm{prep}}(T)\) satisfies
  \begin{align}
    \mathbb E\bigl[\hat\varphi^{(N,r)}_{\mathrm{prep}}(T)\bigr]
    &=O\!\left(
      \frac{p_{\mathrm{prep}}\,C_{\mu,K}\,L_{r,0}(\alpha)}
           {(\Delta_{\min}T)^{M}}
    \right) \notag\\
    &\quad{}+O\!\left(
      \frac{p_{\mathrm{prep}}^{1/2}}{\Delta_{\min}^{M}T^{M+1}}
    \right) \notag\\
    &\quad+O\bigl(p_{\mathrm{prep}}^{1/2}T^{-(2r+2)}\bigr),
    \label{eq:prep_supp}
  \end{align}
  and the finite-sample fluctuation satisfies
  \begin{equation}
    \begin{aligned}
    \sqrt{\operatorname{Var}\bigl(\hat\varphi^{(N,r)}_{\mathrm{prep}}(T)\bigr)}
    &=O\!\left(
      \frac{\sqrt{c_\mu}\,\lambda\,\alpha^{r}L_{r,0}(\alpha)\,H_{\max}T\,p_{\mathrm{prep}}}
           {N^{3/2}}
    \right)\\
    &\quad{}+O\!\left(\frac{p_{\mathrm{prep}}^{1/2}}{N^{3/2}}\right),
    \end{aligned}
    \label{eq:prep_var}
  \end{equation}
\end{theorem}

The proof is given in \cref{app:prep_randomization}. Requiring the first term of
\cref{eq:prep_supp} to be at most the target accuracy \(\varepsilon_B\) of the Berry phase gives
\begin{equation}
  p_{\mathrm{prep}}\lesssim
  \frac{(\Delta_{\min}T)^{M}}{C_{\mu,K}\,L_{r,0}(\alpha)}\,\varepsilon_B ,
  \label{eq:prep_req}
\end{equation}
replacing the requirement \(p_{\mathrm{prep}}\lesssim\varepsilon_B\) that holds
without randomization, the admissible infidelity being larger by the factor
\((\Delta_{\min}T)^{M}\gg1\). Since \(M\) is fixed by the runtime density
alone---arbitrarily large for a compactly supported \(C^\infty\) density, at no cost in
circuit depth---whereas each Richardson level costs a node of depth \(\alpha^\ell T\),
it is natural to raise \(M\) rather than \(r\). Take \(M=2r+2\), the order at which the
preparation bias reaches the same \(O(T^{-(2r+2)})\) floor as the cascade itself. The
cascade runtime is \(T=O(\varepsilon_B^{-1/(2r+2)})\), so \(T^{2r+2}\varepsilon_B\) is a
constant and the right-hand side of \cref{eq:prep_req} no longer depends on
\(\varepsilon_B\): the guiding state need not be prepared more accurately as
\(\varepsilon_B\) is reduced, only \(T\) grows. This matters because the cost of
adiabatic state preparation grows rapidly as the target infidelity is reduced. Increasing \(M\) improves the asymptotic suppression, but the larger constant \(C_{\mu,K}\) can limit the benefit at finite runtime.

The fluctuation \cref{eq:prep_var} grows linearly in \(T\), unlike that of
\cref{thm:randomized_richardson}: the preparation oscillation carries no power of
\(1/T\), so its slope in the runtime multiplier is of order \(\omega_n\alpha^\ell T\).
It decays as \(N^{-3/2}\), however, and with the \(N=O(\varepsilon_B^{-1})\) stratified
runtimes used in \cref{sec:hadamard} it is \(O(p_{\mathrm{prep}}\,H_{\max}T\,
\varepsilon_B^{3/2})\) up to the constants of \cref{eq:prep_var}, which at
\(T=O(\varepsilon_B^{-1/(2r+2)})\) is \(o(\varepsilon_B)\) for any fixed
\(p_{\mathrm{prep}}\).

Finally, this error is beyond the reach of any modification of the schedule. As noted in
\cref{sec:leakage}, boundary cancellation suppresses the oscillatory terms of
\cref{lem:phase_error} because their amplitudes are proportional to the endpoint
couplings \(M_n(0)\) and \(M_n(1)\). The amplitude of \(\varphi_0^{\mathrm{prep}(T)}\),
by contrast, is \(r_n\): fixed before the loop begins and containing no derivative of
\(H\). A reparametrization rescales its frequencies but leaves the amplitude, and its
lack of decay in \(T\), untouched.

\section{Numerical Simulations}\label{sec:numerics}
\begin{figure*}[t]
  \centering
  \includegraphics[width=\textwidth]{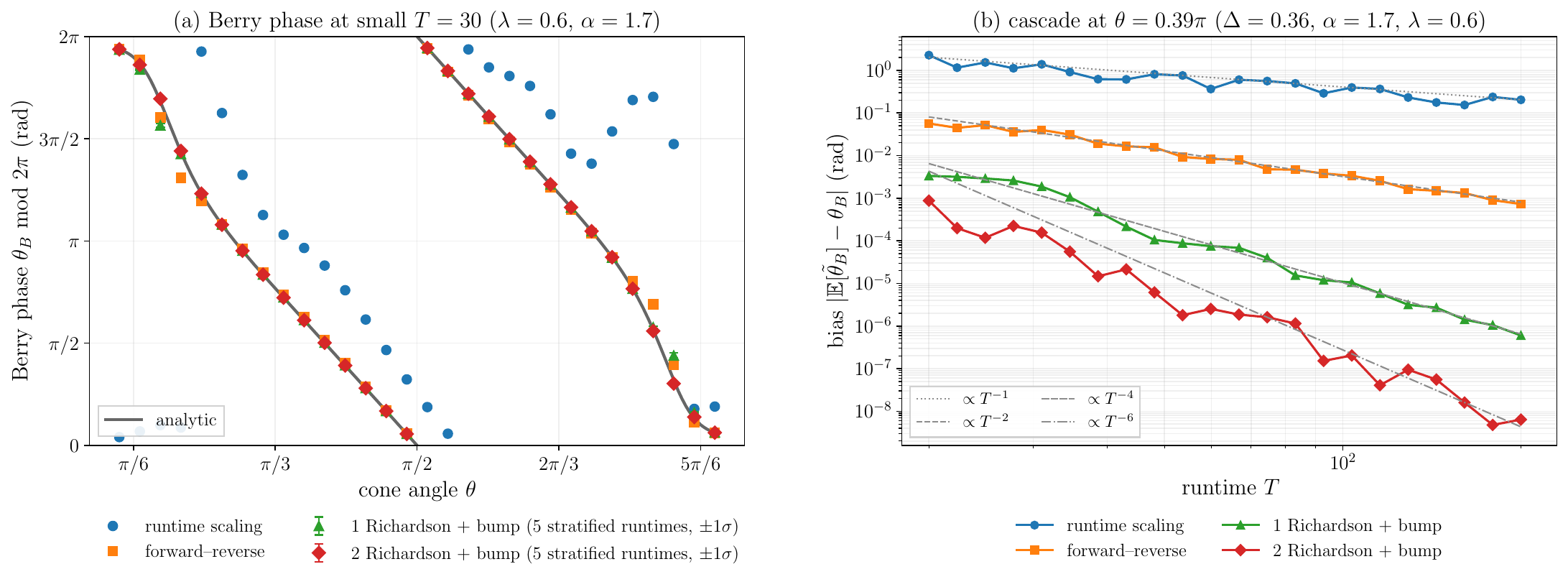}
  \caption{Entangled spiral Heisenberg chain
  (\(L=4\) sites, \(J=1\), \(B_0=1\); \(\lambda=0.6\), \(\alpha=1.7\)).
  \textbf{(a)} Berry phase \(\theta_B\) in the canonical window
  \([0,2\pi)\) versus cone angle at \(T=30\).  The exact analytic value is
  shown as the wrapped line.  The runtime-scaling reconstruction is plotted
  as raw values modulo \(2\pi\), while the forward--reverse and randomized Richardson estimates track
  the correct branch. The randomized Richardson estimates use five stratified
  runtimes (\cref{sec:randomization}); each point is the mean of ten independent
  such estimates and the error bar is the \(\pm1\sigma\) spread of a single
  estimate.
  \textbf{(b)} Error-cancellation cascade at \(\theta=0.39\pi\),
  \(\Delta=0.36\): runtime scaling, forward--reverse, one Richardson level
  + randomization, and two Richardson levels + randomization exhibit
  \(O(T^{-1})\), \(O(T^{-2})\), \(O(T^{-4})\), and \(O(T^{-6})\) behavior.
  Grey lines are reference slopes. Runtime randomization uses the smooth bump density of \cref{sec:randomization}.}
  \label{fig:spiral}
\end{figure*}

\begin{figure*}[t]
  \centering
  \includegraphics[width=\textwidth]{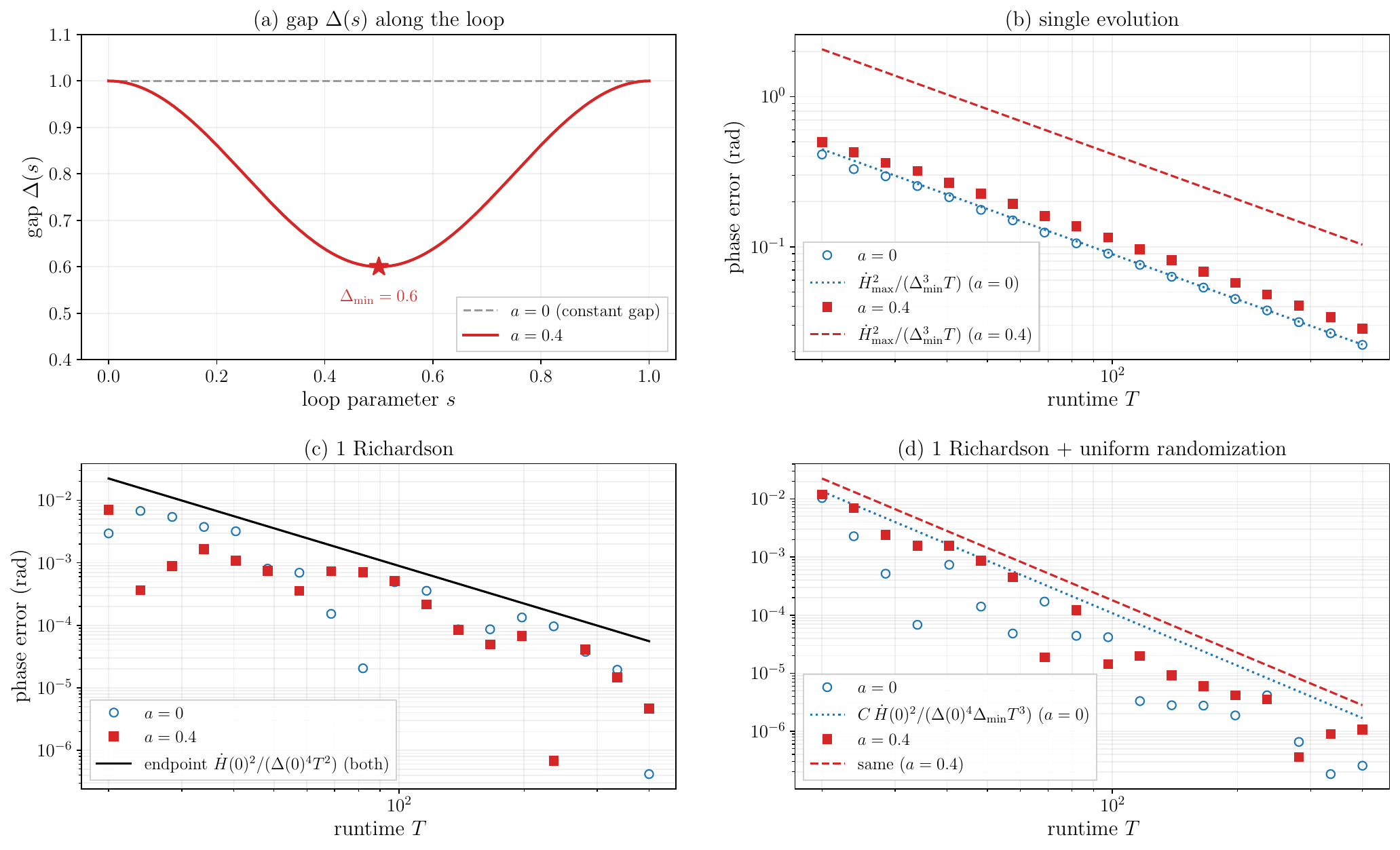}
  \caption{Numerics against the analytic bounds, and the controlling gap for each
  estimator, on a non-isospectral spin-\(1/2\) loop with field modulation
  \(|\mathbf B(s)|=|\mathbf B|(1-a\sin^2(\pi s))\), \(a=0.4\).
  \textbf{(a)} Instantaneous gap \(\Delta(s)\): constant loop
  (\(a=0\), dashed) and dipped loop with \(\Delta_{\min}=0.6\) in the
  interior, with fixed endpoint gap \(\Delta(0)=1\).
  \textbf{(b)} Single-evolution error and the worst-case bound
  \(\|\dot H\|_{\max}^2/(\Delta_{\min}^3T)\).
  \textbf{(c)} One-Richardson residual (extrapolation ratio $\alpha=2$): the dipped and undipped data
oscillate below the same endpoint-controlled reference scale
\(\|\dot H(0)\|^2/(\Delta(0)^4T^2)\).
  \textbf{(d)} One Richardson level with uniform runtime randomization and
  the \(M=1\) oscillatory bound
  \(C\|\dot H(0)\|^2/(\Delta(0)^4\Delta_{\min}T^3)\), with the a priori
  constant $C=C_1L_{1,3}(\alpha)\approx12.1$: one integration by parts gives $|\chi_{\mu,2}(\xi)|\le C_1/|\xi|$ for all $\xi$, with $C_1=(1-\lambda)^{-2}/\lambda\approx15.9$ at $\lambda=0.7$, and $L_{1,3}(\alpha)=\sum_\ell|w_{1,\ell}|\alpha^{-3\ell}\approx0.76$ at $\alpha=1.75$ is the sector factor of \cref{app:richardson_factor}, the $M=1$ decay contributing an extra $\alpha^{-\ell}$ per Richardson level; no constant is fitted. The uniform distribution makes the randomized
  oscillatory bias dominate the \(O(T^{-4})\) Richardson remainder, revealing
  the predicted \(T^{-3}\) scaling.}
  \label{fig:bounds}
\end{figure*}
We test the cancellation cascade of \cref{sec:cancellation} in two
settings.  First, we demonstrate the method on an entangled spiral
Heisenberg chain, comparing the small-runtime Berry-phase reconstruction
and the runtime scaling of the different estimators.  Second, we use a non-isospectral spin-\(1/2\) loop to test the analytic
bounds and to demonstrate how error cancellation replaces the
worst-case minimum-gap dependence by a tighter endpoint-controlled
scaling.

\subsection{Entangled spiral Heisenberg chain}
\label{sec:numerics-manybody}

We first consider a genuinely interacting model: a four-site Heisenberg
chain
\[
  H(s)
  =
  -\sum_i \mathbf B_i(s)\cdot\mathbf S_i
  +
  J\sum_{\langle ij\rangle}\mathbf S_i\cdot\mathbf S_j,
  \qquad
  J=1,
\]
with open boundary conditions (\(\langle ij\rangle\) runs over the \(L-1\) nearest-neighbor bonds) and the rotating field
\[
  \mathbf B_i(s)=B_0\bigl(\sin\theta\cos(2\pi s+\phi_i),\,\sin\theta\sin(2\pi s+\phi_i),\,\cos\theta\bigr),
\]
with cone angle \(\theta\) and a site-dependent azimuthal offset \(\phi_i=2\pi i/L\), \(L\) the number of sites. This spiral
field breaks the global spin-rotation symmetry, so the ground state is
entangled and the virtual-excitation contribution to the phase error is
nonzero. Since the loop is generated by a rigid global rotation, its
instantaneous spectrum is independent of \(s\), and hence
\[
  \Delta(s)=E_1(s)-E_0(s)=\Delta
\]
is constant along the loop. The gap \(\Delta\) is obtained by exact
diagonalization of \(H(0)\). The same rigid rotation also gives an exact propagator and the reference
Berry phase, where \(\langle\cdot\rangle_0\) denotes the expectation value
in the ground state of \(H(0)\):
\[
  \theta_B
  =
  2\pi\langle S^z_{\rm tot}\rangle_0,
  \qquad
  S^z_{\rm tot}:=\sum_{i=1}^{L}S_i^z.
\]

\cref{fig:spiral}(a) shows \(\theta_B\) modulo \(2\pi\) versus the cone
angle at the small runtime \(T=30\). The corresponding gap profile is
shown in \cref{fig:randomizedHT}(c) and varies over \([0.11,0.42]\) across
the sweep. The runtime-scaling reconstruction is shown as raw values
modulo \(2\pi\); because its error can exceed \(\pi\), the displayed
points scatter across the canonical window. By contrast, the
forward--reverse estimate and the randomized Richardson estimates
generally follow the analytic curve. The randomized estimates use five stratified runtimes per estimate
(\cref{sec:randomization}); their sample-to-sample fluctuation, the
\(T^{-1}N^{-3/2}\) term of \cref{thm:randomized_richardson}, grows as the gap
decreases but stays below \(0.03\) rad over the sweep. The two-Richardson
estimate follows the analytic curve at every angle (deviation below
\(0.05\) rad). Near the minimum-gap angles
(\(\theta/\pi\approx0.2,0.8\), \(\Delta\approx0.11\), \(T\Delta\approx3.3\)) the
one-Richardson estimate still shows a systematic deviation of about \(0.4\)
rad: its residual is the oscillatory term left by one Richardson level, whose suppression by the
runtime average requires \(\lambda\omega_1T\gg1\), and here
\(\lambda\omega_1T\approx2\). At \(T=20\) the two-Richardson estimate itself
deviates by about \(0.2\) rad at these angles.

\cref{fig:spiral}(b) traces the error versus runtime at
\(\theta=0.39\pi\), where the gap is \(\Delta=0.36\). As a baseline, we
compare with the runtime-scaling reconstruction of
Ref.~\cite{HayakawaSakamotoKiumi2025}, which shows the expected
\(O(T^{-1})\) behavior. The observed hierarchy is
\[
  T^{-1}
  \xrightarrow{\mathrm{fwd\mbox{-}rev}}
  T^{-2}
  \xrightarrow{\mathrm{1-fold\ Richardson + randomization}}
  T^{-4}
\]
\[
  T^{-1}
  \xrightarrow{\mathrm{fwd\mbox{-}rev}}
  T^{-2}
  \xrightarrow{\mathrm{2-fold\ Richardson + randomization}}
  T^{-6}.
\]
Forward--reverse symmetrization cancels the \(O(T^{-1})\) term, and one
Richardson level with runtime randomization removes the \(T^{-2}\)
non-oscillatory term and suppresses the leading oscillatory term in
expectation, yielding the observed \(T^{-4}\) decay; two levels give
\(T^{-6}\). For these higher-order remainders we do not obtain an explicit
gap dependence. These numerical results confirm the successive cancellation of the leading deterministic and oscillatory contributions predicted by the analysis.

\subsection{Theory bounds and the controlling gap}
\label{sec:numerics-bounds}

We next test the gap-dependent scaling predicted by the analytic bounds.
The spiral Heisenberg loop considered above is isospectral and satisfies
\(\Delta(0)=\Delta_{\min}\), so it cannot distinguish dependence on the
endpoint gap from dependence on the minimum gap along the path. To vary
these two gap scales independently, we now consider a spin-\(1/2\) cone loop and modulate its field magnitude to make it non-isospectral:
\[
  |\mathbf B(s)|
  =
  |\mathbf B|\bigl(1-a\sin^2(\pi s)\bigr),
  \qquad
  a=0.4 .
\]
This deformation preserves the endpoint gap while reducing the gap in
the interior, allowing us to test the predicted dependence of the
single-evolution error, the Richardson residual, and the randomized
oscillatory bias on \(\Delta(0)\) and \(\Delta_{\min}\). Specifically, it lowers the interior minimum gap to
\(\Delta_{\min}=0.6\), while keeping the endpoint quantities
\(\Delta(0)=1\) and \(\|\dot H(0)\|\) fixed. Since the direction of the
field is unchanged, the Berry phase is also unchanged.

The leading \(O(T^{-1})\) single-evolution error is bounded in terms of
the worst-case gap along the entire path,
\[
  |\varphi|
  \lesssim
  \frac{\|\dot H\|_{\max}^2}{\Delta_{\min}^3T}.
\]
Because this estimate retains only the minimum gap and discards where and
over how much of the path the gap becomes small, the localized interior
dip makes the bound substantially looser than in the isospectral case.
Accordingly, although both the numerical error and the bound increase,
the bound increases much more strongly, as shown in
\cref{fig:bounds}(b). By contrast, the
one-Richardson residual is an endpoint-controlled boundary term,
\[
  \frac{\|\dot H(0)\|^2}{\Delta(0)^4T^2}.
\]
Thus the dipped and undipped data oscillate below the same
endpoint-controlled asymptotic scale in \cref{fig:bounds}(c), independent of the
interior gap; the bound proven in \cref{thm:richardson_m} differs from this curve only by the $O(1)$ weight factor $L_1(\alpha)$.

Finally, \cref{fig:bounds}(d) tests the randomized oscillatory bound.
Here we use uniform runtime randomization, whose characteristic function
has decay order \(M=1\). The resulting oscillatory bias scales as
\[
  C
  \frac{\|\dot H(0)\|^2}
       {\Delta(0)^4\Delta_{\min}T^3}.
\]
This \(T^{-3}\) term dominates the non-oscillatory \(O(T^{-4})\)
Richardson remainder, so the data exhibit the predicted oscillatory
scaling and are consistent with its \(\Delta_{\min}^{-1}\) dependence.
The constant \(C\) is determined by the Richardson parameter and the
runtime-randomization parameter for the uniform distribution, as derived
in \cref{sec:richardson,sec:randomization}. It is therefore fixed entirely
a priori, with no fitting parameter introduced.

\section{Berry Phase Estimation Algorithm}\label{sec:algorithm}

We now describe the algorithmic implementation of the cancellation and
extrapolation results developed above.  The main protocol is formulated in
terms of Hadamard-test estimation of overlap amplitudes.  This choice is
natural in the present setting, since the relevant phase information is
contained in the overlap amplitude
\[
  \langle\psi(0)|U_T(1)|\psi(0)\rangle
  =
  e^{-i\theta_D}e^{i\theta_B}\,z(1),
\]
whose real and imaginary parts can be estimated directly at each randomized
runtime.  Although this shot-based approach has sample complexity
\(O(\varepsilon^{-2})\) for estimating an amplitude to accuracy
\(\varepsilon\), it allows the statistical error from runtime randomization
to be treated separately from the deterministic adiabatic bias.  Moreover,
leakage only reduces the overlap signal multiplicatively, and therefore
does not change the asymptotic runtime scaling.

\subsection{Hadamard-test based algorithm}
\label{sec:hadamard}

\begin{figure}[t!]
  \centering

  \begin{quantikz}[row sep=0.6cm, column sep=0.45cm]
    \lstick{$\ket{0}$} & \gate{H} & \ctrl{1} & \gate{H} & \meter{} & \rstick{$b_{+,j,\ell}$} \\
    \lstick{$\ket{\psi(0)}$} & \qw & \gate{U_{T_{j,\ell}}(1)} & \qw & \qw & \qw
  \end{quantikz}

  \vspace{0.2cm}
  \small (a) Forward evolution

  \vspace{0.6cm}

  \begin{quantikz}[row sep=0.6cm, column sep=0.45cm]
    \lstick{$\ket{0}$} & \gate{H} & \ctrl{1} & \gate{H} & \meter{} & \rstick{$b_{-,j,\ell}$} \\
    \lstick{$\ket{\psi(0)}$} & \qw & \gate{\hat U_{T_{j,\ell}}(1)} & \qw & \qw & \qw
  \end{quantikz}

  \vspace{0.2cm}
  \small (b) Reverse evolution

  \vspace{0.6cm}

  \includegraphics[width=\linewidth]{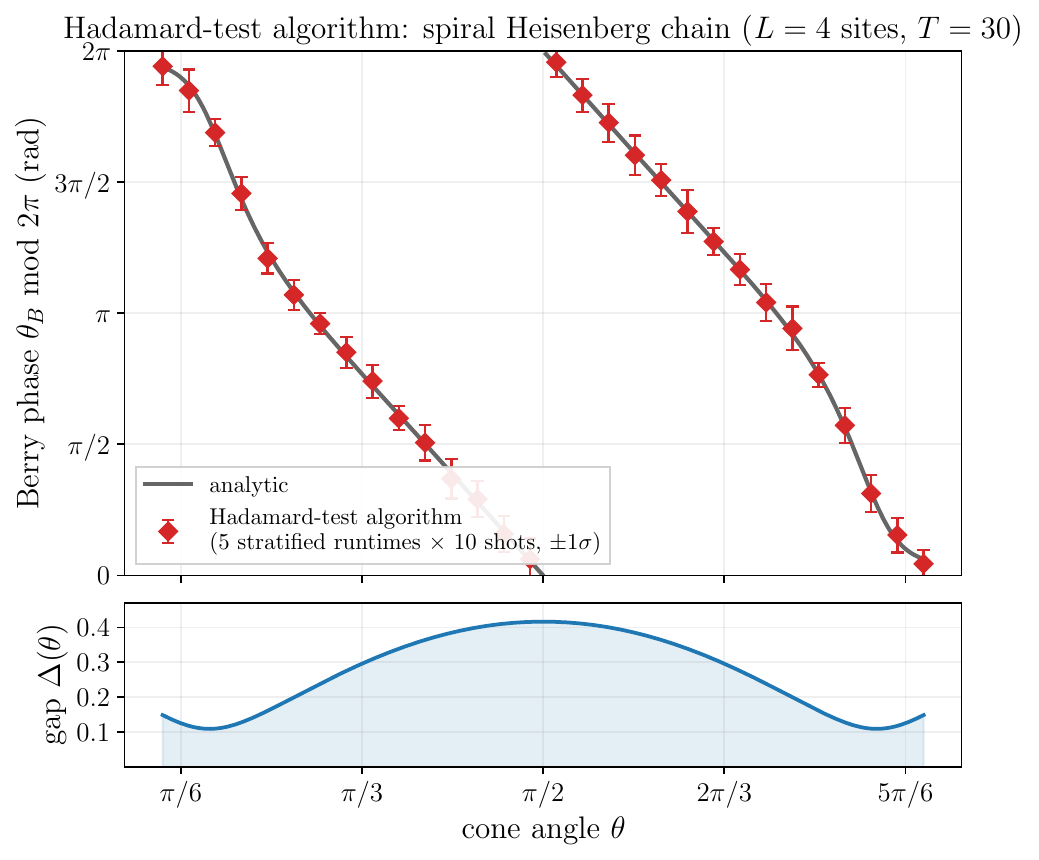}

  \vspace{0.2cm}
  \small (c) Runtime randomization and Richardson extrapolation

  \caption{Hadamard-test implementation of the randomized Berry phase estimation
    algorithm.  Panels (a) and (b) show the forward and reverse Hadamard
    tests used to estimate the overlap amplitudes at
    \(T_{j,\ell}=\alpha^\ell T X_j\), from which a mod-\(\pi\)
    forward--reverse Berry phase estimate is reconstructed.  Panel (c)
    shows the full algorithm on the same entangled spiral Heisenberg chain
    as in \cref{fig:spiral}: the Berry phase is estimated across the
    cone angle sweep at \(T=30\), using runtime randomization, two
    Richardson levels, and integer branch lifting.  Each point is the mean over \(24\) independent runs; within each run five stratified runtimes are used; at each runtime \(T_{j,\ell}\) the real and imaginary parts of the forward and reverse overlaps are each estimated from \(S=10\) Hadamard shots, the phases are lifted, and the Richardson combination is formed per runtime before the weighted average over runtimes.  Error bars show the single-run standard deviation
    from projective shot noise and runtime randomization; the lower strip
    shows the instantaneous gap \(\Delta(\theta)\).}
  \label{fig:randomizedHT}
\end{figure}

Let \(r\ge 0\) be the Richardson level, so that the target deterministic
bias order is \(T^{-(2r+2)}\).  Let \(\alpha>1\), and choose Richardson
weights \(w_{r,\ell}\), \(\ell=0,\dots,r\), satisfying
\[
  \sum_{\ell=0}^{r} w_{r,\ell}=1,
  \qquad
  \sum_{\ell=0}^{r} w_{r,\ell}\alpha^{-2q\ell}=0,
  \qquad
  q=1,\dots,r .
\]
Let \(\rho\) be a probability density supported on
\([1-\lambda,1+\lambda]\), with \(\lambda\in(0,1)\), whose weighted characteristic functions, defined as in \cref{sec:randomization} by $\chi_{\rho,k}(\xi):=\int_{1-\lambda}^{1+\lambda}x^{-k}e^{i\xi x}\rho(x)\,dx$, satisfy
\[
  |\chi_{\rho,k}(\xi)|
  \le
  C_{\rho,r}(1+|\xi|)^{-M},
  \qquad
  2\le k\le K,
\]
for some \(M>0\) with \(M\ge2r\) to achieve the target bias
\(O(T^{-(2r+2)})\), by \cref{thm:randomized_richardson}.
We also consider \(M=2r-1\) for \(r\ge1\) below to obtain an explicit
gap-dependent bound, with bias \(O(T^{-2r-1})\).
A compactly supported \(C^\infty\) bump density is one convenient choice (Example~3 of \cref{sec:randomization}), and satisfies this decay for every fixed $M$.

\medskip
\noindent\textbf{Step 1: Coarse branch resolution.}
We first obtain a coarse estimate \(\hat{\theta}_B^{\,\mathrm{coarse}}\) of
\(\theta_B\) over the full range \([0,2\pi)\), sufficient to identify the
correct \(\pi\)-branch in the subsequent fine estimates.  Since only
constant accuracy is required, this step can be implemented using Hadamard
tests at appropriately chosen runtimes.  Its cost is independent of the
target fine accuracy \(\varepsilon_B\), and is therefore subdominant in the
overall complexity.  The concrete construction of this coarse branch
estimate is described in the next subsection.

\medskip
\noindent\textbf{Step 2: Runtime randomization.}
Partition \([1-\lambda,1+\lambda]\) into \(N\) intervals \(I_1,\dots,I_N\) of equal
width \(2\lambda/N\), set \(p_j:=\int_{I_j}\rho\), and for each \(j=1,\dots,N\) sample
one runtime multiplier from the \(j\)-th stratum (stratified sampling, as in
\cref{sec:randomization}),
\[
  X_j\sim \rho(\,\cdot\mid I_j),
\]
and define the Richardson-node runtimes
\[
  T_{j,\ell}:=\alpha^\ell T X_j,
  \qquad
  \ell=0,\dots,r .
\]

\medskip
\noindent\textbf{Step 3: Hadamard-test sampling.}
For each pair \((j,\ell)\), perform Hadamard tests for the forward and
reverse propagators
\[
  U_{T_{j,\ell}}(1),
  \qquad
  \hat U_{T_{j,\ell}}(1).
\]
Using the appropriate control phases, estimate the real and imaginary parts
of the corresponding overlap amplitudes, each from \(S\) Hadamard shots
(\(4S\) shots per pair \((j,\ell)\)), giving estimates \(\hat z_{f}(T_{j,\ell})\) and
\(\hat z_{r}(T_{j,\ell})\) of \(\braket{\psi(0)|U_{T_{j,\ell}}(1)|\psi(0)}\) and
\(\braket{\psi(0)|\hat U_{T_{j,\ell}}(1)|\psi(0)}\).  From these estimates reconstruct
a forward--reverse Berry phase estimate
\[
  \tilde\theta_B(T_{j,\ell}):=\tfrac12\arg\bigl(\hat z_{f}(T_{j,\ell})\,\hat z_{r}(T_{j,\ell})\bigr)\pmod{\pi}.
\]

\medskip
\noindent\textbf{Step 4: Branch lifting.}
Using the coarse estimate \(\hat{\theta}_B^{\,\mathrm{coarse}}\), lift each
mod-\(\pi\) estimate to the unique representative in the interval
\[
  I :=
  \bigl(
    \hat{\theta}_B^{\,\mathrm{coarse}}-\pi/2,\,
    \hat{\theta}_B^{\,\mathrm{coarse}}+\pi/2
  \bigr).
\]
We denote the lifted value by
\[
  \theta_B^{\mathrm{lift}}(T_{j,\ell}).
\]

\medskip
\noindent\textbf{Step 5: Richardson extrapolation and averaging.}
For each sample \(j\), form the Richardson combination
\[
  \Theta_{R,j}(T)
  :=
  \sum_{\ell=0}^{r}
  w_{r,\ell}\,
  \theta_B^{\mathrm{lift}}(T_{j,\ell}).
\]
Finally output the weighted average over the strata,
\[
  \hat{\theta}
  :=
  \left(
  \sum_{j=1}^{N}p_j\,\Theta_{R,j}(T)
  \right)
  \pmod {2\pi}.
\]

\medskip
\noindent\textbf{Error and cost.}
We assume that a lower bound on \(\Delta_{\min}\) is known,
so that the runtimes below can be chosen; this is a standard assumption in adiabatic
quantum computation.
Conditioned on successful branch lifting, for decay order $M\ge 2r$ (e.g., the bump density) the deterministic bias is $O(T^{-(2r+2)})$, and the asymptotic runtime scaling is $T=O(\varepsilon_B^{-1/(2r+2)})$. However, because the gap dependence
of the non-oscillatory Richardson remainder is not explicit, we derive an
explicit gap-dependent sufficient scaling from the oscillatory term. For
\(r\geq1\), applying \cref{thm:randomized_richardson} with the minimal admissible decay order $M=2r-1$ makes the oscillatory term asymptotically dominant, yielding
\[
  T
  =
  O\!\left(
    \frac{
      \bigl(\|\dot H(0)\|\,\|\dot H(1)\|\bigr)^{1/(2r+1)}
    }{
      \Delta(0)^{4/(2r+1)}
      \Delta_{\min}^{(2r-1)/(2r+1)}
      \varepsilon_B^{1/(2r+1)}
    }
  \right).
\]
The statistical error has three sources: the Hadamard-test shot noise, the bias
of the phase of a finite-shot amplitude estimate, and the fluctuation from sampling the
runtimes. With \(S\) shots for each of the real and imaginary parts of the forward and
reverse overlaps at each runtime \(T_{j,\ell}\) (Step 3), the total number of shots is
\(4(r+1)NS\).
(i) Each phase estimate \(\tilde\theta_B(T_{j,\ell})\) carries shot noise
\(O(S^{-1/2})\), and the weighted average over the \(N\) strata reduces it to
\(O((NS)^{-1/2})\), so we require
\[
  NS=O\!\left(\varepsilon_B^{-2}\log\frac1\eta\right)
\]
to make it \(O(\varepsilon_B)\) with high probability, up to a constant
redistribution of the failure budget.
(ii) The phase of an \(S\)-shot amplitude estimate is biased by \(O(1/S)\): writing
\(\hat z=z+n\) with \(\mathbb E[n]=0\), one has
\(\arg\hat z-\arg z=\Im(n/z)-\tfrac12\Im\!\left[(n/z)^2\right]+\cdots\) with
\(\mathbb E[n^2]=(\Im^2z-\Re^2z)/S\neq0\). This bias is removed by none of the
cancellations above, so we take \(S=O(\varepsilon_B^{-1})\), hence
\(N=O(\varepsilon_B^{-1}\log(1/\eta))\).
(iii) By \cref{thm:randomized_richardson}, the fluctuation from the stratified runtimes
is \(O(T^{-1}N^{-3/2})\) up to path-dependent constants, which at
\(N=O(\varepsilon_B^{-1})\) is \(O(\varepsilon_B^{3/2}T^{-1})\ll\varepsilon_B\).
Altogether the total number of shots is \(O(\varepsilon_B^{-2}\log(1/\eta))\),
distributed over \(O(\varepsilon_B^{-1}\log(1/\eta))\) distinct runtimes.
For \(r\ge1\), with independent instead of stratified runtimes, the fluctuation in (iii) would be
\(O(T^{-2}N^{-1/2})\), and making it \(O(\varepsilon_B)\) would require
\(N=O(\varepsilon_B^{-2}T^{-4})\); with \(S=O(\varepsilon_B^{-1})\) and
\(T=O(\varepsilon_B^{-1/(2r+2)})\) the total number of shots would then be
\(NS=O(\varepsilon_B^{-2-(r-1)/(r+1)})\).

Finally, leakage does not change the deterministic cancellation mechanism.
In the Hadamard-test implementation, leakage reduces the overlap magnitude
and hence multiplies the phase-estimation noise by
\[
  (1-p_{\rm leak}^{\max})^{-1/2},
\]
where \(p_{\rm leak}^{\max}\) is the maximum leakage probability over the
runtimes used by the algorithm.  Since
\[
  p_{\rm leak}^{\max}=O(T^{-2}),
\]
we have
\[
  (1-p_{\rm leak}^{\max})^{-1/2}
  =
  1+O(T^{-2}).
\]
Thus leakage changes only the constants, not the asymptotic
\(\varepsilon_B\)-scaling, provided the algorithm is run in the adiabatic
regime where \(p_{\rm leak}^{\max}\ll 1\).

Moreover, the protocol does not require the exact ground state as input. If
\(\ket{\psi(0)}\) is replaced by a guiding state with preparation infidelity
\(p_{\mathrm{prep}}\), \cref{thm:prep_supp} bounds the extra bias, and the same choice
of \(M\) that fixes the runtime scaling also fixes the admissible infidelity. At
\(M=2r-1\), where the gap dependence above is explicit, the condition reads
\(p_{\mathrm{prep}}\lesssim(\Delta_{\min}T)^{2r-1}\varepsilon_B/(C_{\mu,K}L_{r,0})
\propto\varepsilon_B^{2/(2r+1)}\), a relaxation of the \(p_{\mathrm{prep}}
\lesssim\varepsilon_B\) that holds without randomization; at \(M\ge2r+2\) the
right-hand side is independent of \(\varepsilon_B\) altogether, though its value then
inherits the unspecified constant of the non-oscillatory remainder.
In either case, interference with the excited components gives the
overlap-magnitude lower bound \(1-2p_{\mathrm{prep}}-O(T^{-1})\), so the
shot-noise amplification is bounded by
\((1-2p_{\mathrm{prep}}-O(T^{-1}))^{-1}\) for sufficiently large \(T\).
For \(p_{\mathrm{prep}}\le p_*<1/2\) with fixed \(p_*\), this is absorbed
by a constant increase of \(S\).

\subsection{Branch lifting and topological discrimination}
\label{sec:integer_branch_lifting}

The forward--reverse estimator in \cref{eq:theta_est} determines
\(\theta_B\) only modulo \(\pi\).  Before applying Richardson extrapolation
or runtime randomization, we therefore need a coarse mod-\(2\pi\) estimate
to choose the correct \(\pi\)-branch.  This preprocessing step only requires
constant accuracy; the final high-precision scaling is supplied by the
branch-lifted forward--reverse Richardson and randomized estimator. 

For tasks that require only changes in the Berry phase, the global
branch-lifting procedure may be unnecessary; a simpler local
branch-tracking scheme can be used instead. The situation is different
when the Berry phase itself is quantized. For example, in a
one-dimensional inversion-symmetric band, the Berry phase, or Zak phase,
can take only the values \(0\) and \(\pi\) modulo \(2\pi\). Since these two
values are identical modulo \(\pi\), the forward--reverse estimator cannot
distinguish them. The branch-resolution procedure described below then
becomes the main estimator. 

The integer-runtime construction cancels the
dynamical phase and the chosen non-oscillatory adiabatic terms, although
the oscillatory boundary terms generally remain. Define the circular distance modulo \(L\) by
\[
  d_L(x,y):=\min_{k\in\mathbb Z}|x-y+kL|.
\]
Concretely, suppose that a coarse estimate \(\theta_{\rm br}\) satisfies
\[
  d_{2\pi}(\theta_{\rm br},\theta_B)<\eta_c
\]
and that each fine forward--reverse estimate satisfies
\[
  d_{\pi}\bigl(\tilde\theta_B(T_\ell),\theta_B\bigr)<\eta_f.
\]
If
\[
  \eta_c+\eta_f<\frac{\pi}{2},
\]
then each fine estimate has a unique representative in
\[
  I_{\rm br}
  =
  \left(
    \theta_{\rm br}-\frac{\pi}{2},
    \theta_{\rm br}+\frac{\pi}{2}
  \right),
\]
and this representative lies on the correct branch.  After this lifting
step, Richardson extrapolation and runtime randomization can be applied as
ordinary real linear operations.

We now construct such a coarse estimate.  We use the runtime-scaling
branch-resolution framework of
Ref.~\cite{HayakawaSakamotoKiumi2025}.  In the nearby-runtime version, the
relative runtime is chosen close to one using a Hamiltonian-norm bound, so
that the branch information is extracted from a small phase difference.
Here we instead use fixed integer runtime ratios.  This avoids a
near-degenerate phase comparison depending on \(T\) and \(H_{\max}\), and
the same integer conditions also cancel the leading non-oscillatory
adiabatic error.

Let
\[
  \vartheta_+(T)
  :=
  -\theta_D(T)+\theta_B+\varphi(T)
  \pmod {2\pi}
\]
be the phase of the forward overlap amplitude.
Here $\theta_D(T)$ is exactly linear in \(T\), while by \cref{lem:phase_error},
\[
  \varphi(T)
  =
  \frac{\varphi_1}{T}
  +
  \sum_{k=2}^{K}
  \frac{\varphi_k+\varphi_k^{(T)}}{T^k}
  +
  O(T^{-K-1}).
\]

Choose positive runtime ratios \(q_1,\ldots,q_m\) and integer coefficients
\(n_1,\ldots,n_m\), and define
\[
  \Theta_{\boldsymbol n}(T)
  :=
  \sum_{j=1}^{m}n_j\vartheta_+(q_jT)
  \pmod {2\pi}.
\]
Because the coefficients are integers, \(\Theta_{\boldsymbol n}(T)\) is
well defined modulo \(2\pi\) without unwrapping the individual overlap phases.
We impose
\begin{align*}
  \sum_{j=1}^{m} n_j &= 1,
  \\
  \sum_{j=1}^{m} n_j q_j &= 0,
  \\
  \sum_{j=1}^{m} n_j q_j^{-r} &= 0,
  \qquad r=1,\ldots,p .
\end{align*}
These conditions keep one copy of \(\theta_B\), cancel the dynamical phase,
and cancel the non-oscillatory adiabatic terms
\(T^{-1},\ldots,T^{-p}\).  Substitution gives
\begin{align*}
  \Theta_{\boldsymbol n}&(T)
  =
  \theta_B
  +
  \sum_{k=p+1}^{K}
  \frac{\varphi_k}{T^k}
  \sum_{j=1}^{m} n_j q_j^{-k}
  \nonumber\\
  &
  +
  \sum_{k=2}^{K}
  \frac{1}{T^k}
  \sum_{j=1}^{m}
  n_j q_j^{-k}\varphi_k^{(q_jT)}
  +
  O(T^{-K-1})
  \pmod {2\pi}.
\end{align*}
Thus the non-oscillatory residual starts at order \(T^{-p-1}\).  The
oscillatory boundary terms generally remain, but branch lifting only needs
constant accuracy.

As a simple low-order instance, take
\[
  (q_1,q_2,q_3)=(1,2,4),
  \quad
  (n_1,n_2,n_3)=(-2,5,-2).
\]
Then the above conditions hold with \(p=1\).  Hence
\[
  \Theta_I(T)
  :=
  -2\vartheta_+(T)+5\vartheta_+(2T)-2\vartheta_+(4T)
  \pmod {2\pi}
\]
satisfies
\[
  \Theta_I(T)=\theta_B+O(T^{-2})\pmod {2\pi}.
\]
Thus the integer combination improves the asymptotic order of the coarse
branch estimate from \(T^{-1}\) to \(T^{-2}\).

This improvement comes with a constant overhead.  One evaluation of
\(\Theta_I(T)\) uses the three runtimes \(T,2T,4T\), whose total evolution
time is \(7T\).  In addition, the integer coefficients amplify statistical
errors: for independent phase estimates with comparable variance, the
standard deviation is multiplied by
\[
  \sqrt{(-2)^2+5^2+(-2)^2}=\sqrt{33},
\]
while a deterministic worst-case error bound is multiplied by
\[
  |{-2}|+|5|+|{-2}|=9.
\]
Since branch lifting requires only constant accuracy and is performed at a
relatively small preprocessing runtime, the additional runtimes and larger
integer coefficients may outweigh the asymptotic \(T^{-2}\) improvement in
finite-\(T\) regimes.

In this work we use the construction only at the level of asymptotic branch
selection. Branch lifting assumes that all fine estimates entering the same
Richardson or randomized combination are lifted to the same local branch; this
can fail if the runtime variation is too large or the fine-estimator error is
not uniformly controlled over the runtimes used. A certified finite-precision
choice of \(T_{\rm br}\) and a guarantee of branch consistency would both
require a non-asymptotic adiabatic bound, including phase error and leakage,
as in Ref.~\cite{JansenRuskaiSeiler2007}. We instead focus on the asymptotic
cancellation order and precision scaling, leaving this certification to future work.

\section{Conclusion and Outlook}\label{sec:discussion}
\paragraph{\textbf{Summary of Results}}
We have developed a systematic analysis of finite-runtime adiabatic errors in Berry-phase estimation and shown how their structure can be exploited algorithmically. \Cref{lem:phase_error} separates the phase error into non-oscillatory bulk and oscillatory boundary contributions. \Cref{thm:cancellation} cancels all odd-order bulk terms by the forward--reverse combination, \cref{thm:richardson_m} removes the even-order bulk terms through \(T^{-2r}\) by Richardson extrapolation, and \cref{thm:randomized_richardson} suppresses the remaining oscillatory boundary terms by runtime randomization, reducing the bias to \(O(T^{-(2r+2)})\) with a subleading fluctuation \(O(T^{-1}N^{-3/2})\). \Cref{lem:prep_bias,thm:prep_supp} show that the zeroth-order oscillatory bias from an imperfect guiding state, which no deterministic step removes, is suppressed by the same randomization, so that the admissible preparation infidelity no longer depends on \(\varepsilon_B\). These results yield a randomized Hadamard-test algorithm over the full range \([0,2\pi)\) with coherent runtime \(O(\varepsilon_B^{-1/(2r+2)})\) at the standard \(O(\varepsilon_B^{-2})\) sample complexity, and the numerical results confirm the predicted cancellation hierarchy.
  
\paragraph{\textbf{Comparison with Energy Estimation}}
Ground-energy estimation and Berry-phase estimation address fundamentally different computational problems: the former concerns a spectral property of a fixed Hamiltonian, whereas the latter probes the geometry traced by an eigenstate along a closed path of Hamiltonians and therefore involves additional finite-runtime adiabatic errors. When a suitable guiding state is supplied as input, both problems are BQP-complete~\cite{HayakawaSakamotoKiumi2025}. In both settings, however, preparing a state with sufficient overlap with the relevant ground state can be computationally difficult and may dominate the practical cost. Berry-phase estimation may nevertheless be more favorable in some applications because the quantity of interest is determined by the path traversed by the eigenstate rather than by the energy of a particular target state. The loop may therefore begin at a reference Hamiltonian with a simple, efficiently preparable ground state while still probing nontrivial eigenstate geometry along the path. Identifying physically relevant settings in which this feature reduces the overall state-preparation cost is an important direction for future work.

\paragraph{\textbf{Extensions to Broader Quantum Geometry}}

The present work focuses on the adiabatic \(U(1)\) Berry phase of a nondegenerate eigenstate transported around a closed gapped loop, but geometric phases extend far beyond this setting. Nonadiabatic cyclic evolution gives rise to the Aharonov--Anandan phase~\cite{AharonovAnandan1987}, while noncyclic evolution is described by open-path geometric phases~\cite{SamuelBhandari1988}. Mixed states and open quantum systems admit several notions of geometric phase and holonomy~\cite{Sjoqvist2000,Tong2004,Carollo2003,SarandyLidar2006,Whitney2005,Uhlmann1986}. These extensions suggest that geometric error cancellation may provide a broadly applicable principle for controlling phase errors in driven and dissipative quantum dynamics. For degenerate eigenspaces, geometric transport is described by the non-Abelian Wilczek--Zee holonomy~\cite{WilczekZee1984}, which forms the basis of adiabatic holonomic quantum computation~\cite{Zanardi1999,DuanCiracZoller2001}. In this setting, finite-runtime adiabatic errors impose a direct tradeoff between gate speed and accuracy; suppressing them would allow faster holonomic gates at fixed accuracy and could overcome a central limitation of conventional adiabatic holonomic computation, complementing nonadiabatic holonomic approaches~\cite{Sjoqvist2012}. Berry phases and curvatures also underlie a wide range of physical observables and computational tasks, including Chern numbers and quantized Hall responses~\cite{TKNN1982,FukuiHatsugaiSuzuki2005}, electric polarization~\cite{KingSmith1993,Resta2000}, quantized charge pumping~\cite{Thouless1983}, and anomalous transport phenomena~\cite{XiaoChangNiu2010,Nagaosa2010}. Extending geometric error cancellation to these settings could enable more efficient quantum algorithms for topological invariants, transport coefficients, and many-body geometric responses. More broadly, the Berry curvature and quantum metric constitute the imaginary and real parts of the quantum geometric tensor~\cite{ProvostVallee1980,Kolodrubetz2017}, which connects eigenstate geometry to distinguishability, nonadiabatic response, localization, superfluid weight, flat-band physics, and quantum metrology~\cite{Peotta2015,OzawaGoldman2018,Torma2022}. Establishing error-cancellation principles for the full quantum geometric tensor could therefore open a unified route toward quantum algorithms for topology, many-body response, geometric control, and quantum information geometry.

\paragraph{\textbf{Robustness Against Other Sources of Error}}
In this work, we have focused on the finite-runtime adiabatic error of the ideal continuous-time evolution. A digital implementation introduces an additional product-formula error because the time-ordered evolution must be approximated by a finite sequence of gates; this error is controlled by the time step and the commutator structure of the Hamiltonian terms~\cite{ChildsSuTranWiebeZhu2021}. Whether the forward--reverse construction also suppresses this error depends on how the corresponding effective error Hamiltonian transforms under \(H\mapsto -H\), and generic Trotter errors need not cancel. Hardware implementations are further affected by control errors, decoherence, and measurement noise. Previous theoretical and experimental studies have shown that geometric phases can be insensitive to certain slowly varying or stochastic control fluctuations~\cite{DeChiaraPalma2003,Filipp2009,ZhuZanardi2005}. This robustness is not universal, however, and geometric protocols can have noise sensitivities comparable to those of dynamical protocols under common conditions~\cite{Colmenar2022}. An important direction for future work is therefore to identify the discretization errors and noise channels that are canceled or suppressed by the forward--reverse symmetry and to quantify their combined effect on the Berry-phase estimator.

\paragraph{\textbf{Toward an Optimal QPE-Based Algorithm}}
Our analysis is formulated for Hadamard-test readout, whose sample complexity \(O(\varepsilon_B^{-2})\) is not Heisenberg-limited. Quantum phase estimation would reach the Heisenberg limit, but it returns an eigenphase of \(U_T(1)\) rather than the phase of the overlap \(\langle\psi(0)|U_T(1)|\psi(0)\rangle\); the two need not coincide, and leakage and imperfect preparation then appear as failure probabilities rather than biases. Extending the cancellation principles established here to the eigenphase is a natural route toward an algorithm whose total cost is simultaneously optimal in the gap and in \(\varepsilon_B\), which we leave to future work.

\paragraph{\textbf{Toward Practical Quantum Advantage}}

An important direction for future work is to combine the present approach with recently proposed algorithms tailored to the early fault-tolerant regime~\cite{Granet2024,KiumiKoczor2025,PagesKiumiMorohoshi2026,HayataKikuchi2026}, as well as with Hamiltonian-simulation methods based on randomization and extrapolation, including multi-product-formula approaches~\cite{CarreraVazquezEggerOchsnerWoerner2023,Watson2025}. Our results suggest that Berry-phase estimation is particularly well suited to this setting: the forward--reverse structure suppresses the dominant adiabatic error, Richardson extrapolation removes the leading non-oscillatory residual, and the large shot budgets available in sampling-based implementations can simultaneously support runtime randomization and statistical estimation.

This raises the possibility that Berry phases, and perhaps geometric quantities more broadly, may provide favorable targets for practical quantum advantage. A meaningful assessment, however, requires an end-to-end resource analysis that jointly accounts for initial-state preparation, Hamiltonian-simulation cost, discretization error, sampling overhead, and hardware noise. Such an analysis should go beyond asymptotic scaling and establish explicit finite-runtime error bounds, including constants and confidence guarantees, so that the different error sources can be combined into a quantitative resource estimate. Because these contributions are strongly interdependent, combining geometric error cancellation with noise-robust simulation and efficiently preparable reference states may reduce the overall quantum resources far more substantially than suggested by the adiabatic-runtime improvement alone. Establishing such rigorous end-to-end resource guarantees in concrete applications is therefore an important direction for future work.

\section*{Data Availability}
The simulation code used in this work is available at \url{https://github.com/CKiumi/berry-phase-cancellation}.

\section*{Acknowledgements}
The author thanks Ryu Hayakawa and Kazuki Sakamoto for valuable discussions. This work was supported by JST PRESTO, Grant Number JPMJPR25F1; JSPS KAKENHI, Grant Number JP26K17052; and JST ASPIRE, Grant Number JPMJAP2319. 

\appendix
\crefalias{section}{appendix}
\crefalias{subsection}{appendix}

\section{Adiabatic Perturbation Theory}\label{app:APT}
In this appendix we review the adiabatic perturbation theory (APT) developed by Rigolin, Ortiz, and Ponce~\cite{RigolinOrtizPonce2008}, adapted to the notation of the main text.
The APT provides systematic $1/T$ corrections to the adiabatic approximation via recursion relations that reduce each order to explicit quadratures, without solving any differential equations.

Let $H(s)$ be a smooth family of Hamiltonians for $s\in[0,1]$ with $H(0)=H(1)$.
For each $s$, let $E_n(s)$ and $\ket{n(s)}$ denote the $n$-th instantaneous eigenvalue and eigenstate, with the ground state $\ket{\psi(s)}=\ket{0(s)}$ assumed non-degenerate for all~$s$.
We define
\[
  \Delta_{nm}(s) := E_n(s) - E_m(s),\quad
  M_{nm}(s) := \braket{n(s)|\dot{m}(s)},
\]
where the dot denotes $d/ds$.
Note that $M_{nn}(s)$ is purely imaginary (since $\braket{n(s)|n(s)}=1$ implies $\Re(M_{nn})=0$), and $M_{nm}(s) = -M_{mn}^*(s)$ for $n\neq m$.
We also define the instantaneous dynamical and Berry phases for each eigenstate:
\[
  \theta_D^{(m)}(s) := T\int_0^s E_m(\sigma)\,d\sigma,\quad
  \theta_B^{(m)}(s) := i\int_0^s M_{mm}(\sigma)\,d\sigma,
\]
and the relative Berry phase $\beta_{mn}(s) := \theta_B^{(m)}(s) - \theta_B^{(n)}(s)$. The evolved state $\ket{\Psi(s)} := U_T(s)\ket{\Psi(0)}$ of an arbitrary initial state $\ket{\Psi(0)}=\sum_m b_m(0)\ket{m(0)}$ is expanded in the instantaneous eigenbasis as
\[
  \ket{\Psi(s)} = \sum_n a_n(s)\ket{n(s)}.
\]
As in the main text, we assume for notational simplicity that the entire instantaneous spectrum is non-degenerate for all $s\in[0,1]$, so that every gap $\Delta_{nm}(s)$ appearing in the recursion below is nonzero. As noted in Ref.~\cite{RigolinOrtizPonce2008}, the expansion itself only requires that the initially occupied level---here the ground state---never crosses any other level: for the ground-state initial condition $b_n(0)=\delta_{n0}$, the low-order coefficients involve only the gaps $\Delta_{n0}(s)$ to the ground state, so degeneracies or level crossings within the excited spectrum can be accommodated at the cost of more cumbersome notation.
Substituting into the Schr\"odinger equation $i\ket{\dot{\Psi}(s)} = T H(s)\ket{\Psi(s)}$ and projecting onto $\bra{m(s)}$ yields the coupled equations
\begin{equation}\label{eq:am_eom}
  \dot{a}_m(s) = -iT E_m(s)\,a_m(s) - \sum_n a_n(s)\,M_{mn}(s).
\end{equation}
To remove the explicit dynamical and Berry phases, we define
\begin{equation*}
  a_m(s)=e^{-i\theta_D^{(m)}(s)}e^{i\theta_B^{(m)}(s)}\,b_m(s).
\end{equation*}
Substituting this into~\eqref{eq:am_eom} yields
\begin{equation*}
  \dot b_m(s)
  =
  -\sum_{n\neq m}
  e^{-iT\omega_{nm}(s)}e^{i\beta_{nm}(s)}M_{mn}(s)\,b_n(s),
\end{equation*}
where
\[
  \omega_{nm}(s):=\int_0^s \Delta_{nm}(\sigma)\,d\sigma.
\]
In terms of the main-text notation, $\Delta_n(s)=\Delta_{n0}(s)$, $M_n(s)=M_{n0}(s)$, $\beta_n(s)=\beta_{n0}(s)$, $\theta_D(s)=\theta_D^{(0)}(s)$, $\theta_B(s)=\theta_B^{(0)}(s)$, and $\omega_n=\omega_{n0}(1)$.
Thus, after extracting the dynamical and Berry phases, the remaining amplitudes \(b_m(s)\) are coupled only through the off-diagonal adiabatic couplings \(M_{mn}(s)\), modulated by rapidly oscillating phase factors. This motivates the adiabatic perturbation expansion.

We therefore expand the state as
\begin{equation*}
  \ket{\Psi(s)}=\sum_{p=0}^\infty T^{-p}\ket{\Psi^{(p)}(s)},
\end{equation*}
with
\begin{equation}\label{eq:Psi_p}
  \ket{\Psi^{(p)}(s)}
  =
  \sum_n
  e^{-i\theta_D^{(n)}(s)}e^{i\theta_B^{(n)}(s)}
  b_n^{(p)}(s)\ket{n(s)},
\end{equation}
where
\[
  b_n^{(p)}(s)
  =
  \sum_m
  e^{iT\omega_{nm}(s)}e^{-i\beta_{nm}(s)}\,b_{nm}^{(p)}(s).
\]
Equivalently,
\[
  \ket{\Psi(s)}
  =
  \sum_{n,m}\sum_{p=0}^{\infty}
  T^{-p}
  e^{-i\theta_D^{(m)}(s)}
  e^{i\theta_B^{(m)}(s)}
  b_{nm}^{(p)}(s)\ket{n(s)}.
\]
Substituting this ansatz into the Schr\"odinger equation, projecting onto \(\bra{n(s)}\), and matching powers of \(1/T\), we obtain
\begin{widetext}
\begin{equation}\label{eq:APT_recursion}
  i\Delta_{nm}(s)\,b_{nm}^{(p+1)}(s)
  +\dot b_{nm}^{(p)}(s)
  +W_{nm}(s)b_{nm}^{(p)}(s)
  +\sum_{k\neq n} M_{nk}(s)b_{km}^{(p)}(s)=0,
\end{equation}
\end{widetext}
where
\[
  W_{nm}(s):=M_{nn}(s)-M_{mm}(s).
\]
This recursion determines the coefficient \(b_{nm}^{(p+1)}(s)\) at the next order in terms of the lower-order amplitudes. In this way, APT generates the adiabatic expansion systematically, with corrections organized in successive powers of \(1/T\).
The initial conditions at each order are fixed by requiring the expansion to reproduce the initial state $\ket{\Psi(0)}$ at $s=0$, which gives
\begin{equation}
  \label{eq:APT_IC}
  \begin{aligned}
    b_{nm}^{(0)}(0)
    &= b_n(0)\,\delta_{nm},\\
    b_n^{(p)}(0)
    &= \sum_{m\neq n} b_{nm}^{(p)}(0)
     + b_{nn}^{(p)}(0)
     = 0,
    \qquad p\geq 1.
  \end{aligned}
  \end{equation}
For a system starting in the ground state, $b_n(0)=\delta_{n0}$.

\noindent\textbf{Zeroth order.}

At zeroth order, the recursion relation with $p=-1$ forces the off-diagonal coefficients to vanish, $b_{nm}^{(0)}(s)=0$ for $n\neq m$, and the diagonal part of the recursion at $p=0$ then gives $\dot b_{nn}^{(0)}(s)=0$. Together with the initial conditions, this yields
\begin{equation*}
  b_{nm}^{(0)}(s) = b_n(0)\,\delta_{nm}.
\end{equation*}
The zeroth-order solution yields the standard
adiabatic approximation:
\[
  \ket{\Psi^{(0)}(s)} = \sum_n 
  e^{-i\theta_D^{(n)}(s)}\,
  e^{i\theta_B^{(n)}(s)}\,b_n(0)\,\ket{n(s)}.
\]

\noindent\textbf{First order.}

\smallskip
\noindent\textit{Off-diagonal coefficients ($n\neq m$):}

Setting $p=0$ in the recursion~\eqref{eq:APT_recursion} and using $\dot{b}_{nm}^{(0)}=0$, $W_{nm}\,b_{nm}^{(0)}=0$ (since $b_{nm}^{(0)}\propto\delta_{nm}$ and $W_{nn}=0$), we obtain
\begin{equation}\label{eq:b1_offdiag}
  b_{nm}^{(1)}(s) = \frac{iM_{nm}(s)}{\Delta_{nm}(s)}\,b_m(0),\qquad n\neq m.
\end{equation}

\smallskip
\noindent\textit{Diagonal coefficients ($n=m$):}

When $n=m$, the term $i\Delta_{nn}(s)\,b_{nn}^{(p+1)}(s)$ vanishes identically, so the recursion at order $p=0$ becomes an identity.
Instead, one uses the recursion at $p=1$ with $n=m$:
\begin{equation*}
  \dot{b}_{nn}^{(1)}(s) + \sum_{k\neq n} M_{nk}(s)\,b_{kn}^{(1)}(s) = 0.
\end{equation*}
Integrating with~\eqref{eq:b1_offdiag} and using $M_{nk} = -M_{kn}^*$ gives
\begin{equation}\label{eq:b1_diag}
  b_{nn}^{(1)}(s) = i\sum_{m\neq n} J_{mn}(s)\,b_n(0) + b_{nn}^{(1)}(0),
\end{equation}
where
\begin{equation*}
  J_{mn}(s) := \int_0^s \frac{|M_{mn}(\sigma)|^2}{\Delta_{mn}(\sigma)}\,d\sigma,
\end{equation*}
and the initial condition~\eqref{eq:APT_IC} determines
\begin{equation}\label{eq:b1_diag_IC}
  b_{nn}^{(1)}(0) = -i\sum_{m\neq n}\frac{M_{nm}(0)}{\Delta_{nm}(0)}\,b_m(0).
\end{equation}

\smallskip
\noindent\textit{First-order correction to the state:}

Assembling these results, the first-order correction is
\begin{widetext}
\begin{align}\label{eq:Psi1}
  \ket{\Psi^{(1)}(s)}
  &= i\sum_{n\neq m}\bigg[
  e^{-i\theta_D^{(n)}(s)}\,e^{i\theta_B^{(n)}(s)}\,
  J_{mn}(s)\,b_n(0) \notag\\
  &\qquad\qquad
  + \left(
  \frac{e^{-i\theta_D^{(m)}(s)}\,e^{i\theta_B^{(m)}(s)}\,
  M_{nm}(s)}{\Delta_{nm}(s)}
  - \frac{e^{-i\theta_D^{(n)}(s)}\,e^{i\theta_B^{(n)}(s)}\,
  M_{nm}(0)}{\Delta_{nm}(0)}
  \right)b_m(0)
  \bigg]\ket{n(s)}.
\end{align}
\end{widetext}
The term proportional to \(J_{mn}(s)\) gives a correction to the amplitude of each initial component, which is absent in the standard non-APT treatment. The remaining two terms describe couplings between different instantaneous eigenstates: the term proportional to \(M_{nm}(s)/\Delta_{nm}(s)\) is non-oscillatory, while the term proportional to \(M_{nm}(0)/\Delta_{nm}(0)\) carries the dynamical phase of the \(n\)-th level and is oscillatory.

\smallskip
\noindent\textbf{Second order.}

\smallskip
\noindent\textit{Off-diagonal coefficients ($n\neq m$):}

Setting $p=1$ and $n\neq m$ in the recursion~\eqref{eq:APT_recursion}, using the first-order results, and separating the $k=m$ term from the sum yields
\begin{widetext}\[b_{nm}^{(2)} (s)=\frac{i}{\Delta _{nm}( s)}\left[ i\frac{d}{ds}\left(\frac{M_{nm}( s)}{\Delta _{nm}( s)}\right) b_{m} (0)+W_{nm} b_{nm}^{(1)} (s)+M_{nm}( s) b_{mm}^{(1)} (s)+\sum _{k\neq n,m} M_{nk}( s) b_{km}^{(1)} (s)\right].\]\end{widetext}
The right-hand side still contains the first-order coefficients. We now substitute the off-diagonal solution~\eqref{eq:b1_offdiag} and the diagonal solution~\eqref{eq:b1_diag} together with its initial condition~\eqref{eq:b1_diag_IC}, and collect the terms proportional to \(b_m(0)\) separately from those generated by the remaining initial amplitudes \(b_k(0)\) with \(k\neq m\). The first group consists of non-oscillatory contributions built from instantaneous quantities at \(s\), whereas the second group carries the endpoint factors \(M_{mk}(0)/\Delta_{mk}(0)\) inherited from \(b_{mm}^{(1)}(0)\), which are the source of the oscillatory boundary terms in the phase-error analysis below. Simplifying in this way gives
\begin{widetext}
\begin{align}\label{eq:b2_offdiag}
  b_{nm}^{(2)}(s)
  &= -\frac{1}{\Delta_{nm}(s)}\Biggl[
    \frac{d}{ds}\!\left(\frac{M_{nm}(s)}{\Delta_{nm}(s)}\right)
    + \frac{W_{nm}(s)\,M_{nm}(s)}{\Delta_{nm}(s)}+M_{nm}(s)\sum_{k\neq m}J_{km}(s)
    + \sum_{\substack{k\neq n\\k\neq m}}\frac{M_{nk}(s)\,M_{km}(s)}{\Delta_{km}(s)}
  \Biggr]b_m(0) \notag\\
  &\quad + \frac{M_{nm}(s)}{\Delta_{nm}(s)}\sum_{k\neq m}\frac{M_{mk}(0)}{\Delta_{mk}(0)}\,b_k(0).
\end{align}
\end{widetext}

\smallskip
\noindent\textit{Diagonal coefficients ($n=m$):}
Setting $p=2$ and $n=m$ in the recursion gives
\begin{equation}\label{eq:b2_diag}
  b_{nn}^{(2)}(s) = -\int_0^s \sum_{m\neq n}M_{nm}(\sigma)\,b_{mn}^{(2)}(\sigma)\,d\sigma + b_{nn}^{(2)}(0),
\end{equation}
where $b_{nn}^{(2)}(0) = -\sum_{m\neq n}b_{nm}^{(2)}(0)$ from~\eqref{eq:APT_IC}, and $b_{nm}^{(2)}(0)$ is obtained from~\eqref{eq:b2_offdiag} at $s=0$.

\smallskip
\noindent\textit{Second-order correction to the state:}
The full second-order correction is
\begin{equation*}
  \ket{\Psi^{(2)}(s)} = \sum_{n,m} e^{-i\theta_D^{(m)}(s)}\,e^{i\theta_B^{(m)}(s)}\,b_{nm}^{(2)}(s)\,\ket{n(s)},
\end{equation*}
with the coefficients given by~\eqref{eq:b2_offdiag} and~\eqref{eq:b2_diag}.
A key feature of the APT is that all corrections at order $p$ are determined from the coefficients at order $p-1$ by algebraic operations and quadratures, requiring only the instantaneous eigenvalues $E_n(s)$ and eigenstates $\ket{n(s)}$ of~$H(s)$.

Finally, we record the proof of the leading-order leakage expansion.
\begin{proof}[Proof of \cref{prop:leakage_error}]
For the ground-state initial condition $b_m(0)=\delta_{m0}$, the first-order result~\eqref{eq:Psi1} gives the amplitude on the excited state $\ket{n(1)}$ at the end of the loop as
\begin{align*}
  \frac{i}{T}\Bigl(
  &e^{-i\theta_D^{(0)}(1)}e^{i\theta_B^{(0)}(1)}
  \frac{M_{n0}(1)}{\Delta_{n0}(1)}
  \\
  &\quad
  -
  e^{-i\theta_D^{(n)}(1)}e^{i\theta_B^{(n)}(1)}
  \frac{M_{n0}(0)}{\Delta_{n0}(0)}
  \Bigr)
  +O(T^{-2}).
\end{align*}
Taking the squared modulus and summing over $n\neq0$ yields
\begin{align*}
  p_{\mathrm{leak}}(1)
  &=
  \frac{1}{T^2}\sum_{n\neq0}
  \left|
  \frac{M_{n0}(1)}{\Delta_{n0}(1)}
  -
  e^{-iT\omega_{n0}(1)}e^{i\beta_{n0}(1)}
  \frac{M_{n0}(0)}{\Delta_{n0}(0)}
  \right|^2
  \\
  &\quad
  +O(T^{-3}),
\end{align*}
which is precisely \cref{eq:pleak_expansion}.
\end{proof}

\section{APT analysis of the adiabatic phase error}\label{app:phase_APT}

\subsection{Exact initial state}\label{app:phase_exact}

In this subsection we use the adiabatic perturbation theory reviewed in \cref{app:APT} to derive the phase error expansion of \cref{lem:phase_error}. The scalar adiabatic error amplitude $z(s)=e^{i\theta_D(s)}e^{-i\theta_B(s)}\braket{\psi(s)|U_T(s)|\psi(0)}$ is obtained directly from the APT ansatz \eqref{eq:Psi_p}: projecting $\ket{\Psi(s)}$ onto the ground state $\bra{0(s)}$ and removing the ground-state dynamical and Berry phases gives
\begin{equation}\label{eq:z_APT}
  z(s)
  =
  \sum_{m}\sum_{p\ge0}
  T^{-p}\,
  e^{-iT\omega_{m0}(s)}\,
  e^{i\beta_{m0}(s)}\,
  b_{0m}^{(p)}(s),
\end{equation}
where the $b_{0m}^{(p)}(s)$ are the $T$-independent coefficients of the recursion \eqref{eq:APT_recursion}, and we used the ground-state initial condition $b_n(0)=\delta_{n0}$. The $m=0$ terms are non-oscillatory and start from $b_{00}^{(0)}(s)=1$, since the zeroth-order amplitudes are $b_{0m}^{(0)}=\delta_{0m}$. Writing $z=1+\delta$, the error amplitude splits into non-oscillatory ($m=0$) and oscillatory ($m\neq0$) parts, $\delta=\delta_{\mathrm{bulk}}+\delta_{\mathrm{osc}}$, where
\begin{align*}
  \delta_{\mathrm{bulk}}(s)
  &=
  \sum_{p\geq1}
  \frac{b_{00}^{(p)}(s)}{T^p},
  \\
  \delta_{\mathrm{osc}}(s)
  &=
  \sum_{p\geq1}
  \frac{1}{T^p}
  \sum_{m\neq0}
  e^{-iT\omega_{m0}(s)}
  e^{i\beta_{m0}(s)}
  b_{0m}^{(p)}(s).
\end{align*}
The phase error $\varphi(s):=\arg z(s)$ is the imaginary part of the principal logarithm,
\begin{equation}\label{eq:phi_taylor}
  \varphi(s)
  =
  \Im\log\bigl(1+\delta(s)\bigr)
  =
  \sum_{k=1}^{\infty}\frac{(-1)^{k-1}}{k}\,\Im\bigl(\delta(s)^{k}\bigr),
\end{equation}
convergent for $|\delta(s)|<1$. Factoring $1+\delta=(1+\delta_{\mathrm{bulk}})\bigl(1+\delta_{\mathrm{osc}}/(1+\delta_{\mathrm{bulk}})\bigr)$ before taking the logarithm separates the geometric and oscillatory contributions, $\varphi=\varphi_{\mathrm{bulk}}+\varphi_{\mathrm{osc}}$, with
\begin{align*}
  \varphi_{\mathrm{bulk}}
  &=
  \sum_{k=1}^{\infty}
  \frac{(-1)^{k-1}}{k}\,
  \Im\bigl(\delta_{\mathrm{bulk}}^{\,k}\bigr),
  \\
  \varphi_{\mathrm{osc}}
  &=
  \sum_{k=1}^{\infty}
  \frac{(-1)^{k-1}}{k}\,
  \Im\!\left[
    \left(
      \frac{\delta_{\mathrm{osc}}}
           {1+\delta_{\mathrm{bulk}}}
    \right)^{\!k}
  \right].
\end{align*}
The first sum is built solely from the geometric coefficients $b_{00}^{(p)}$ and carries no $T$-oscillation; collecting powers of $1/T$ gives the bulk coefficients $\varphi_k$. In the second sum, every term carries the factor $\delta_{\mathrm{osc}}^{\,k}$ with $k\ge1$, hence at least one $e^{-iT\omega_{m0}}$, and gives the boundary coefficients $\varphi_k^{(T)}$. Since the expansion contains only powers of $\delta_{\mathrm{osc}}/(1+\delta_{\mathrm{bulk}})$ and never a conjugate, each oscillatory term retains a net frequency $\Omega\ge\Delta_{\min}$ given by a positive sum of gaps. Because $b_{0m}^{(1)}=0$ for $m\neq0$, the oscillatory sum is $O(T^{-2})$, so the leading boundary term is $\varphi_2^{(T)}$. As emphasized in the introduction, the expansion is understood in the asymptotic sense: for each fixed $K$, truncating at order $T^{-K}$ leaves a remainder $O(T^{-K-1})$ inherited from the APT expansion, and we do not pursue uniform nonasymptotic bounds on this remainder.

More generally, the oscillatory coefficients have the following form.
For each fixed order $k$, since the Hilbert space---and hence the number of contributing modes---is finite-dimensional in our setting, there exists a finite index set $I_k$ such that
\begin{equation}\label{eq:osc_general_form}
  \varphi_k^{(T)}
  =
  \sum_{\nu\in I_k}
  \left[
    a_{k,\nu}\cos(\Omega_{k,\nu}T)
    +
    b_{k,\nu}\sin(\Omega_{k,\nu}T)
  \right],
\end{equation}
where $a_{k,\nu}$ and $b_{k,\nu}$ are independent of $T$. Each frequency
$\Omega_{k,\nu}$ is a positive sum of integrated gaps,
\begin{equation}\label{eq:osc_frequency_form}
  \Omega_{k,\nu}
  =
  \omega_{n_1 0}(1)+\cdots+\omega_{n_q 0}(1),
  \qquad q\geq 1,
\end{equation}
and hence satisfies
\[
  \Omega_{k,\nu}\geq \Delta_{\min}>0 .
\]
Thus no zero-frequency component appears in the oscillatory sector.

\noindent\textbf{First order.} At first order the recursion~\eqref{eq:APT_recursion} with the ground-state initial condition gives the non-oscillatory coefficient
\begin{equation*}
  b_{00}^{(1)}(s) = i\sum_{m\neq 0} J_{m0}(s),
\end{equation*}
while the off-diagonal survival coefficients vanish, $b_{0m}^{(1)}(s)=0$ for $m\neq0$. Hence at order $1/T$ the oscillatory part $\delta_{\mathrm{osc}}$ contributes nothing, and from~\eqref{eq:phi_taylor},
\begin{align}
  \varphi_1(s)
  &=
  \Im\bigl(b_{00}^{(1)}(s)\bigr)
  =
  \int_0^s
  \sum_{n\neq 0}
  \frac{|M_{n0}(\sigma)|^2}
       {\Delta_{n0}(\sigma)}
  \,d\sigma.
  \label{eq:phi1}
\end{align}
The vanishing of $b_{0m}^{(1)}$ is the reason there is no first-order boundary term; the oscillatory contribution first appears at second order, through $b_{0m}^{(2)}$, as $\varphi_2^{(T)}$.

\noindent\textbf{Second order.} At second order both the bulk and oscillatory coefficients contribute. From~\eqref{eq:phi_taylor}, collecting the $1/T^2$ terms of $\Im\log(1+\delta)$ gives
\begin{align*}
  \frac{\varphi_2+\varphi_2^{(T)}}{T^2}
  &=
  \frac{1}{T^2}
  \Im\!\left(
    b_{00}^{(2)}
    -\tfrac12\bigl(b_{00}^{(1)}\bigr)^2
  \right)
  \\
  &\quad
  +
  \frac{1}{T^2}
  \sum_{m\neq0}
  \Im\!\left(
    e^{-iT\omega_{m0}}
    e^{i\beta_{m0}}
    b_{0m}^{(2)}
  \right),
\end{align*}
since $b_{0m}^{(1)}=0$ removes the $m\neq0$ contribution to $\delta_{\mathrm{osc}}$ at first order. The first group is non-oscillatory and gives $\varphi_2$; the second carries $e^{-iT\omega_{m0}}$ and gives $\varphi_2^{(T)}$.

\emph{Bulk coefficient.} The ground-state coefficient $b_{00}^{(2)}(s)$ follows from the recursion~\eqref{eq:APT_recursion} with $n=0$. Two contributions drop out of the imaginary part: $\bigl(b_{00}^{(1)}\bigr)^2=-\bigl(\sum_{m\neq0}J_{m0}\bigr)^2$ is purely real, and the initial value $b_{00}^{(2)}(0)=-\sum_{m\neq0}|M_{m0}(0)|^2/\Delta_{m0}(0)^2$ is also real, so $\varphi_2(s)=\Im\,b_{00}^{(2)}(s)$, with only the integral term of \eqref{eq:b2_diag} contributing. Using $W_{n0}=M_{nn}-M_{00}=-i(\dot\theta_B^{(n)}-\dot\theta_B)$ and writing $M_{n0}=|M_{n0}|\,e^{i\phi_n}$ (defined wherever $M_{n0}\neq0$; only the combinations $|M_{n0}|^2\dot\phi_n$ and $|M_{n0}|\sin(\cdot)$ appear below, so zeros of the coupling cause no ambiguity), so that $\Im(\dot M_{n0}M_{n0}^*)=\dot\phi_n|M_{n0}|^2$, the imaginary part evaluates to
\begin{align}
  \varphi_2(s)
  &=
  \sum_{n\neq 0}\int_0^s
  \frac{
    \bigl(\dot{\beta}_{n0}(\sigma)-\dot{\phi}_n(\sigma)\bigr)
    |M_{n0}(\sigma)|^2
  }{
    \Delta_{n0}(\sigma)^2
  }\,d\sigma
  \nonumber\\
  &\quad
  -i\int_0^s
  \mathcal{A}^{(2)}(\sigma)\,d\sigma,
  \label{eq:phi2}
\end{align}
where $\dot{\beta}_{n0}=\dot{\theta}_B^{(n)}-\dot{\theta}_B$ and
\[
  \mathcal{A}^{(2)}(s) :=
  \sum_{\substack{n,k\neq 0\\ n\neq k}}
  \frac{M_{0n}(s)\,M_{nk}(s)\,M_{k0}(s)}
  {\Delta_{n0}(s)\,\Delta_{k0}(s)}
\]
is purely imaginary, so that $-i\int_0^s\mathcal{A}^{(2)}\,d\sigma$ is real, as follows from the anti-hermiticity relation $M_{ab}^*=-M_{ba}$ and relabeling the summation indices.
Physically, \(\varphi_2\) consists of two parts: an integral involving the mismatch between the relative Berry-phase rate \(\dot{\beta}_{n0}\) and the coupling-phase rate \(\dot{\phi}_n\), and the integral of \(\mathcal{A}^{(2)}\), which captures virtual excitation processes
\(
\ket{0}\to\ket{n}\to\ket{k}\to\ket{0}
\)
weighted by the spectral gaps.

\emph{Oscillatory coefficient.} The off-diagonal coefficient $b_{0m}^{(2)}(s)$ inherits the endpoint factor $M_{m0}(0)/\Delta_{m0}(0)$ from the first-order excited amplitude. Substituting into the oscillatory group above and writing $M_{m0}=|M_{m0}|\,e^{i\phi_m}$ yields
\begin{widetext}
\begin{equation}\label{eq:phi2T}
  \varphi_2^{(T)}(s)
  =
  \sum_{m\neq 0}
  \frac{|M_{m0}(0)|\,|M_{m0}(s)|}{\Delta_{m0}(s)\,\Delta_{m0}(0)}\,
  \sin\bigl(\beta_{m0}(s)+\phi_m(0)-\phi_m(s)-T\omega_{m0}(s)\bigr).
\end{equation}
\end{widetext}
At the endpoint $s=1$, with the closed loop $\Delta_{m0}(1)=\Delta_{m0}(0)$, this reduces to the form stated in \cref{lem:phase_error}. The bounds stated in \cref{lem:phase_error} also follow readily: since $|M_{n0}|=|\braket{n|\dot H|\psi}|/\Delta_{n0}$, completeness of the eigenbasis gives $\sum_{n\neq0}|\braket{n|\dot H|\psi}|^2\le\dot H_{\max}^2$, hence $\varphi_1\le\dot H_{\max}^2/\Delta_{\min}^3$; similarly, using $\Delta_{n0}(0)\ge\Delta(0)$ and the Cauchy--Schwarz inequality, $|\varphi_2^{(T)}|\le\sum_{n\neq0}|M_{n0}(0)|\,|M_{n0}(1)|/\Delta_{n0}(0)^2\le\|\dot H(0)\|\,\|\dot H(1)\|/\Delta(0)^4$.

We can now assemble the proof of the lemma.
\begin{proof}[Proof of \cref{lem:phase_error}]
Substituting the APT expansion into \eqref{eq:z_APT} and collecting powers of $1/T$ in \eqref{eq:phi_taylor} yields the expansion \eqref{eq:phi_expansion}, with the $T$-independent non-oscillatory coefficients $\varphi_k$ arising from the bulk sector and the oscillatory coefficients $\varphi_k^{(T)}$ of the form \eqref{eq:osc_general_form}, whose frequencies obey \eqref{eq:osc_frequency_form} and hence satisfy $\Omega_{k,\nu}\ge\Delta_{\min}>0$. The leading coefficients are given by \eqref{eq:phi1} and by \eqref{eq:phi2T} at $s=1$, and the stated bounds on $\varphi_1$ and $\varphi_2^{(T)}$ were derived above. The truncation error at order $K$ is $O(T^{-K-1})$ in the asymptotic sense discussed above.
\end{proof}

\subsection{Imperfect initial state}\label{app:prep_expansion}

We extend the phase-error analysis of \cref{app:phase_exact} to an imperfect initial state and derive the preparation-bias expansion of \cref{lem:prep_bias}.
We use the notation of \cref{sec:imperfect}: \(\ket{\psi_{\mathrm{prep}}}=\sum_n c_n\ket{n(0)}\), \(p_{\mathrm{prep}}=1-|c_0|^2<1/2\), and \(r_n=|c_n|^2/|c_0|^2\). We abbreviate \(\omega_n=\omega_{n0}(1)\) and \(\beta_n=\beta_{n0}(1)\).

  The APT of \cref{app:APT} is set up for a general initial state, so the
  expansion of \(\ket{\Psi(1)}\) holds with \(b_n(0)=c_n\). Contracting it with
  \(\bra{\psi_{\mathrm{prep}}}=\sum_n\bar c_n\bra{n(1)}\) and removing the ground-state
  dynamical and Berry phases gives
  \begin{equation*}
    \tilde z
    =\sum_{n,m}\sum_{p\ge0}
    T^{-p}\,\bar c_n\,
    e^{-iT\omega_{m0}(1)}e^{i\beta_{m0}(1)}\,
    b_{nm}^{(p)}(1).
  \end{equation*}

As in \cref{app:phase_exact}, we separate the \(m=0\) and \(m\ne0\) sectors. After dividing by the positive ground-state weight \(|c_0|^2\), write
\[
  \frac{\tilde z}{|c_0|^2}
  =1+\tilde\delta_{\mathrm{bulk}}+\tilde\delta_{\mathrm{osc}},
\]
\begin{align*}
  \tilde\delta_{\mathrm{bulk}}
  &=\frac{1}{|c_0|^2}\sum_{p\ge1}T^{-p}
    \sum_n\bar c_n b_{n0}^{(p)}(1),\\
  \tilde\delta_{\mathrm{osc}}
  &=\frac{1}{|c_0|^2}\sum_{p\ge0}T^{-p}
    \sum_{m\ne0}e^{-iT\omega_m}e^{i\beta_m}
    \sum_n\bar c_n b_{nm}^{(p)}(1).
\end{align*}
These are the counterparts of \(\delta_{\mathrm{bulk}}\) and \(\delta_{\mathrm{osc}}\) in \cref{app:phase_exact}. The bulk correction still starts at \(O(T^{-1})\). However, since \(b_{nm}^{(0)}(1)=c_n\delta_{nm}\), the oscillatory correction now starts at order \(T^0\):
\[
  \tilde\delta_{\mathrm{osc}}=w+O(T^{-1}),
  \qquad
  w:=\sum_{n\ge1}r_ne^{i(\beta_n-\omega_nT)}.
\]
When the initial state is exact, \(w=0\) and the two corrections reduce to those of \cref{app:phase_exact}, with \(\delta_{\mathrm{osc}}=O(T^{-2})\).

We now use the same factorization as in \cref{app:phase_exact}:
\[
  \tilde z=|c_0|^2(1+\tilde\delta_{\mathrm{bulk}})
  \left(1+\frac{\tilde\delta_{\mathrm{osc}}}
                    {1+\tilde\delta_{\mathrm{bulk}}}\right).
\]
The assumption \(p_{\mathrm{prep}}<1/2\) gives \(|w|\le\sum_{n\ge1}r_n=p_{\mathrm{prep}}/(1-p_{\mathrm{prep}})<1\). Thus, for sufficiently large \(T\), both logarithmic series below converge and give the phase on the same branch as the original overlap:
\begin{align*}
  \arg\tilde z
  &=\tilde\varphi_{\mathrm{bulk}}+\tilde\varphi_{\mathrm{osc}},\\
  \tilde\varphi_{\mathrm{bulk}}
  &:=\Im\log(1+\tilde\delta_{\mathrm{bulk}})\\
  &=\sum_{q\ge1}\frac{(-1)^{q-1}}q
    \Im\bigl(\tilde\delta_{\mathrm{bulk}}^{\,q}\bigr),\\
  \tilde\varphi_{\mathrm{osc}}
  &:=\Im\log\left(1+\frac{\tilde\delta_{\mathrm{osc}}}
                         {1+\tilde\delta_{\mathrm{bulk}}}\right)\\
  &=\sum_{q\ge1}\frac{(-1)^{q-1}}q
    \Im\left[\left(\frac{\tilde\delta_{\mathrm{osc}}}
                            {1+\tilde\delta_{\mathrm{bulk}}}\right)^q\right].
\end{align*}
The bulk series contains only \(T\)-independent coefficients multiplying powers of \(1/T\). Every term in the oscillatory series contains at least one factor from the \(m\ne0\) sector. Products add the positive frequencies \(\omega_m\), without introducing conjugate factors, so every resulting frequency is at least \(\Delta_{\min}\). Hence the second series has no non-oscillatory component.

\noindent\textbf{Zeroth order.} The bulk phase has no \(T^0\) term. Setting the inverse-runtime corrections to zero in the oscillatory phase gives
\[
  \varphi_0^{\mathrm{prep}(T)}=\Im\log(1+w),
\]
which is \cref{eq:prep_k0}. This term is absent for the exact initial state.

  Expanding the logarithm in \cref{eq:prep_k0} gives the absolutely convergent
  Fourier series
  \begin{align}
    \varphi_0^{\mathrm{prep}(T)}
    &=\sum_{q\ge1}\frac{(-1)^{q-1}}{q}
    \sum_{n_1,\dots,n_q\ge1}
    r_{n_1}\cdots r_{n_q} \notag\\
    &\quad\times
    \sin\Bigl(\textstyle\sum_{j=1}^{q}\beta_{n_j}
    -T\sum_{j=1}^{q}\omega_{n_j}\Bigr).
    \label{eq:prep_k0_series}
  \end{align}
  Every frequency is a sum of \(q\ge1\) of the \(\omega_n\), hence at least
  \(\omega_1>0\): no \(T\)-independent component survives at \(k=0\). Summing the
  moduli gives \(\sum_{q\ge1}(\sum_{n\ge1}r_n)^{q}/q
  =\log\bigl[(1-p_{\mathrm{prep}})/(1-2p_{\mathrm{prep}})\bigr]
  =p_{\mathrm{prep}}+O(p_{\mathrm{prep}}^{2})\), and retaining \(q=1\) gives
  \(\varphi_0^{\mathrm{prep}(T)}=\sum_{n\ge1}r_n\sin(\beta_n-\omega_nT)
  +O(p_{\mathrm{prep}}^{2})\).

\noindent\textbf{Higher orders.} Subtracting the bulk and oscillatory phases of the exact initial state separately gives
\begin{align*}
  \tilde\varphi_{\mathrm{bulk}}-\varphi_{\mathrm{bulk}}
  &=\sum_{k=1}^{K}\frac{\varphi_k^{\mathrm{prep}}}{T^k}
    +O(T^{-K-1}),\\
  \tilde\varphi_{\mathrm{osc}}-\varphi_{\mathrm{osc}}
  &=\varphi_0^{\mathrm{prep}(T)}
    +\sum_{k=1}^{K}\frac{\varphi_k^{\mathrm{prep}(T)}}{T^k}
    +O(T^{-K-1}).
\end{align*}
Their sum is \cref{eq:prep_expansion}. Unlike the exact case, arbitrarily many factors of \(w\) can occur at a fixed order in \(1/T\), since \(w\) carries no inverse-runtime factor. The oscillatory coefficients can therefore contain countably many modes. At each fixed order, these sums remain absolutely convergent: their terms are bounded by a polynomial in the number of factors times a geometric sequence with ratio \(\sum_{n\ge1}r_n<1\). All frequencies remain positive sums of integrated gaps, as established above.

After subtracting the exact-state result, the higher-order coefficients vanish when the excited-state amplitudes vanish. Their dependence on those amplitudes is smooth near the exact state, so they are \(O(p_{\mathrm{prep}}^{1/2})\). The zeroth-order term instead depends on the populations \(|c_n|^2\) and is \(O(p_{\mathrm{prep}})\), as shown above.

We can now assemble the proof of the lemma.
\begin{proof}[Proof of \cref{lem:prep_bias}]
The factorization of \(\tilde z\) and the logarithmic expansion above yield \cref{eq:prep_expansion} after subtracting \(\varphi\). The zeroth-order term is \cref{eq:prep_k0}, with the absolutely convergent Fourier series \cref{eq:prep_k0_series}. The frequency bounds, absolute convergence at higher orders, and preparation-infidelity dependence were established above. Truncation at fixed order \(K\) gives the asymptotic remainder \(O(T^{-K-1})\), as in \cref{app:phase_exact}.
\end{proof}

\section{Bound on the second-order term}\label{app:bound}

Here we bound the second-order coefficients introduced in \cref{lem:phase_error}: Terms~1 and~2 bound the two contributions to \(\varphi_2\) in \cref{eq:phi2}, and Term~3 bounds the oscillatory coefficient \(\Phi_2^{(T)}\) of \cref{thm:cancellation}. Combining them yields the leading-order bound on \(\tilde{\theta}_B(T)-\theta_B\) for the forward--reverse estimator in terms of physical parameters.
Explicitly, by \cref{eq:phi2,eq:Phi2T}, the quantities to be bounded are
\begin{align*}
  \Phi_2=\varphi_2
  &=
  \underbrace{\sum_{n\neq 0}\int_0^1
  \frac{\bigl(\dot{\beta}_{n0}-\dot{\phi}_n\bigr)|M_{n0}|^2}{\Delta_{n0}^2}\,ds}_{\text{Term 1}}
  \\
  &\quad
  \underbrace{\;-\;i\int_0^1 \mathcal{A}^{(2)}(s)\,ds}_{\text{Term 2}},
  \\
  \Phi_2^{(T)}
  &=
  \underbrace{\sum_{n\neq 0}B_n\cos(\omega_nT)}_{\text{Term 3}}.
\end{align*}

\paragraph{Term 1: the $(\dot{\beta}_{n0}-\dot{\phi}_n)$ 
integral.}
This term involves $M_{n0}^*\,b_{n0}^{(2)}$ before 
splitting into sub-terms.
From the explicit form of $b_{n0}^{(2)}$ 
in~\eqref{eq:b2_offdiag}, each contribution is 
bounded by products of $|M|$'s, $|\dot{M}|$'s, and 
inverse gaps.
Since the integrand $(\dot\beta_{n0}-\dot\phi_n)|M_{n0}|^2/\Delta_{n0}^2$ is gauge invariant, we may evaluate it in the parallel-transport gauge, in which $M_{nn}=M_{00}=0$ and hence $\dot\beta_{n0}=0$; the remaining contribution obeys $|\dot\phi_n|\,|M_{n0}|^2\le|\dot M_{n0}|\,|M_{n0}|$.
Differentiating $M_{n0} = -\braket{n|\dot{H}|\psi}
/\Delta_{n0}$ gives, in this gauge,
\[
  |\dot{M}_{n0}| \leq
  O\!\left(\frac{\ddot{H}_{\max}}{\Delta_{\min}}
  + \frac{\dot{H}_{\max}^2}{\Delta_{\min}\gamma_{\wedge}}\right),
\]
where $\gamma_{\wedge}:=\min(\Delta_{\min},\gamma_{\mathrm{ex}})$, with the minimum excited-state splitting
\[
  \gamma_{\mathrm{ex}}
  :=
  \min_{s\in[0,1]}
  \min_{\substack{n\neq k\\ n,k\neq 0}}
  |E_n(s)-E_k(s)|,
\]
and the $\gamma_{\wedge}$ term arises from the derivative $\ket{\dot n}$ of the excited eigenvector.
Combining with $|M_{n0}|\le\dot H_{\max}/\Delta_{\min}$, the integrand satisfies
\[
  \left|\frac{(\dot{\beta}_{n0}-\dot{\phi}_n)\,
  |M_{n0}|^2}{\Delta_{n0}^2}\right|
  \leq
  O\!\left(
  \frac{\dot H_{\max}\,\ddot{H}_{\max}}{\Delta_{\min}^4}
  + \frac{\dot{H}_{\max}^3}{\Delta_{\min}^4\,\gamma_{\wedge}}
  \right).
\]

\paragraph{Term 2: the $\mathcal{A}^{(2)}$ integral.}
Introduce the vector \(v(s)\) on the excited subspace by
\[
  v_n(s):=\frac{M_{n0}(s)}{\Delta_{n0}(s)},\qquad n\neq 0.
\]
Then \(\mathcal A^{(2)}\) can be viewed as the quadratic form
\begin{align*}
  \mathcal A^{(2)}(s)
  &=
  \sum_{\substack{n,k\neq 0\\ n\neq k}}
  \frac{M_{0n}(s)\,M_{nk}(s)\,v_k(s)}{\Delta_{n0}(s)}
  \\
  &=
  -\sum_{\substack{n,k\neq 0\\ n\neq k}}
  \overline{v_n(s)}\,M_{nk}(s)\,v_k(s),
\end{align*}
where we used \(M_{0n}=-\overline{M_{n0}}\).
Hence
\[
  |\mathcal A^{(2)}(s)|
  \le
  \|M_Q(s)\|\,\|v(s)\|^2,
\]
where \(M_Q(s)=(M_{nk}(s))_{n,k\neq 0}\) denotes the excited-sector
Berry-connection matrix.
Now
\[
  \|v(s)\|^2
  =
  \sum_{n\neq 0}\frac{|M_{n0}(s)|^2}{\Delta_{n0}(s)^2}
  \le
  \frac{\dot H_{\max}^2}{\Delta_{\min}^4},
\]
using \(|M_{n0}(s)|\le \dot H_{\max}/\Delta_{\min}\).
Recalling the minimum excited-state splitting $\gamma_{\mathrm{ex}}$ defined in Term~1, the Hellmann--Feynman relation bounds the matrix entries, $|M_{nk}(s)|\le\dot H_{\max}/\gamma_{\mathrm{ex}}$. Since an entrywise estimate does not by itself control the operator norm, we use in addition that the off-diagonal inverse-gap (Hilbert-transform-type) Schur multiplier is bounded up to an absolute constant, giving
\[
  \|M_Q(s)\|
  \le
  O\!\left(\frac{\dot H_{\max}}{\gamma_{\mathrm{ex}}}\right).
\]
Therefore
\[
  |\mathcal A^{(2)}(s)|
  \le
  O\!\left(
  \frac{\dot H_{\max}^3}
  {\Delta_{\min}^4\,\gamma_{\mathrm{ex}}}
  \right),
\]
and hence
\[
  \left|\int_0^1 \mathcal A^{(2)}(s)\,ds\right|
  \le
  O\!\left(
  \frac{\dot H_{\max}^3}
  {\Delta_{\min}^4\,\gamma_{\mathrm{ex}}}
  \right).
\]

When degeneracies or level crossings occur within the excited spectrum, the level-by-level quantities \(M_{nk}\) are no longer the most natural objects, since the relevant geometry is that of the excited manifold as a whole rather than of individual excited states. Accordingly, we expect that the term \(\mathcal A^{(2)}\) should be reformulated in terms of the Berry connection on the excited manifold, while the overall Berry-phase error scaling should remain unchanged as long as the ground state stays non-degenerate and separated from the excited manifold by a nonzero gap.

\paragraph{Term 3: the oscillatory boundary term.}
Using $|\sin\gamma_n|\le1$ and $|\cos(\omega_nT)|\le1$ in \cref{eq:Phi2T,eq:Bn_def}, and summing over $n$ by the Cauchy--Schwarz inequality and completeness as in the bounds of \cref{lem:phase_error}, this term is controlled by endpoint data alone:
\[
  |\Phi_2^{(T)}| \leq 
  \frac{\|\dot H(0)\|\,\|\dot H(1)\|}{\Delta(0)^4}.
\]

\paragraph{Combined bound.}
Summing over excited states by the Cauchy--Schwarz inequality and completeness of the eigenbasis, which removes any explicit dependence on the number of levels, and integrating over $s\in[0,1]$, we obtain
\begin{equation*}
  |\tilde{\theta}_B(T)-\theta_B |
  \leq \frac{C_2}{T^2}+O(T^{-3}),
\end{equation*}
\begin{equation*}
  C_2 = O\!\left(
  \frac{\dot{H}_{\max}\,\ddot{H}_{\max}}{\Delta_{\min}^4}
  + \frac{\dot H_{\max}^3}{\Delta_{\min}^4\,\gamma_{\wedge}}
  + \frac{\|\dot H(0)\|\,\|\dot H(1)\|}{\Delta(0)^4}
  \right).
\end{equation*}
We note that the contribution proportional to \(\gamma_{\wedge}^{-1}\) is the only part of the present bound that is sensitive to degeneracies or near-degeneracies within the excited spectrum. However, as discussed below, this term is eliminated by Richardson extrapolation, so the final complexity bound is unaffected by the detailed structure of the excited-state splittings.

\section{Proof of the Richardson cancellation theorem}
\label{app:richardson_details}

\begin{proof}[Proof of \cref{thm:richardson_m}]
Let $x_\ell=\alpha^{-2\ell}$. The Richardson weights are the Lagrange weights
\[
  w_{r,\ell}
  =
  \prod_{q\in Q_{r,\ell}}
  \frac{-x_q}{x_\ell-x_q},
  \quad
  Q_{r,\ell}:=\{0,\ldots,r\}\setminus\{\ell\}.
\]
They satisfy
\[
  \sum_{\ell=0}^{r}w_{r,\ell}=1,
  \quad
  \sum_{\ell=0}^{r}w_{r,\ell}\alpha^{-2j\ell}=0,
  \quad
  j=1,\ldots,r .
\]
Substituting the expansion
\[
  \tilde{\theta}_B(T)-\theta_B
  =
  \sum_{j=1}^{\lfloor K/2\rfloor}
  \frac{\Phi_{2j}}{T^{2j}}
  +
  \sum_{k=2}^{K}
  \frac{\Phi_k^{(T)}}{T^k}
  +
  O(T^{-K-1})
\]
into the Richardson extrapolant gives
\begin{align*}
  \tilde{\theta}_{B,R}^{(r)}(T)-\theta_B
  &=
  \sum_{j=1}^{\lfloor K/2\rfloor}
  \frac{\Phi_{2j}}{T^{2j}}
  \sum_{\ell=0}^{r}
  w_{r,\ell}\alpha^{-2j\ell}
  \nonumber\\
  &\quad
  +
  \sum_{k=2}^{K}
  \frac{\Phi_{k,r}^{(T)}}{T^k}
  +
  O(T^{-K-1}).
\end{align*}
where
\[
  \Phi_{k,r}^{(T)}
  :=
  \sum_{\ell=0}^{r}
  w_{r,\ell}\alpha^{-k\ell}
  \Phi_k^{(\alpha^\ell T)} .
\]
The Richardson conditions cancel the terms \(j=1,\ldots,r\).  Hence
\[
  \tilde{\theta}_{B,R}^{(r)}(T)-\theta_B
  =
  \sum_{k=2}^{K}
  \frac{\Phi_{k,r}^{(T)}}{T^k}
  +
  \sum_{j=r+1}^{\lfloor K/2\rfloor}
  \frac{C_{j,r}\Phi_{2j}}{T^{2j}}
  +
  O(T^{-K-1}),
\]
with
\[
  C_{j,r}
  :=
  \sum_{\ell=0}^{r}w_{r,\ell}\alpha^{-2j\ell}.
\]
This proves that the first possible remaining non-oscillatory term is of
order \(T^{-(2r+2)}\).

It remains to record the oscillatory structure.  If
\[
  \Phi_k^{(T)}
  =
  \sum_{\nu\in I_k}
  \left[
    a_{k,\nu}\cos(\Omega_{k,\nu}T)
    +
    b_{k,\nu}\sin(\Omega_{k,\nu}T)
  \right],
\]
then
\[
  \Phi_{k,r}^{(T)}
  =
  \sum_{\nu\in I_k}
  \sum_{\ell=0}^{r}
  w_{r,\ell}\alpha^{-k\ell}
  S_{k,\nu}^{(\ell)}(T),
\]
where
\[
  S_{k,\nu}^{(\ell)}(T)
  :=
  a_{k,\nu}\cos(\alpha^\ell\Omega_{k,\nu}T)
  +
  b_{k,\nu}\sin(\alpha^\ell\Omega_{k,\nu}T).
\]
In particular,
\[
  \Phi_{2,r}^{(T)}
  =
  \sum_{n\neq0}B_n
  \sum_{\ell=0}^{r}
  w_{r,\ell}\alpha^{-2\ell}
  \cos(\alpha^\ell\omega_nT).
\]
Therefore,
\[
  |\Phi_{2,r}^{(T)}|
  \leq
  L_r(\alpha)\sum_{n\neq 0}|B_n|,
  \qquad
  L_r(\alpha)
  :=
  \sum_{\ell=0}^{r}
  |w_{r,\ell}|\alpha^{-2\ell}.
\]
Using the endpoint estimate
\[
  \sum_{n\neq0}|B_n|
  \le
  \frac{\|\dot H(0)\|\,\|\dot H(1)\|}
       {\Delta(0)^4},
\]
we obtain
\[
  |\Phi_{2,r}^{(T)}|
  \le
  L_r(\alpha)
  \frac{\|\dot H(0)\|\,\|\dot H(1)\|}
       {\Delta(0)^4}.
\]
Finally, for fixed \(\alpha>1\), the weights obey
\[
  L_r(\alpha)
  \le
  (r+1)
  \left(
    \frac{\alpha^2}{\alpha^2-1}
  \right)^r .
\]
Indeed, from the Lagrange form,
\[
  |w_{r,\ell}|
  =
  \prod_{q\neq \ell}
  \frac{x_q}{|x_\ell-x_q|}.
\]
Each factor is bounded by \(\alpha^2/(\alpha^2-1)\), and
\(\alpha^{-2\ell}\le1\).  Hence the stated bound follows.

\end{proof}

\section{Properties of the Richardson factor}
\label{app:richardson_factor}

We collect elementary properties of the factor
\begin{equation*}
  L_r(\alpha)
  :=
  \sum_{\ell=0}^{r}|w_{r,\ell}|\,\alpha^{-2\ell},
  \qquad \alpha>1,
\end{equation*}
which appears in the bound for the leading oscillatory sector. This
quantity should not be confused with the statistical amplification factor
of the Richardson estimator. Rather, \(L_r(\alpha)\) is the effective
coefficient factor for the \(T^{-2}\) oscillatory contribution.

More generally, the \(k\)-th oscillatory sector is controlled by
\begin{equation*}
  L_{r,k}(\alpha)
  :=
  \sum_{\ell=0}^{r}|w_{r,\ell}|\,\alpha^{-k\ell}.
\end{equation*}
Indeed, the factor \(\alpha^{-k\ell}\) comes from
\[
  (\alpha^\ell T)^{-k}
  =
  \alpha^{-k\ell}T^{-k}.
\]
Thus
\begin{equation*}
  L_r(\alpha)=L_{r,2}(\alpha).
\end{equation*}
Moreover, for every \(k\ge2\),
\begin{equation*}
  L_{r,k}(\alpha)
  \le
  L_r(\alpha),
\end{equation*}
because \(\alpha^{-k\ell}\le\alpha^{-2\ell}\). Hence \(L_r(\alpha)\)
also gives a uniform upper bound for all higher oscillatory sectors
\(k\ge2\), although the exact factor for \(k>2\) is \(L_{r,k}(\alpha)\).

The Richardson weights imply the product representation
\begin{equation*}
  L_r(\alpha)
  =
  \alpha^{-2r}
  \frac{
    \prod_{j=0}^{r-1}(1+\alpha^{-2j})
  }{
    \prod_{j=1}^{r}(1-\alpha^{-2j})
  }.
\end{equation*}
Consequently,
\begin{equation*}
  \frac{L_{r+1}(\alpha)}{L_r(\alpha)}
  =
  \alpha^{-2}
  \frac{1+\alpha^{-2r}}{1-\alpha^{-(2r+2)}}.
\end{equation*}
Hence the step from $L_r$ to $L_{r+1}$ is decreasing whenever $\alpha^{-2}+2\alpha^{-(2r+2)}\le1$. For $r\ge1$ this holds whenever $\alpha\ge\sqrt2$, whereas the first step $L_1\le L_0=1$ requires the stronger condition $\alpha\ge\sqrt3$; for example, $\alpha=\sqrt2$ gives $L_1=2>1=L_0$. Consequently, $L_r(\alpha)$ is monotone decreasing in $r$ for all $r\ge0$ whenever $\alpha\ge\sqrt3$.
Thus, for the common choice \(\alpha=2\), the effective coefficient of the
leading oscillatory sector decreases with the Richardson order. For
example, for \(k=2\),
\begin{equation*}
  L_{1,2}(2)=\frac23,
  \qquad
  L_{2,2}(2)=\frac29,
  \qquad
  L_{3,2}(2)=\frac{34}{567}.
\end{equation*}
For higher oscillatory sectors the factor is even smaller. For instance,
for \(k=3\),
\begin{equation*}
  L_{1,3}(2)=\frac12,
  \quad
  L_{2,3}(2)=\frac{1}{10}.
\end{equation*}
On the other hand, the crude estimate
\begin{equation*}
  L_r(\alpha)
  \le
  (r+1)
  \left(
    \frac{\alpha^2}{\alpha^2-1}
  \right)^r
\end{equation*}
is only a worst-case upper bound. It does not capture the possible
decrease of \(L_r(\alpha)\) with \(r\).

\section{Analysis of Runtime Randomization}
\label{app:runtime_randomization}

In this appendix we prove \cref{thm:randomized_richardson}.  Recall that
\(\mu\) is supported on \([1-\lambda,1+\lambda]\) with density
\(\rho\le c_\mu/(2\lambda)\), that the support is partitioned into \(N\) intervals
\(I_1,\dots,I_N\) of equal width \(h:=2\lambda/N\) with \(p_j:=\mu(I_j)\), that
\(X_j\sim\mu(\,\cdot\mid I_j)\) are independent, and that
\[
  T_j=TX_j .
\]
For each randomized runtime, the \(r\)-fold Richardson estimator is
\[
  \tilde{\theta}_{B,R}^{(r)}(T_j)
  :=
  \sum_{\ell=0}^{r}
  w_{r,\ell}\tilde{\theta}_B(\alpha^\ell T_j).
\]
After Richardson extrapolation, the non-oscillatory terms
\(T_j^{-2},T_j^{-4},\ldots,T_j^{-2r}\) cancel.  Hence, for any fixed
\(K\ge 2r+2\),
\begin{equation}\label{eq:app_theta_single_event_random_richardson}
  \tilde{\theta}_{B,R}^{(r)}(T_j)-\theta_B
  =
  \sum_{k=2}^{K}
  \frac{\Phi_{k,r}^{(T_j)}}{T_j^k}
  +
  O(T_j^{-(2r+2)})
  +
  O(T_j^{-K-1}).
\end{equation}
Here the $O(T_j^{-(2r+2)})$ contribution is the non-oscillatory Richardson remainder. The randomized estimator is
\[
  \hat{\theta}_{B,R}^{(N,r)}(T)
  :=
  \sum_{j=1}^{N}p_j\,
  \tilde{\theta}_{B,R}^{(r)}(T_j).
\]
For any bounded \(g\), \(\sum_jp_j\,\mathbb E[g(X_j)]=\sum_j\int_{I_j}g\,d\mu
=\int g\,d\mu\), so every expectation in the bias estimate below is the
\(\mu\)-expectation of a single sample, exactly as for independent sampling; we write
\(\mathbb E\) for it.

\subsection{Bias estimate}

\begin{proof}[Proof of the bias bound in \cref{thm:randomized_richardson}]
We first estimate the mean oscillatory contribution.  For each \(k\), define
the weighted characteristic function
\begin{equation*}
  \chi_{\mu,k}(\xi)
  :=
  \int_{1-\lambda}^{1+\lambda}
  x^{-k}e^{i\xi x}\,d\mu(x).
\end{equation*}
The factor \(x^{-k}\) appears because \(T_j^{-k}=T^{-k}X_j^{-k}\).

Using the finite-mode form of the oscillatory coefficient,
\[
  \Phi_{k,r}^{(T)}
  =
  \sum_{\nu\in I_k}
  \left[
    a_{k,\nu}\cos(\Omega_{k,\nu}T)
    +
    b_{k,\nu}\sin(\Omega_{k,\nu}T)
  \right],
\]
we obtain
\begin{align*}
  \mathbb E\left[
    \frac{\Phi_{k,r}^{(T_j)}}{T_j^k}
  \right]
  &=
  \frac{1}{T^k}
  \sum_{\nu\in I_k}
  \Bigg[
    a_{k,\nu}
    \int_{1-\lambda}^{1+\lambda}
    x^{-k}\cos(\Omega_{k,\nu}Tx)\,d\mu(x)
    \\
  &\qquad\qquad
    +
    b_{k,\nu}
    \int_{1-\lambda}^{1+\lambda}
    x^{-k}\sin(\Omega_{k,\nu}Tx)\,d\mu(x)
  \Bigg].
\end{align*}
Therefore,
\begin{equation*}
  \left|
  \mathbb E\left[
    \frac{\Phi_{k,r}^{(T_j)}}{T_j^k}
  \right]
  \right|
  \le
  \frac{1}{T^k}
  \sum_{\nu\in I_k}
  \bigl(|a_{k,\nu}|+|b_{k,\nu}|\bigr)
  |\chi_{\mu,k}(\Omega_{k,\nu}T)|.
\end{equation*}

Assume that, for some \(M>0\),
\begin{equation*}
  |\chi_{\mu,k}(\xi)|
  \le
  C_{\mu,K}(1+|\xi|)^{-M},
  \qquad
  2\le k\le K .
\end{equation*}
Since the frequencies obey \(\Omega_{k,\nu}\ge\Delta_{\min}\), the total
oscillatory bias satisfies
\begin{align*}
  \left|
  \mathbb E\left[
    \sum_{k=2}^{K}
    \frac{\Phi_{k,r}^{(T_j)}}{T_j^k}
  \right]
  \right|
  &=
  O\left(
    \sum_{k=2}^{K}
    \frac{1}{T^k}
    \frac{1}{(\Delta_{\min}T)^M}
  \right)
  \\
  &=
  O\left(
    \frac{1}{\Delta_{\min}^M T^{M+2}}
  \right).
\end{align*}
Thus the leading randomized oscillatory bias comes from the \(k=2\) sector.

For this leading sector, the boundary estimate gives
\begin{align*}
  \sum_{\nu\in I_2}
  \bigl(|a_{2,\nu}|+|b_{2,\nu}|\bigr)
  &\le
  L_r(\alpha)\sum_{n\neq0}|B_n|
  \\
  &\le
  L_r(\alpha)
  \frac{\|\dot H(0)\|\,\|\dot H(1)\|}
       {\Delta(0)^4},
\end{align*}
where
\[
  L_r(\alpha)
  :=
  \sum_{\ell=0}^{r}
  |w_{r,\ell}|\alpha^{-2\ell}.
\]
Consequently,
\begin{equation*}
  \left|
  \mathbb E\left[
    \frac{\Phi_{2,r}^{(T_j)}}{T_j^2}
  \right]
  \right|
  \le
  C_{\mu,K}\,L_r(\alpha)
  \frac{\|\dot H(0)\|\,\|\dot H(1)\|}
       {\Delta(0)^4\Delta_{\min}^M T^{M+2}} .
\end{equation*}
For a fixed Hamiltonian path, the higher oscillatory sectors
\(k\ge3\) contribute \(O(\Delta_{\min}^{-M}T^{-M-3})\), which is
absorbed into \(O(T^{-(2r+2)})\) under \(M\ge2r-1\).
Combining this with the non-oscillatory Richardson remainder in
\cref{eq:app_theta_single_event_random_richardson} gives
\begin{align*}
  \mathbb E[
    \hat{\theta}_{B,R}^{(N,r)}(T)
  ]-\theta_B
  &=
  O\left(
    \frac{\|\dot H(0)\|\,\|\dot H(1)\|}
         {\Delta(0)^4\Delta_{\min}^{M}T^{M+2}}
  \right)
  \\
  &\quad
  +
  O(T^{-(2r+2)})
  +
  O(T^{-K-1}).
\end{align*}

\end{proof}

\subsection{Variance estimate}\label{app:runtime_variance}

\begin{proof}[Proof of the variance bound in \cref{thm:randomized_richardson}]

We now estimate the finite-sample fluctuation of the randomized estimator.
Let
\[
  Y(T_j)
  :=
  \sum_{k=2}^{K}
  \frac{\Phi_{k,r}^{(T_j)}}{T_j^k},
  \qquad
  \overline Y_N
  :=
  \sum_{j=1}^{N}p_j\,Y(T_j).
\]

Define
\[
  A_k
  :=
  \sum_{\nu\in I_k}
  \bigl(|a_{k,\nu}|+|b_{k,\nu}|\bigr).
\]
Write \(f(x):=Y(Tx)\), so that \(\overline Y_N=\sum_jp_jf(X_j)\) with the
\(X_j\) independent and \(X_j\in I_j\). Hence
\[
  \operatorname{Var}(\overline Y_N)
  =
  \sum_{j=1}^{N}p_j^{2}\operatorname{Var}\bigl(f(X_j)\bigr).
\]
On the interval \(I_j\), of width \(h\), the function \(f\) varies by at most
\(h\sup|f'|\), so \(\operatorname{Var}(f(X_j))\le(h\sup|f'|)^2/4\). Moreover
\(p_j\le h\,c_\mu/(2\lambda)=c_\mu/N\) and \(\sum_jp_j=1\). Therefore
\[
  \operatorname{Var}(\overline Y_N)
  \le
  \frac{c_\mu}{N}\cdot\frac{h^{2}\sup|f'|^{2}}{4}
  =
  \frac{c_\mu\lambda^{2}\sup|f'|^{2}}{N^{3}},
\]
that is,
\[
  \sqrt{\operatorname{Var}(\overline Y_N)}
  \le
  \frac{\sqrt{c_\mu}\,\lambda\,\sup|f'|}{N^{3/2}}.
\]
It remains to bound the slope. By the explicit form of \(\Phi_{2,r}^{(T)}\) in
\cref{app:richardson_details},
\begin{align*}
  &\frac{d}{dx}\,\frac{\Phi_{2,r}^{(Tx)}}{T^{2}x^{2}}
  =
  -\sum_{n\neq0}B_n\sum_{\ell=0}^{r}w_{r,\ell}\alpha^{-2\ell}\\
  &\qquad\times
  \Bigl[
    \frac{\alpha^{\ell}\omega_nT\sin(\alpha^{\ell}\omega_nTx)}{T^{2}x^{2}}
    +
    \frac{2\cos(\alpha^{\ell}\omega_nTx)}{T^{2}x^{3}}
  \Bigr],
\end{align*}
so that, using \(\omega_n\le2H_{\max}\),
\(L_{r,1}(\alpha)=\sum_\ell|w_{r,\ell}|\alpha^{-\ell}\), \(\bar B:=\sum_{n\neq0}|B_n|\),
and \(x\ge1-\lambda\),
\[
  \sup_x\left|\frac{d}{dx}\,\frac{\Phi_{2,r}^{(Tx)}}{T^{2}x^{2}}\right|
  \le
  \bar B\left[
    \frac{2H_{\max}L_{r,1}(\alpha)}{(1-\lambda)^{2}\,T}
    +
    \frac{2L_r(\alpha)}{(1-\lambda)^{3}\,T^{2}}
  \right].
\]
The sectors \(k\ge3\) are treated identically: each frequency
\(\Omega_{k,\nu}\) is a finite sum of the \(\omega_n\), so the slope of
\(\Phi_{k,r}^{(Tx)}/(Tx)^k\) is \(O(A_kT^{-k+1})=O(T^{-2})\) for fixed \(K\), with
constants depending on \(\lambda,K,\alpha\). Hence
\begin{align*}
  \sqrt{\operatorname{Var}(\overline Y_N)}
  &\le
  \frac{\sqrt{c_\mu}\,\lambda\,\bar B}{N^{3/2}}
  \left[
    \frac{2H_{\max}L_{r,1}(\alpha)}{(1-\lambda)^{2}\,T}
    +
    \frac{2L_r(\alpha)}{(1-\lambda)^{3}\,T^{2}}
  \right]\\
  &\quad+
  O\!\left(\frac{1}{T^{2}N^{3/2}}\right).
\end{align*}
The non-oscillatory Richardson remainder has an \(x\)-derivative
of order \(T^{-(2r+2)}\). For the order-\(K\) APT remainder \(R_K\),
retaining two additional APT orders and differentiating Duhamel's formula
gives \(\sup_x|\partial_x R_K(Tx)|=O(T^{-K})\).
The same stratified variance bound and the triangle inequality for
standard deviations therefore extend the estimate to the full estimator,
with both contributions absorbed into \(O(T^{-2}N^{-3/2})\).

Using the endpoint estimate
\(\bar B\le\|\dot H(0)\|\,\|\dot H(1)\|/\Delta(0)^4\) and
absorbing the factor \(2\) gives the variance statement in
\cref{thm:randomized_richardson}. Equivalently, the randomized finite-sample
fluctuation is of this order in probability. For comparison, independent sampling
\(X_j\sim\mu\) gives only \(\operatorname{Var}(\overline Y_N)=\operatorname{Var}(Y)/N
=O(\bar B^2T^{-4}/N){}+O(T^{-6}/N)\): stratification trades the amplitude \(\bar BT^{-2}\) of the
oscillation for its slope over a stratum, \((2\lambda/N)\,H_{\max}\bar BT^{-1}\).

\end{proof}

\subsection{Imperfect initial state}\label{app:prep_randomization}

\begin{proof}[Proof of \cref{thm:prep_supp}]
We first bound the randomized bias, separating the zeroth-order preparation term from the higher-order terms. We then estimate the finite-sample fluctuation using the stratified variance bound.

  The forward--reverse combination replaces the coefficients \(a_\nu,b_\nu\) of
  \cref{eq:prep_k0} by combinations of the same frequencies whose moduli obey the same
  bounds, so we write the symmetrized term in the form \cref{eq:prep_k0} without loss.

  \noindent\textbf{Zeroth-order bias.} Since \(\varphi_0^{\mathrm{prep}(T)}\)
  carries no power of the runtime, at \(T_j=TX_j\) it acquires no factor \(X_j^{-k}\),
  and the Richardson combination is
  \(\sum_{\ell}w_{r,\ell}\varphi_0^{\mathrm{prep}(\alpha^\ell TX_j)}\), with no inverse-runtime factor \(\alpha^{-k\ell}\), since \(k=0\). The expectation of the weighted stratified average equals the expectation over a single \(X\sim\mu\). Taking this expectation replaces
  \(\cos(\Omega_\nu\alpha^\ell TX)\) and \(\sin(\Omega_\nu\alpha^\ell TX)\) by the
  real and imaginary parts of \(\chi_{\mu,0}(\Omega_\nu\alpha^\ell T)\); the interchange
  of \(\mathbb E\) and \(\sum_\nu\) is justified by the absolute convergence in
  \cref{lem:prep_bias}. Hence
  \begin{align*}
    &\Bigl|\mathbb E\Bigl[\sum_{\ell}w_{r,\ell}
    \varphi_0^{\mathrm{prep}(\alpha^\ell TX)}\Bigr]\Bigr| \\
    &\quad\le\sum_{\ell}|w_{r,\ell}|\sum_\nu\bigl(|a_\nu|+|b_\nu|\bigr)
    \bigl|\chi_{\mu,0}(\Omega_\nu\alpha^\ell T)\bigr| \\
    &\quad=O\!\left(
      \frac{p_{\mathrm{prep}}\,C_{\mu,K}\,L_{r,0}(\alpha)}{(\Delta_{\min}T)^{M}}
    \right),
  \end{align*}
  using \(\Omega_\nu\alpha^\ell T\ge\Delta_{\min}T\), the hypothesis on
  \(\chi_{\mu,0}\), \(\sum_\nu(|a_\nu|+|b_\nu|)=O(p_{\mathrm{prep}})\), and
  \(L_{r,0}(\alpha)=\sum_\ell|w_{r,\ell}|\).

  \noindent\textbf{Higher-order bias.} For \(k\ge1\), \(\varphi_k^{\mathrm{prep}(T)}/T^k\) at runtime \(\alpha^\ell TX\) contributes \(\alpha^{-k\ell}T^{-k}X^{-k}\cos(\Omega\alpha^\ell TX)\) (and the sine analogue); as in the proof of \cref{thm:randomized_richardson}, its expectation is \(\alpha^{-k\ell}T^{-k}\Re\chi_{\mu,k}(\Omega\alpha^\ell T)\), bounded by the hypothesis on \(\chi_{\mu,k}\), and the interchange with \(\sum_\nu\) is justified by the absolute convergence at each order in \cref{lem:prep_bias}. Since these coefficients are \(O(p_{\mathrm{prep}}^{1/2})\), the total contribution is
  \[
    O\!\left(p_{\mathrm{prep}}^{1/2}\Delta_{\min}^{-M}T^{-M-1}\right).
  \]
  This term is retained because \(p_{\mathrm{prep}}\) may decrease with \(T\). The
  non-oscillatory \(\varphi_k^{\mathrm{prep}}\) are removed by \cref{thm:cancellation}
  for odd \(k\) and by \cref{thm:richardson_m} for even \(k\le2r\), so the first that
  survives is of order \(T^{-(2r+2)}\); it carries one factor \(c_n\) with \(n\ge1\) and
  is therefore \(O(p_{\mathrm{prep}}^{1/2}T^{-(2r+2)})\). Combining these three contributions gives \cref{eq:prep_supp}.

  \noindent\textbf{Finite-sample fluctuation.}
  For the zeroth-order contribution, define
  \[
    f(x):=\sum_{\ell=0}^{r}w_{r,\ell}
    \varphi_0^{\mathrm{prep}(\alpha^\ell Tx)},
    \qquad
    \overline f_N:=\sum_{j=1}^{N}p_j f(X_j).
  \]
  By the stratified variance bound of \cref{app:runtime_variance},
  \[
    \sqrt{\operatorname{Var}(\overline f_N)}
    \le \frac{\sqrt{c_\mu}\,\lambda}{N^{3/2}}
    \sup_x|f'(x)|.
  \]
  It remains to bound the slope of \(f\). Differentiating the Fourier series gives
  \begin{align*}
    |f'(x)|
    &\le T\sum_{\ell=0}^{r}|w_{r,\ell}|\alpha^\ell
      \sum_\nu\Omega_\nu(|a_\nu|+|b_\nu|)\\
    &\le \alpha^r L_{r,0}(\alpha)T
      \sum_\nu\Omega_\nu(|a_\nu|+|b_\nu|).
  \end{align*}
  The frequency-weighted sum is finite: in \cref{eq:prep_k0_series}, a term of degree \(q\) has frequency at most \(2qH_{\max}\), which cancels the factor \(1/q\) in the logarithmic series. Consequently,
  \begin{align*}
    \sum_\nu\Omega_\nu(|a_\nu|+|b_\nu|)
    &\le 2\sqrt2 H_{\max}
      \sum_{q\ge1}\left(\sum_{n\ge1}r_n\right)^q\\
    &=\frac{2\sqrt2 H_{\max}p_{\mathrm{prep}}}
           {1-2p_{\mathrm{prep}}}\\
    &=O(H_{\max}p_{\mathrm{prep}}).
  \end{align*}
  The factor \(\sqrt2\) accounts for splitting each shifted sine into cosine and sine components. This bound also justifies termwise differentiation. Substituting it into the slope and variance bounds yields
  \[
    \sqrt{\operatorname{Var}(\overline f_N)}
    =O\!\left(
      \frac{\sqrt{c_\mu}\,\lambda\,\alpha^rL_{r,0}(\alpha)
      H_{\max}T p_{\mathrm{prep}}}{N^{3/2}}
    \right).
  \]

  For \(k\ge1\), the frequency-weighted Fourier coefficients are
  absolutely summable, and the \(x\)-derivative of each oscillatory sector
  is \(O(p_{\mathrm{prep}}^{1/2}T^{1-k})\).
  The non-oscillatory terms and the APT remainder are controlled as in
  \cref{app:runtime_variance}, with the additional factor
  \(p_{\mathrm{prep}}^{1/2}\) from subtracting the exact-preparation result.
  The same stratified variance bound therefore adds
  \(O(p_{\mathrm{prep}}^{1/2}N^{-3/2})\) to the standard deviation.
  Combining gives \cref{eq:prep_var}.
\end{proof}


\begin{thebibliography}{73}%
  \makeatletter
  \providecommand \@ifxundefined [1]{%
   \@ifx{#1\undefined}
  }%
  \providecommand \@ifnum [1]{%
   \ifnum #1\expandafter \@firstoftwo
   \else \expandafter \@secondoftwo
   \fi
  }%
  \providecommand \@ifx [1]{%
   \ifx #1\expandafter \@firstoftwo
   \else \expandafter \@secondoftwo
   \fi
  }%
  \providecommand \natexlab [1]{#1}%
  \providecommand \enquote  [1]{``#1''}%
  \providecommand \bibnamefont  [1]{#1}%
  \providecommand \bibfnamefont [1]{#1}%
  \providecommand \citenamefont [1]{#1}%
  \providecommand \href@noop [0]{\@secondoftwo}%
  \providecommand \href [0]{\begingroup \@sanitize@url \@href}%
  \providecommand \@href[1]{\@@startlink{#1}\@@href}%
  \providecommand \@@href[1]{\endgroup#1\@@endlink}%
  \providecommand \@sanitize@url [0]{\catcode `\\12\catcode `\$12\catcode `\&12\catcode `\#12\catcode `\^12\catcode `\_12\catcode `\%12\relax}%
  \providecommand \@@startlink[1]{}%
  \providecommand \@@endlink[0]{}%
  \providecommand \url  [0]{\begingroup\@sanitize@url \@url }%
  \providecommand \@url [1]{\endgroup\@href {#1}{\urlprefix }}%
  \providecommand \urlprefix  [0]{URL }%
  \providecommand \Eprint [0]{\href }%
  \providecommand \doibase [0]{https://doi.org/}%
  \providecommand \selectlanguage [0]{\@gobble}%
  \providecommand \bibinfo  [0]{\@secondoftwo}%
  \providecommand \bibfield  [0]{\@secondoftwo}%
  \providecommand \translation [1]{[#1]}%
  \providecommand \BibitemOpen [0]{}%
  \providecommand \bibitemStop [0]{}%
  \providecommand \bibitemNoStop [0]{.\EOS\space}%
  \providecommand \EOS [0]{\spacefactor3000\relax}%
  \providecommand \BibitemShut  [1]{\csname bibitem#1\endcsname}%
  \let\auto@bib@innerbib\@empty
  %</preamble>
  \bibitem [{\citenamefont {Simon}(1983)}]{Simon1983}%
    \BibitemOpen
    \bibfield  {author} {\bibinfo {author} {\bibfnamefont {B.}~\bibnamefont {Simon}},\ }\bibfield  {title} {\bibinfo {title} {Holonomy, the quantum adiabatic theorem, and {Berry}'s phase},\ }\href {https://doi.org/10.1103/PhysRevLett.51.2167} {\bibfield  {journal} {\bibinfo  {journal} {Phys. Rev. Lett.}\ }\textbf {\bibinfo {volume} {51}},\ \bibinfo {pages} {2167} (\bibinfo {year} {1983})}\BibitemShut {NoStop}%
  \bibitem [{\citenamefont {Berry}(1984)}]{Berry1984}%
    \BibitemOpen
    \bibfield  {author} {\bibinfo {author} {\bibfnamefont {M.~V.}\ \bibnamefont {Berry}},\ }\bibfield  {title} {\bibinfo {title} {Quantal phase factors accompanying adiabatic changes},\ }\href {https://doi.org/10.1098/rspa.1984.0023} {\bibfield  {journal} {\bibinfo  {journal} {Proc. R. Soc. London, Ser. A}\ }\textbf {\bibinfo {volume} {392}},\ \bibinfo {pages} {45} (\bibinfo {year} {1984})}\BibitemShut {NoStop}%
  \bibitem [{\citenamefont {Thouless}\ \emph {et~al.}(1982)\citenamefont {Thouless}, \citenamefont {Kohmoto}, \citenamefont {Nightingale},\ and\ \citenamefont {den Nijs}}]{TKNN1982}%
    \BibitemOpen
    \bibfield  {author} {\bibinfo {author} {\bibfnamefont {D.~J.}\ \bibnamefont {Thouless}}, \bibinfo {author} {\bibfnamefont {M.}~\bibnamefont {Kohmoto}}, \bibinfo {author} {\bibfnamefont {M.~P.}\ \bibnamefont {Nightingale}},\ and\ \bibinfo {author} {\bibfnamefont {M.}~\bibnamefont {den Nijs}},\ }\bibfield  {title} {\bibinfo {title} {Quantized {Hall} conductance in a two-dimensional periodic potential},\ }\href {https://doi.org/10.1103/PhysRevLett.49.405} {\bibfield  {journal} {\bibinfo  {journal} {Phys. Rev. Lett.}\ }\textbf {\bibinfo {volume} {49}},\ \bibinfo {pages} {405} (\bibinfo {year} {1982})}\BibitemShut {NoStop}%
  \bibitem [{\citenamefont {Hasan}\ and\ \citenamefont {Kane}(2010)}]{HasanKane2010}%
    \BibitemOpen
    \bibfield  {author} {\bibinfo {author} {\bibfnamefont {M.~Z.}\ \bibnamefont {Hasan}}\ and\ \bibinfo {author} {\bibfnamefont {C.~L.}\ \bibnamefont {Kane}},\ }\bibfield  {title} {\bibinfo {title} {Colloquium: Topological insulators},\ }\href {https://doi.org/10.1103/RevModPhys.82.3045} {\bibfield  {journal} {\bibinfo  {journal} {Rev. Mod. Phys.}\ }\textbf {\bibinfo {volume} {82}},\ \bibinfo {pages} {3045} (\bibinfo {year} {2010})}\BibitemShut {NoStop}%
  \bibitem [{\citenamefont {Qi}\ and\ \citenamefont {Zhang}(2011)}]{QiZhang2011}%
    \BibitemOpen
    \bibfield  {author} {\bibinfo {author} {\bibfnamefont {X.-L.}\ \bibnamefont {Qi}}\ and\ \bibinfo {author} {\bibfnamefont {S.-C.}\ \bibnamefont {Zhang}},\ }\bibfield  {title} {\bibinfo {title} {Topological insulators and superconductors},\ }\href {https://doi.org/10.1103/RevModPhys.83.1057} {\bibfield  {journal} {\bibinfo  {journal} {Rev. Mod. Phys.}\ }\textbf {\bibinfo {volume} {83}},\ \bibinfo {pages} {1057} (\bibinfo {year} {2011})}\BibitemShut {NoStop}%
  \bibitem [{\citenamefont {King-Smith}\ and\ \citenamefont {Vanderbilt}(1993)}]{KingSmith1993}%
    \BibitemOpen
    \bibfield  {author} {\bibinfo {author} {\bibfnamefont {R.~D.}\ \bibnamefont {King-Smith}}\ and\ \bibinfo {author} {\bibfnamefont {D.}~\bibnamefont {Vanderbilt}},\ }\bibfield  {title} {\bibinfo {title} {Theory of polarization of crystalline solids},\ }\href {https://doi.org/10.1103/PhysRevB.47.1651} {\bibfield  {journal} {\bibinfo  {journal} {Phys. Rev. B}\ }\textbf {\bibinfo {volume} {47}},\ \bibinfo {pages} {1651} (\bibinfo {year} {1993})}\BibitemShut {NoStop}%
  \bibitem [{\citenamefont {Resta}(2000)}]{Resta2000}%
    \BibitemOpen
    \bibfield  {author} {\bibinfo {author} {\bibfnamefont {R.}~\bibnamefont {Resta}},\ }\bibfield  {title} {\bibinfo {title} {Manifestations of {Berry}'s phase in molecules and condensed matter},\ }\href {https://doi.org/10.1088/0953-8984/12/9/201} {\bibfield  {journal} {\bibinfo  {journal} {J. Phys.: Condens. Matter}\ }\textbf {\bibinfo {volume} {12}},\ \bibinfo {pages} {R107} (\bibinfo {year} {2000})}\BibitemShut {NoStop}%
  \bibitem [{\citenamefont {Vanderbilt}(2018)}]{Vanderbilt2018}%
    \BibitemOpen
    \bibfield  {author} {\bibinfo {author} {\bibfnamefont {D.}~\bibnamefont {Vanderbilt}},\ }\href {https://doi.org/10.1017/9781316662205} {\emph {\bibinfo {title} {Berry Phases in Electronic Structure Theory}}}\ (\bibinfo  {publisher} {Cambridge University Press},\ \bibinfo {year} {2018})\BibitemShut {NoStop}%
  \bibitem [{\citenamefont {Nagaosa}\ \emph {et~al.}(2010)\citenamefont {Nagaosa}, \citenamefont {Sinova}, \citenamefont {Onoda}, \citenamefont {MacDonald},\ and\ \citenamefont {Ong}}]{Nagaosa2010}%
    \BibitemOpen
    \bibfield  {author} {\bibinfo {author} {\bibfnamefont {N.}~\bibnamefont {Nagaosa}}, \bibinfo {author} {\bibfnamefont {J.}~\bibnamefont {Sinova}}, \bibinfo {author} {\bibfnamefont {S.}~\bibnamefont {Onoda}}, \bibinfo {author} {\bibfnamefont {A.~H.}\ \bibnamefont {MacDonald}},\ and\ \bibinfo {author} {\bibfnamefont {N.~P.}\ \bibnamefont {Ong}},\ }\bibfield  {title} {\bibinfo {title} {Anomalous {Hall} effect},\ }\href {https://doi.org/10.1103/RevModPhys.82.1539} {\bibfield  {journal} {\bibinfo  {journal} {Rev. Mod. Phys.}\ }\textbf {\bibinfo {volume} {82}},\ \bibinfo {pages} {1539} (\bibinfo {year} {2010})}\BibitemShut {NoStop}%
  \bibitem [{\citenamefont {Xiao}\ \emph {et~al.}(2010)\citenamefont {Xiao}, \citenamefont {Chang},\ and\ \citenamefont {Niu}}]{XiaoChangNiu2010}%
    \BibitemOpen
    \bibfield  {author} {\bibinfo {author} {\bibfnamefont {D.}~\bibnamefont {Xiao}}, \bibinfo {author} {\bibfnamefont {M.-C.}\ \bibnamefont {Chang}},\ and\ \bibinfo {author} {\bibfnamefont {Q.}~\bibnamefont {Niu}},\ }\bibfield  {title} {\bibinfo {title} {Berry phase effects on electronic properties},\ }\href {https://doi.org/10.1103/RevModPhys.82.1959} {\bibfield  {journal} {\bibinfo  {journal} {Rev. Mod. Phys.}\ }\textbf {\bibinfo {volume} {82}},\ \bibinfo {pages} {1959} (\bibinfo {year} {2010})}\BibitemShut {NoStop}%
  \bibitem [{\citenamefont {Peotta}\ and\ \citenamefont {T\"orm\"a}(2015)}]{Peotta2015}%
    \BibitemOpen
    \bibfield  {author} {\bibinfo {author} {\bibfnamefont {S.}~\bibnamefont {Peotta}}\ and\ \bibinfo {author} {\bibfnamefont {P.}~\bibnamefont {T\"orm\"a}},\ }\bibfield  {title} {\bibinfo {title} {Superfluidity in topologically nontrivial flat bands},\ }\href {https://doi.org/10.1038/ncomms9944} {\bibfield  {journal} {\bibinfo  {journal} {Nat. Commun.}\ }\textbf {\bibinfo {volume} {6}},\ \bibinfo {pages} {8944} (\bibinfo {year} {2015})}\BibitemShut {NoStop}%
  \bibitem [{\citenamefont {T\"orm\"a}\ \emph {et~al.}(2022)\citenamefont {T\"orm\"a}, \citenamefont {Peotta},\ and\ \citenamefont {Bernevig}}]{Torma2022}%
    \BibitemOpen
    \bibfield  {author} {\bibinfo {author} {\bibfnamefont {P.}~\bibnamefont {T\"orm\"a}}, \bibinfo {author} {\bibfnamefont {S.}~\bibnamefont {Peotta}},\ and\ \bibinfo {author} {\bibfnamefont {B.~A.}\ \bibnamefont {Bernevig}},\ }\bibfield  {title} {\bibinfo {title} {Superconductivity, superfluidity and quantum geometry in twisted multilayer systems},\ }\href {https://doi.org/10.1038/s42254-022-00466-y} {\bibfield  {journal} {\bibinfo  {journal} {Nat. Rev. Phys.}\ }\textbf {\bibinfo {volume} {4}},\ \bibinfo {pages} {528} (\bibinfo {year} {2022})}\BibitemShut {NoStop}%
  \bibitem [{\citenamefont {Ahn}\ \emph {et~al.}(2022)\citenamefont {Ahn}, \citenamefont {Guo}, \citenamefont {Nagaosa},\ and\ \citenamefont {Vishwanath}}]{Ahn2022}%
    \BibitemOpen
    \bibfield  {author} {\bibinfo {author} {\bibfnamefont {J.}~\bibnamefont {Ahn}}, \bibinfo {author} {\bibfnamefont {G.-Y.}\ \bibnamefont {Guo}}, \bibinfo {author} {\bibfnamefont {N.}~\bibnamefont {Nagaosa}},\ and\ \bibinfo {author} {\bibfnamefont {A.}~\bibnamefont {Vishwanath}},\ }\bibfield  {title} {\bibinfo {title} {Riemannian geometry of resonant optical responses},\ }\href {https://doi.org/10.1038/s41567-021-01465-z} {\bibfield  {journal} {\bibinfo  {journal} {Nat. Phys.}\ }\textbf {\bibinfo {volume} {18}},\ \bibinfo {pages} {290} (\bibinfo {year} {2022})}\BibitemShut {NoStop}%
  \bibitem [{\citenamefont {Roy}(2014)}]{Roy2014}%
    \BibitemOpen
    \bibfield  {author} {\bibinfo {author} {\bibfnamefont {R.}~\bibnamefont {Roy}},\ }\bibfield  {title} {\bibinfo {title} {Band geometry of fractional topological insulators},\ }\href {https://doi.org/10.1103/PhysRevB.90.165139} {\bibfield  {journal} {\bibinfo  {journal} {Phys. Rev. B}\ }\textbf {\bibinfo {volume} {90}},\ \bibinfo {pages} {165139} (\bibinfo {year} {2014})}\BibitemShut {NoStop}%
  \bibitem [{\citenamefont {Liu}\ \emph {et~al.}(2025)\citenamefont {Liu}, \citenamefont {Qiang}, \citenamefont {Lu},\ and\ \citenamefont {Xie}}]{Liu2025}%
    \BibitemOpen
    \bibfield  {author} {\bibinfo {author} {\bibfnamefont {T.}~\bibnamefont {Liu}}, \bibinfo {author} {\bibfnamefont {X.-B.}\ \bibnamefont {Qiang}}, \bibinfo {author} {\bibfnamefont {H.-Z.}\ \bibnamefont {Lu}},\ and\ \bibinfo {author} {\bibfnamefont {X.~C.}\ \bibnamefont {Xie}},\ }\bibfield  {title} {\bibinfo {title} {Quantum geometry in condensed matter},\ }\href {https://doi.org/10.1093/nsr/nwae334} {\bibfield  {journal} {\bibinfo  {journal} {Natl. Sci. Rev.}\ }\textbf {\bibinfo {volume} {12}},\ \bibinfo {pages} {nwae334} (\bibinfo {year} {2025})}\BibitemShut {NoStop}%
  \bibitem [{\citenamefont {Mart{\'i}n-Mart{\'i}nez}\ \emph {et~al.}(2013)\citenamefont {Mart{\'i}n-Mart{\'i}nez}, \citenamefont {Dragan}, \citenamefont {Mann},\ and\ \citenamefont {Fuentes}}]{MartinMartinez2013}%
    \BibitemOpen
    \bibfield  {author} {\bibinfo {author} {\bibfnamefont {E.}~\bibnamefont {Mart{\'i}n-Mart{\'i}nez}}, \bibinfo {author} {\bibfnamefont {A.}~\bibnamefont {Dragan}}, \bibinfo {author} {\bibfnamefont {R.~B.}\ \bibnamefont {Mann}},\ and\ \bibinfo {author} {\bibfnamefont {I.}~\bibnamefont {Fuentes}},\ }\bibfield  {title} {\bibinfo {title} {Berry phase quantum thermometer},\ }\href {https://doi.org/10.1088/1367-2630/15/5/053036} {\bibfield  {journal} {\bibinfo  {journal} {New J. Phys.}\ }\textbf {\bibinfo {volume} {15}},\ \bibinfo {pages} {053036} (\bibinfo {year} {2013})}\BibitemShut {NoStop}%
  \bibitem [{\citenamefont {Ledbetter}\ \emph {et~al.}(2012)\citenamefont {Ledbetter}, \citenamefont {Jensen}, \citenamefont {Fischer}, \citenamefont {Jarmola},\ and\ \citenamefont {Budker}}]{Ledbetter2012}%
    \BibitemOpen
    \bibfield  {author} {\bibinfo {author} {\bibfnamefont {M.~P.}\ \bibnamefont {Ledbetter}}, \bibinfo {author} {\bibfnamefont {K.}~\bibnamefont {Jensen}}, \bibinfo {author} {\bibfnamefont {R.}~\bibnamefont {Fischer}}, \bibinfo {author} {\bibfnamefont {A.}~\bibnamefont {Jarmola}},\ and\ \bibinfo {author} {\bibfnamefont {D.}~\bibnamefont {Budker}},\ }\bibfield  {title} {\bibinfo {title} {Gyroscopes based on nitrogen-vacancy centers in diamond},\ }\href {https://doi.org/10.1103/PhysRevA.86.052116} {\bibfield  {journal} {\bibinfo  {journal} {Phys. Rev. A}\ }\textbf {\bibinfo {volume} {86}},\ \bibinfo {pages} {052116} (\bibinfo {year} {2012})}\BibitemShut {NoStop}%
  \bibitem [{\citenamefont {Arai}\ \emph {et~al.}(2018)\citenamefont {Arai}, \citenamefont {Lee}, \citenamefont {Belthangady}, \citenamefont {Glenn}, \citenamefont {Zhang},\ and\ \citenamefont {Walsworth}}]{Arai2018}%
    \BibitemOpen
    \bibfield  {author} {\bibinfo {author} {\bibfnamefont {K.}~\bibnamefont {Arai}}, \bibinfo {author} {\bibfnamefont {J.}~\bibnamefont {Lee}}, \bibinfo {author} {\bibfnamefont {C.}~\bibnamefont {Belthangady}}, \bibinfo {author} {\bibfnamefont {D.~R.}\ \bibnamefont {Glenn}}, \bibinfo {author} {\bibfnamefont {H.}~\bibnamefont {Zhang}},\ and\ \bibinfo {author} {\bibfnamefont {R.~L.}\ \bibnamefont {Walsworth}},\ }\bibfield  {title} {\bibinfo {title} {Geometric phase magnetometry using a solid-state spin},\ }\href {https://doi.org/10.1038/s41467-018-07489-z} {\bibfield  {journal} {\bibinfo  {journal} {Nat. Commun.}\ }\textbf {\bibinfo {volume} {9}},\ \bibinfo {pages} {4996} (\bibinfo {year} {2018})}\BibitemShut {NoStop}%
  \bibitem [{\citenamefont {Johnsson}\ \emph {et~al.}(2020)\citenamefont {Johnsson}, \citenamefont {Mukty}, \citenamefont {Burgarth}, \citenamefont {Volz},\ and\ \citenamefont {Brennen}}]{Johnsson2020}%
    \BibitemOpen
    \bibfield  {author} {\bibinfo {author} {\bibfnamefont {M.~T.}\ \bibnamefont {Johnsson}}, \bibinfo {author} {\bibfnamefont {N.~R.}\ \bibnamefont {Mukty}}, \bibinfo {author} {\bibfnamefont {D.}~\bibnamefont {Burgarth}}, \bibinfo {author} {\bibfnamefont {T.}~\bibnamefont {Volz}},\ and\ \bibinfo {author} {\bibfnamefont {G.~K.}\ \bibnamefont {Brennen}},\ }\bibfield  {title} {\bibinfo {title} {Geometric pathway to scalable quantum sensing},\ }\href {https://doi.org/10.1103/PhysRevLett.125.190403} {\bibfield  {journal} {\bibinfo  {journal} {Phys. Rev. Lett.}\ }\textbf {\bibinfo {volume} {125}},\ \bibinfo {pages} {190403} (\bibinfo {year} {2020})}\BibitemShut {NoStop}%
  \bibitem [{\citenamefont {Cohen}\ \emph {et~al.}(2019)\citenamefont {Cohen}, \citenamefont {Larocque}, \citenamefont {Bouchard}, \citenamefont {Nejadsattari}, \citenamefont {Gefen},\ and\ \citenamefont {Karimi}}]{Cohen2019}%
    \BibitemOpen
    \bibfield  {author} {\bibinfo {author} {\bibfnamefont {E.}~\bibnamefont {Cohen}}, \bibinfo {author} {\bibfnamefont {H.}~\bibnamefont {Larocque}}, \bibinfo {author} {\bibfnamefont {F.}~\bibnamefont {Bouchard}}, \bibinfo {author} {\bibfnamefont {F.}~\bibnamefont {Nejadsattari}}, \bibinfo {author} {\bibfnamefont {Y.}~\bibnamefont {Gefen}},\ and\ \bibinfo {author} {\bibfnamefont {E.}~\bibnamefont {Karimi}},\ }\bibfield  {title} {\bibinfo {title} {Geometric phase from {Aharonov}--{Bohm} to {Pancharatnam}--{Berry} and beyond},\ }\href {https://doi.org/10.1038/s42254-019-0071-1} {\bibfield  {journal} {\bibinfo  {journal} {Nat. Rev. Phys.}\ }\textbf {\bibinfo {volume} {1}},\ \bibinfo {pages} {437} (\bibinfo {year} {2019})}\BibitemShut {NoStop}%
  \bibitem [{\citenamefont {Cisowski}\ \emph {et~al.}(2022)\citenamefont {Cisowski}, \citenamefont {G{\"o}tte},\ and\ \citenamefont {Franke-Arnold}}]{Cisowski2022}%
    \BibitemOpen
    \bibfield  {author} {\bibinfo {author} {\bibfnamefont {C.}~\bibnamefont {Cisowski}}, \bibinfo {author} {\bibfnamefont {J.~B.}\ \bibnamefont {G{\"o}tte}},\ and\ \bibinfo {author} {\bibfnamefont {S.}~\bibnamefont {Franke-Arnold}},\ }\bibfield  {title} {\bibinfo {title} {Colloquium: Geometric phases of light: Insights from fiber bundle theory},\ }\href {https://doi.org/10.1103/RevModPhys.94.031001} {\bibfield  {journal} {\bibinfo  {journal} {Rev. Mod. Phys.}\ }\textbf {\bibinfo {volume} {94}},\ \bibinfo {pages} {031001} (\bibinfo {year} {2022})}\BibitemShut {NoStop}%
  \bibitem [{\citenamefont {Zanardi}\ and\ \citenamefont {Rasetti}(1999)}]{Zanardi1999}%
    \BibitemOpen
    \bibfield  {author} {\bibinfo {author} {\bibfnamefont {P.}~\bibnamefont {Zanardi}}\ and\ \bibinfo {author} {\bibfnamefont {M.}~\bibnamefont {Rasetti}},\ }\bibfield  {title} {\bibinfo {title} {Holonomic quantum computation},\ }\href {https://doi.org/10.1016/S0375-9601(99)00803-8} {\bibfield  {journal} {\bibinfo  {journal} {Phys. Lett. A}\ }\textbf {\bibinfo {volume} {264}},\ \bibinfo {pages} {94} (\bibinfo {year} {1999})}\BibitemShut {NoStop}%
  \bibitem [{\citenamefont {Jones}\ \emph {et~al.}(2000)\citenamefont {Jones}, \citenamefont {Vedral}, \citenamefont {Ekert},\ and\ \citenamefont {Castagnoli}}]{Jones2000}%
    \BibitemOpen
    \bibfield  {author} {\bibinfo {author} {\bibfnamefont {J.~A.}\ \bibnamefont {Jones}}, \bibinfo {author} {\bibfnamefont {V.}~\bibnamefont {Vedral}}, \bibinfo {author} {\bibfnamefont {A.}~\bibnamefont {Ekert}},\ and\ \bibinfo {author} {\bibfnamefont {G.}~\bibnamefont {Castagnoli}},\ }\bibfield  {title} {\bibinfo {title} {Geometric quantum computation using nuclear magnetic resonance},\ }\href {https://doi.org/10.1038/35002528} {\bibfield  {journal} {\bibinfo  {journal} {Nature}\ }\textbf {\bibinfo {volume} {403}},\ \bibinfo {pages} {869} (\bibinfo {year} {2000})}\BibitemShut {NoStop}%
  \bibitem [{\citenamefont {Sj\"oqvist}\ \emph {et~al.}(2012)\citenamefont {Sj\"oqvist}, \citenamefont {Tong}, \citenamefont {Andersson}, \citenamefont {Hessmo}, \citenamefont {Johansson},\ and\ \citenamefont {Singh}}]{Sjoqvist2012}%
    \BibitemOpen
    \bibfield  {author} {\bibinfo {author} {\bibfnamefont {E.}~\bibnamefont {Sj\"oqvist}}, \bibinfo {author} {\bibfnamefont {D.~M.}\ \bibnamefont {Tong}}, \bibinfo {author} {\bibfnamefont {L.~M.}\ \bibnamefont {Andersson}}, \bibinfo {author} {\bibfnamefont {B.}~\bibnamefont {Hessmo}}, \bibinfo {author} {\bibfnamefont {M.}~\bibnamefont {Johansson}},\ and\ \bibinfo {author} {\bibfnamefont {K.}~\bibnamefont {Singh}},\ }\bibfield  {title} {\bibinfo {title} {Non-adiabatic holonomic quantum computation},\ }\href {https://doi.org/10.1088/1367-2630/14/10/103035} {\bibfield  {journal} {\bibinfo  {journal} {New J. Phys.}\ }\textbf {\bibinfo {volume} {14}},\ \bibinfo {pages} {103035} (\bibinfo {year} {2012})}\BibitemShut {NoStop}%
  \bibitem [{\citenamefont {Leibfried}\ \emph {et~al.}(2003)\citenamefont {Leibfried}, \citenamefont {DeMarco}, \citenamefont {Meyer}, \citenamefont {Lucas}, \citenamefont {Barrett}, \citenamefont {Britton}, \citenamefont {Itano}, \citenamefont {Jelenkovi{\'c}}, \citenamefont {Langer}, \citenamefont {Rosenband},\ and\ \citenamefont {Wineland}}]{Leibfried2003}%
    \BibitemOpen
    \bibfield  {author} {\bibinfo {author} {\bibfnamefont {D.}~\bibnamefont {Leibfried}}, \bibinfo {author} {\bibfnamefont {B.}~\bibnamefont {DeMarco}}, \bibinfo {author} {\bibfnamefont {V.}~\bibnamefont {Meyer}}, \bibinfo {author} {\bibfnamefont {D.}~\bibnamefont {Lucas}}, \bibinfo {author} {\bibfnamefont {M.}~\bibnamefont {Barrett}}, \bibinfo {author} {\bibfnamefont {J.}~\bibnamefont {Britton}}, \bibinfo {author} {\bibfnamefont {W.~M.}\ \bibnamefont {Itano}}, \bibinfo {author} {\bibfnamefont {B.}~\bibnamefont {Jelenkovi{\'c}}}, \bibinfo {author} {\bibfnamefont {C.}~\bibnamefont {Langer}}, \bibinfo {author} {\bibfnamefont {T.}~\bibnamefont {Rosenband}},\ and\ \bibinfo {author} {\bibfnamefont {D.~J.}\ \bibnamefont {Wineland}},\ }\bibfield  {title} {\bibinfo {title} {Experimental demonstration of a robust, high-fidelity geometric two ion-qubit phase gate},\ }\href {https://doi.org/10.1038/nature01492} {\bibfield  {journal} {\bibinfo  {journal} {Nature}\ }\textbf {\bibinfo {volume} {422}},\ \bibinfo {pages} {412} (\bibinfo {year} {2003})}\BibitemShut {NoStop}%
  \bibitem [{\citenamefont {Colmenar}\ \emph {et~al.}(2022)\citenamefont {Colmenar}, \citenamefont {G{\"u}ng{\"o}rd{\"u}},\ and\ \citenamefont {Kestner}}]{Colmenar2022}%
    \BibitemOpen
    \bibfield  {author} {\bibinfo {author} {\bibfnamefont {R.~K.~L.}\ \bibnamefont {Colmenar}}, \bibinfo {author} {\bibfnamefont {U.}~\bibnamefont {G{\"u}ng{\"o}rd{\"u}}},\ and\ \bibinfo {author} {\bibfnamefont {J.~P.}\ \bibnamefont {Kestner}},\ }\bibfield  {title} {\bibinfo {title} {Conditions for equivalent noise sensitivity of geometric and dynamical quantum gates},\ }\href {https://doi.org/10.1103/PRXQuantum.3.030310} {\bibfield  {journal} {\bibinfo  {journal} {PRX Quantum}\ }\textbf {\bibinfo {volume} {3}},\ \bibinfo {pages} {030310} (\bibinfo {year} {2022})}\BibitemShut {NoStop}%
  \bibitem [{\citenamefont {Liu}\ \emph {et~al.}(2021)\citenamefont {Liu}, \citenamefont {Wang},\ and\ \citenamefont {Yung}}]{Liu2021}%
    \BibitemOpen
    \bibfield  {author} {\bibinfo {author} {\bibfnamefont {B.-J.}\ \bibnamefont {Liu}}, \bibinfo {author} {\bibfnamefont {Y.-S.}\ \bibnamefont {Wang}},\ and\ \bibinfo {author} {\bibfnamefont {M.-H.}\ \bibnamefont {Yung}},\ }\bibfield  {title} {\bibinfo {title} {Super-robust nonadiabatic geometric quantum control},\ }\href {https://doi.org/10.1103/PhysRevResearch.3.L032066} {\bibfield  {journal} {\bibinfo  {journal} {Phys. Rev. Research}\ }\textbf {\bibinfo {volume} {3}},\ \bibinfo {pages} {L032066} (\bibinfo {year} {2021})}\BibitemShut {NoStop}%
  \bibitem [{\citenamefont {Feynman}(1982)}]{Feynman1982}%
    \BibitemOpen
    \bibfield  {author} {\bibinfo {author} {\bibfnamefont {R.~P.}\ \bibnamefont {Feynman}},\ }\bibfield  {title} {\bibinfo {title} {Simulating physics with computers},\ }\href {https://doi.org/10.1007/BF02650179} {\bibfield  {journal} {\bibinfo  {journal} {Int. J. Theor. Phys.}\ }\textbf {\bibinfo {volume} {21}},\ \bibinfo {pages} {467} (\bibinfo {year} {1982})}\BibitemShut {NoStop}%
  \bibitem [{\citenamefont {Lloyd}(1996)}]{Lloyd1996}%
    \BibitemOpen
    \bibfield  {author} {\bibinfo {author} {\bibfnamefont {S.}~\bibnamefont {Lloyd}},\ }\bibfield  {title} {\bibinfo {title} {Universal quantum simulators},\ }\href {https://doi.org/10.1126/science.273.5278.1073} {\bibfield  {journal} {\bibinfo  {journal} {Science}\ }\textbf {\bibinfo {volume} {273}},\ \bibinfo {pages} {1073} (\bibinfo {year} {1996})}\BibitemShut {NoStop}%
  \bibitem [{\citenamefont {Georgescu}\ \emph {et~al.}(2014)\citenamefont {Georgescu}, \citenamefont {Ashhab},\ and\ \citenamefont {Nori}}]{Georgescu2014}%
    \BibitemOpen
    \bibfield  {author} {\bibinfo {author} {\bibfnamefont {I.~M.}\ \bibnamefont {Georgescu}}, \bibinfo {author} {\bibfnamefont {S.}~\bibnamefont {Ashhab}},\ and\ \bibinfo {author} {\bibfnamefont {F.}~\bibnamefont {Nori}},\ }\bibfield  {title} {\bibinfo {title} {Quantum simulation},\ }\href {https://doi.org/10.1103/RevModPhys.86.153} {\bibfield  {journal} {\bibinfo  {journal} {Rev. Mod. Phys.}\ }\textbf {\bibinfo {volume} {86}},\ \bibinfo {pages} {153} (\bibinfo {year} {2014})}\BibitemShut {NoStop}%
  \bibitem [{\citenamefont {Arute}\ \emph {et~al.}(2019)\citenamefont {Arute}, \citenamefont {Arya}, \citenamefont {Babbush}, \citenamefont {Bacon}, \citenamefont {Bardin}, \citenamefont {Barends}, \citenamefont {Biswas}, \citenamefont {Boixo}, \citenamefont {Brandao}, \citenamefont {Buell}, \citenamefont {Burkett}, \citenamefont {Chen}, \citenamefont {Chen}, \citenamefont {Chiaro}, \citenamefont {Collins}, \citenamefont {Courtney}, \citenamefont {Dunsworth}, \citenamefont {Farhi}, \citenamefont {Foxen}, \citenamefont {Fowler}, \citenamefont {Gidney}, \citenamefont {Giustina}, \citenamefont {Graff}, \citenamefont {Guerin}, \citenamefont {Habegger}, \citenamefont {Harrigan}, \citenamefont {Hartmann}, \citenamefont {Ho}, \citenamefont {Hoffmann}, \citenamefont {Huang}, \citenamefont {Humble}, \citenamefont {Isakov}, \citenamefont {Jeffrey}, \citenamefont {Jiang}, \citenamefont {Kafri}, \citenamefont {Kechedzhi}, \citenamefont {Kelly}, \citenamefont {Klimov}, \citenamefont {Knysh}, \citenamefont {Korotkov}, \citenamefont {Kostritsa}, \citenamefont {Landhuis}, \citenamefont {Lindmark}, \citenamefont {Lucero}, \citenamefont {Lyakh}, \citenamefont {Mandr{\`a}}, \citenamefont {McClean}, \citenamefont {McEwen}, \citenamefont {Megrant}, \citenamefont {Mi}, \citenamefont {Michielsen}, \citenamefont {Mohseni}, \citenamefont {Mutus}, \citenamefont {Naaman}, \citenamefont {Neeley}, \citenamefont {Neill}, \citenamefont {Niu}, \citenamefont {Ostby}, \citenamefont {Petukhov}, \citenamefont {Platt}, \citenamefont {Quintana}, \citenamefont {Rieffel}, \citenamefont {Roushan}, \citenamefont {Rubin}, \citenamefont {Sank}, \citenamefont {Satzinger}, \citenamefont {Smelyanskiy}, \citenamefont {Sung}, \citenamefont {Trevithick}, \citenamefont {Vainsencher}, \citenamefont {Villalonga}, \citenamefont {White}, \citenamefont {Yao}, \citenamefont {Yeh}, \citenamefont {Zalcman}, \citenamefont {Neven},\ and\ \citenamefont {Martinis}}]{Arute2019}%
    \BibitemOpen
    \bibfield  {author} {\bibinfo {author} {\bibfnamefont {F.}~\bibnamefont {Arute}}, \bibinfo {author} {\bibfnamefont {K.}~\bibnamefont {Arya}}, \bibinfo {author} {\bibfnamefont {R.}~\bibnamefont {Babbush}}, \bibinfo {author} {\bibfnamefont {D.}~\bibnamefont {Bacon}}, \bibinfo {author} {\bibfnamefont {J.~C.}\ \bibnamefont {Bardin}}, \bibinfo {author} {\bibfnamefont {R.}~\bibnamefont {Barends}}, \bibinfo {author} {\bibfnamefont {R.}~\bibnamefont {Biswas}}, \bibinfo {author} {\bibfnamefont {S.}~\bibnamefont {Boixo}}, \bibinfo {author} {\bibfnamefont {F.~G. S.~L.}\ \bibnamefont {Brandao}}, \bibinfo {author} {\bibfnamefont {D.~A.}\ \bibnamefont {Buell}}, \bibinfo {author} {\bibfnamefont {B.}~\bibnamefont {Burkett}}, \bibinfo {author} {\bibfnamefont {Y.}~\bibnamefont {Chen}}, \bibinfo {author} {\bibfnamefont {Z.}~\bibnamefont {Chen}}, \bibinfo {author} {\bibfnamefont {B.}~\bibnamefont {Chiaro}}, \bibinfo {author} {\bibfnamefont {R.}~\bibnamefont {Collins}}, \bibinfo {author} {\bibfnamefont {W.}~\bibnamefont {Courtney}}, \bibinfo {author} {\bibfnamefont {A.}~\bibnamefont {Dunsworth}}, \bibinfo {author} {\bibfnamefont {E.}~\bibnamefont {Farhi}}, \bibinfo {author} {\bibfnamefont {B.}~\bibnamefont {Foxen}}, \bibinfo {author} {\bibfnamefont {A.}~\bibnamefont {Fowler}}, \bibinfo {author} {\bibfnamefont {C.}~\bibnamefont {Gidney}}, \bibinfo {author} {\bibfnamefont {M.}~\bibnamefont {Giustina}}, \bibinfo {author} {\bibfnamefont {R.}~\bibnamefont {Graff}}, \bibinfo {author} {\bibfnamefont {K.}~\bibnamefont {Guerin}}, \bibinfo {author} {\bibfnamefont {S.}~\bibnamefont {Habegger}}, \bibinfo {author} {\bibfnamefont {M.~P.}\ \bibnamefont {Harrigan}}, \bibinfo {author} {\bibfnamefont {M.~J.}\ \bibnamefont {Hartmann}}, \bibinfo {author} {\bibfnamefont {A.}~\bibnamefont {Ho}}, \bibinfo {author} {\bibfnamefont {M.}~\bibnamefont {Hoffmann}}, \bibinfo {author} {\bibfnamefont {T.}~\bibnamefont {Huang}}, \bibinfo {author} {\bibfnamefont {T.~S.}\ \bibnamefont {Humble}}, \bibinfo {author} {\bibfnamefont {S.~V.}\ \bibnamefont {Isakov}}, \bibinfo {author} {\bibfnamefont {E.}~\bibnamefont {Jeffrey}}, \bibinfo {author} {\bibfnamefont {Z.}~\bibnamefont {Jiang}}, \bibinfo {author} {\bibfnamefont {D.}~\bibnamefont {Kafri}}, \bibinfo {author} {\bibfnamefont {K.}~\bibnamefont {Kechedzhi}}, \bibinfo {author} {\bibfnamefont {J.}~\bibnamefont {Kelly}}, \bibinfo {author} {\bibfnamefont {P.~V.}\ \bibnamefont {Klimov}}, \bibinfo {author} {\bibfnamefont {S.}~\bibnamefont {Knysh}}, \bibinfo {author} {\bibfnamefont {A.}~\bibnamefont {Korotkov}}, \bibinfo {author} {\bibfnamefont {F.}~\bibnamefont {Kostritsa}}, \bibinfo {author} {\bibfnamefont {D.}~\bibnamefont {Landhuis}}, \bibinfo {author} {\bibfnamefont {M.}~\bibnamefont {Lindmark}}, \bibinfo {author} {\bibfnamefont {E.}~\bibnamefont {Lucero}}, \bibinfo {author} {\bibfnamefont {D.}~\bibnamefont {Lyakh}}, \bibinfo {author} {\bibfnamefont {S.}~\bibnamefont {Mandr{\`a}}}, \bibinfo {author} {\bibfnamefont {J.~R.}\ \bibnamefont {McClean}}, \bibinfo {author} {\bibfnamefont {M.}~\bibnamefont {McEwen}}, \bibinfo {author} {\bibfnamefont {A.}~\bibnamefont {Megrant}}, \bibinfo {author} {\bibfnamefont {X.}~\bibnamefont {Mi}}, \bibinfo {author} {\bibfnamefont {K.}~\bibnamefont {Michielsen}}, \bibinfo {author} {\bibfnamefont {M.}~\bibnamefont {Mohseni}}, \bibinfo {author} {\bibfnamefont {J.}~\bibnamefont {Mutus}}, \bibinfo {author} {\bibfnamefont {O.}~\bibnamefont {Naaman}}, \bibinfo {author} {\bibfnamefont {M.}~\bibnamefont {Neeley}}, \bibinfo {author} {\bibfnamefont {C.}~\bibnamefont {Neill}}, \bibinfo {author} {\bibfnamefont {M.~Y.}\ \bibnamefont {Niu}}, \bibinfo {author} {\bibfnamefont {E.}~\bibnamefont {Ostby}}, \bibinfo {author} {\bibfnamefont {A.}~\bibnamefont {Petukhov}}, \bibinfo {author} {\bibfnamefont {J.~C.}\ \bibnamefont {Platt}}, \bibinfo {author} {\bibfnamefont {C.}~\bibnamefont {Quintana}}, \bibinfo {author} {\bibfnamefont {E.~G.}\ \bibnamefont {Rieffel}}, \bibinfo {author} {\bibfnamefont {P.}~\bibnamefont {Roushan}}, \bibinfo {author} {\bibfnamefont {N.~C.}\ \bibnamefont {Rubin}}, \bibinfo {author} {\bibfnamefont {D.}~\bibnamefont {Sank}}, \bibinfo {author} {\bibfnamefont {K.~J.}\ \bibnamefont {Satzinger}}, \bibinfo {author} {\bibfnamefont {V.}~\bibnamefont {Smelyanskiy}}, \bibinfo {author} {\bibfnamefont {K.~J.}\ \bibnamefont {Sung}}, \bibinfo {author} {\bibfnamefont {M.~D.}\ \bibnamefont {Trevithick}}, \bibinfo {author} {\bibfnamefont {A.}~\bibnamefont {Vainsencher}}, \bibinfo {author} {\bibfnamefont {B.}~\bibnamefont {Villalonga}}, \bibinfo {author} {\bibfnamefont {T.}~\bibnamefont {White}}, \bibinfo {author} {\bibfnamefont {Z.~J.}\ \bibnamefont {Yao}}, \bibinfo {author} {\bibfnamefont {P.}~\bibnamefont {Yeh}}, \bibinfo {author} {\bibfnamefont {A.}~\bibnamefont {Zalcman}}, \bibinfo {author} {\bibfnamefont {H.}~\bibnamefont {Neven}},\ and\ \bibinfo {author} {\bibfnamefont {J.~M.}\ \bibnamefont {Martinis}},\ }\bibfield  {title} {\bibinfo {title} {Quantum supremacy using a programmable superconducting processor},\ }\href {https://doi.org/10.1038/s41586-019-1666-5} {\bibfield  {journal} {\bibinfo  {journal} {Nature}\ }\textbf {\bibinfo {volume} {574}},\ \bibinfo {pages} {505} (\bibinfo {year} {2019})}\BibitemShut {NoStop}%
  \bibitem [{\citenamefont {Preskill}(2018)}]{Preskill2018}%
    \BibitemOpen
    \bibfield  {author} {\bibinfo {author} {\bibfnamefont {J.}~\bibnamefont {Preskill}},\ }\bibfield  {title} {\bibinfo {title} {Quantum computing in the {NISQ} era and beyond},\ }\href {https://doi.org/10.22331/q-2018-08-06-79} {\bibfield  {journal} {\bibinfo  {journal} {Quantum}\ }\textbf {\bibinfo {volume} {2}},\ \bibinfo {pages} {79} (\bibinfo {year} {2018})}\BibitemShut {NoStop}%
  \bibitem [{\citenamefont {Kim}\ \emph {et~al.}(2023)\citenamefont {Kim}, \citenamefont {Eddins}, \citenamefont {Anand}, \citenamefont {Wei}, \citenamefont {van~den Berg}, \citenamefont {Rosenblatt}, \citenamefont {Nayfeh}, \citenamefont {Wu}, \citenamefont {Zaletel}, \citenamefont {Temme},\ and\ \citenamefont {Kandala}}]{Kim2023}%
    \BibitemOpen
    \bibfield  {author} {\bibinfo {author} {\bibfnamefont {Y.}~\bibnamefont {Kim}}, \bibinfo {author} {\bibfnamefont {A.}~\bibnamefont {Eddins}}, \bibinfo {author} {\bibfnamefont {S.}~\bibnamefont {Anand}}, \bibinfo {author} {\bibfnamefont {K.~X.}\ \bibnamefont {Wei}}, \bibinfo {author} {\bibfnamefont {E.}~\bibnamefont {van~den Berg}}, \bibinfo {author} {\bibfnamefont {S.}~\bibnamefont {Rosenblatt}}, \bibinfo {author} {\bibfnamefont {H.}~\bibnamefont {Nayfeh}}, \bibinfo {author} {\bibfnamefont {Y.}~\bibnamefont {Wu}}, \bibinfo {author} {\bibfnamefont {M.}~\bibnamefont {Zaletel}}, \bibinfo {author} {\bibfnamefont {K.}~\bibnamefont {Temme}},\ and\ \bibinfo {author} {\bibfnamefont {A.}~\bibnamefont {Kandala}},\ }\bibfield  {title} {\bibinfo {title} {Evidence for the utility of quantum computing before fault tolerance},\ }\href {https://doi.org/10.1038/s41586-023-06096-3} {\bibfield  {journal} {\bibinfo  {journal} {Nature}\ }\textbf {\bibinfo {volume} {618}},\ \bibinfo {pages} {500} (\bibinfo {year} {2023})}\BibitemShut {NoStop}%
  \bibitem [{\citenamefont {Fukui}\ \emph {et~al.}(2005)\citenamefont {Fukui}, \citenamefont {Hatsugai},\ and\ \citenamefont {Suzuki}}]{FukuiHatsugaiSuzuki2005}%
    \BibitemOpen
    \bibfield  {author} {\bibinfo {author} {\bibfnamefont {T.}~\bibnamefont {Fukui}}, \bibinfo {author} {\bibfnamefont {Y.}~\bibnamefont {Hatsugai}},\ and\ \bibinfo {author} {\bibfnamefont {H.}~\bibnamefont {Suzuki}},\ }\bibfield  {title} {\bibinfo {title} {Chern numbers in discretized {Brillouin} zone: Efficient method of computing (spin) {Hall} conductances},\ }\href {https://doi.org/10.1143/JPSJ.74.1674} {\bibfield  {journal} {\bibinfo  {journal} {J. Phys. Soc. Jpn.}\ }\textbf {\bibinfo {volume} {74}},\ \bibinfo {pages} {1674} (\bibinfo {year} {2005})}\BibitemShut {NoStop}%
  \bibitem [{\citenamefont {Yu}\ \emph {et~al.}(2011)\citenamefont {Yu}, \citenamefont {Qi}, \citenamefont {Bernevig}, \citenamefont {Fang},\ and\ \citenamefont {Dai}}]{Yu2011}%
    \BibitemOpen
    \bibfield  {author} {\bibinfo {author} {\bibfnamefont {R.}~\bibnamefont {Yu}}, \bibinfo {author} {\bibfnamefont {X.~L.}\ \bibnamefont {Qi}}, \bibinfo {author} {\bibfnamefont {A.}~\bibnamefont {Bernevig}}, \bibinfo {author} {\bibfnamefont {Z.}~\bibnamefont {Fang}},\ and\ \bibinfo {author} {\bibfnamefont {X.}~\bibnamefont {Dai}},\ }\bibfield  {title} {\bibinfo {title} {Equivalent expression of $\mathbb{Z}_2$ topological invariant for band insulators using the non-{Abelian} {Berry} connection},\ }\href {https://doi.org/10.1103/PhysRevB.84.075119} {\bibfield  {journal} {\bibinfo  {journal} {Phys. Rev. B}\ }\textbf {\bibinfo {volume} {84}},\ \bibinfo {pages} {075119} (\bibinfo {year} {2011})}\BibitemShut {NoStop}%
  \bibitem [{\citenamefont {Soluyanov}\ and\ \citenamefont {Vanderbilt}(2011)}]{Soluyanov2011}%
    \BibitemOpen
    \bibfield  {author} {\bibinfo {author} {\bibfnamefont {A.~A.}\ \bibnamefont {Soluyanov}}\ and\ \bibinfo {author} {\bibfnamefont {D.}~\bibnamefont {Vanderbilt}},\ }\bibfield  {title} {\bibinfo {title} {Computing topological invariants without inversion symmetry},\ }\href {https://doi.org/10.1103/PhysRevB.83.235401} {\bibfield  {journal} {\bibinfo  {journal} {Phys. Rev. B}\ }\textbf {\bibinfo {volume} {83}},\ \bibinfo {pages} {235401} (\bibinfo {year} {2011})}\BibitemShut {NoStop}%
  \bibitem [{\citenamefont {Hayakawa}\ \emph {et~al.}(2025)\citenamefont {Hayakawa}, \citenamefont {Sakamoto},\ and\ \citenamefont {Kiumi}}]{HayakawaSakamotoKiumi2025}%
    \BibitemOpen
    \bibfield  {author} {\bibinfo {author} {\bibfnamefont {R.}~\bibnamefont {Hayakawa}}, \bibinfo {author} {\bibfnamefont {K.}~\bibnamefont {Sakamoto}},\ and\ \bibinfo {author} {\bibfnamefont {C.}~\bibnamefont {Kiumi}},\ }\href {https://doi.org/10.48550/arXiv.2509.13423} {\bibinfo {title} {Computational complexity of {Berry} phase estimation in topological phases of matter}} (\bibinfo {year} {2025}),\ \Eprint {https://arxiv.org/abs/2509.13423} {arXiv:2509.13423 [quant-ph]} \BibitemShut {NoStop}%
  \bibitem [{\citenamefont {Zhang}\ \emph {et~al.}(2010)\citenamefont {Zhang}, \citenamefont {Wei},\ and\ \citenamefont {Papageorgiou}}]{ZhangWei2010}%
    \BibitemOpen
    \bibfield  {author} {\bibinfo {author} {\bibfnamefont {C.}~\bibnamefont {Zhang}}, \bibinfo {author} {\bibfnamefont {Z.}~\bibnamefont {Wei}},\ and\ \bibinfo {author} {\bibfnamefont {A.}~\bibnamefont {Papageorgiou}},\ }\bibfield  {title} {\bibinfo {title} {Adiabatic quantum counting by geometric phase estimation},\ }\href {https://doi.org/10.1007/s11128-009-0132-y} {\bibfield  {journal} {\bibinfo  {journal} {Quantum Inf. Process.}\ }\textbf {\bibinfo {volume} {9}},\ \bibinfo {pages} {369} (\bibinfo {year} {2010})}\BibitemShut {NoStop}%
  \bibitem [{\citenamefont {Murta}\ \emph {et~al.}(2020)\citenamefont {Murta}, \citenamefont {Catarina},\ and\ \citenamefont {Fern{\'a}ndez-Rossier}}]{MurtaCatarinaFernandezRossier2020}%
    \BibitemOpen
    \bibfield  {author} {\bibinfo {author} {\bibfnamefont {B.}~\bibnamefont {Murta}}, \bibinfo {author} {\bibfnamefont {G.}~\bibnamefont {Catarina}},\ and\ \bibinfo {author} {\bibfnamefont {J.}~\bibnamefont {Fern{\'a}ndez-Rossier}},\ }\bibfield  {title} {\bibinfo {title} {Berry phase estimation in gate-based adiabatic quantum simulation},\ }\href {https://doi.org/10.1103/PhysRevA.101.020302} {\bibfield  {journal} {\bibinfo  {journal} {Phys. Rev. A}\ }\textbf {\bibinfo {volume} {101}},\ \bibinfo {pages} {020302} (\bibinfo {year} {2020})}\BibitemShut {NoStop}%
  \bibitem [{\citenamefont {Tamiya}\ \emph {et~al.}(2021)\citenamefont {Tamiya}, \citenamefont {Koh},\ and\ \citenamefont {Nakagawa}}]{TamiyaKohNakagawa2021}%
    \BibitemOpen
    \bibfield  {author} {\bibinfo {author} {\bibfnamefont {S.}~\bibnamefont {Tamiya}}, \bibinfo {author} {\bibfnamefont {S.}~\bibnamefont {Koh}},\ and\ \bibinfo {author} {\bibfnamefont {Y.~O.}\ \bibnamefont {Nakagawa}},\ }\bibfield  {title} {\bibinfo {title} {Calculating nonadiabatic couplings and {Berry}'s phase by variational quantum eigensolvers},\ }\href {https://doi.org/10.1103/PhysRevResearch.3.023244} {\bibfield  {journal} {\bibinfo  {journal} {Phys. Rev. Research}\ }\textbf {\bibinfo {volume} {3}},\ \bibinfo {pages} {023244} (\bibinfo {year} {2021})}\BibitemShut {NoStop}%
  \bibitem [{\citenamefont {Mootz}\ and\ \citenamefont {Yao}(2026)}]{MootzYao2026}%
    \BibitemOpen
    \bibfield  {author} {\bibinfo {author} {\bibfnamefont {M.}~\bibnamefont {Mootz}}\ and\ \bibinfo {author} {\bibfnamefont {Y.-X.}\ \bibnamefont {Yao}},\ }\bibfield  {title} {\bibinfo {title} {Efficient {Berry} phase calculation via adaptive variational quantum computing approach},\ }\href {https://doi.org/10.1063/5.0294540} {\bibfield  {journal} {\bibinfo  {journal} {APL Quantum}\ }\textbf {\bibinfo {volume} {3}},\ \bibinfo {pages} {016111} (\bibinfo {year} {2026})}\BibitemShut {NoStop}%
  \bibitem [{\citenamefont {Rigolin}\ \emph {et~al.}(2008)\citenamefont {Rigolin}, \citenamefont {Ortiz},\ and\ \citenamefont {Ponce}}]{RigolinOrtizPonce2008}%
    \BibitemOpen
    \bibfield  {author} {\bibinfo {author} {\bibfnamefont {G.}~\bibnamefont {Rigolin}}, \bibinfo {author} {\bibfnamefont {G.}~\bibnamefont {Ortiz}},\ and\ \bibinfo {author} {\bibfnamefont {V.~H.}\ \bibnamefont {Ponce}},\ }\bibfield  {title} {\bibinfo {title} {Beyond the quantum adiabatic approximation: Adiabatic perturbation theory},\ }\href {https://doi.org/10.1103/PhysRevA.78.052508} {\bibfield  {journal} {\bibinfo  {journal} {Phys. Rev. A}\ }\textbf {\bibinfo {volume} {78}},\ \bibinfo {pages} {052508} (\bibinfo {year} {2008})}\BibitemShut {NoStop}%
  \bibitem [{\citenamefont {Avron}\ \emph {et~al.}(1987)\citenamefont {Avron}, \citenamefont {Seiler},\ and\ \citenamefont {Yaffe}}]{AvronSeilerYaffe1987}%
    \BibitemOpen
    \bibfield  {author} {\bibinfo {author} {\bibfnamefont {J.~E.}\ \bibnamefont {Avron}}, \bibinfo {author} {\bibfnamefont {R.}~\bibnamefont {Seiler}},\ and\ \bibinfo {author} {\bibfnamefont {L.~G.}\ \bibnamefont {Yaffe}},\ }\bibfield  {title} {\bibinfo {title} {Adiabatic theorems and applications to the quantum {Hall} effect},\ }\href {https://doi.org/10.1007/BF01209015} {\bibfield  {journal} {\bibinfo  {journal} {Commun. Math. Phys.}\ }\textbf {\bibinfo {volume} {110}},\ \bibinfo {pages} {33} (\bibinfo {year} {1987})}\BibitemShut {NoStop}%
  \bibitem [{\citenamefont {Nenciu}(1993)}]{Nenciu1993}%
    \BibitemOpen
    \bibfield  {author} {\bibinfo {author} {\bibfnamefont {G.}~\bibnamefont {Nenciu}},\ }\bibfield  {title} {\bibinfo {title} {Linear adiabatic theory. exponential estimates},\ }\href {https://doi.org/10.1007/BF02096616} {\bibfield  {journal} {\bibinfo  {journal} {Commun. Math. Phys.}\ }\textbf {\bibinfo {volume} {152}},\ \bibinfo {pages} {479} (\bibinfo {year} {1993})}\BibitemShut {NoStop}%
  \bibitem [{\citenamefont {Jansen}\ \emph {et~al.}(2007)\citenamefont {Jansen}, \citenamefont {Ruskai},\ and\ \citenamefont {Seiler}}]{JansenRuskaiSeiler2007}%
    \BibitemOpen
    \bibfield  {author} {\bibinfo {author} {\bibfnamefont {S.}~\bibnamefont {Jansen}}, \bibinfo {author} {\bibfnamefont {M.-B.}\ \bibnamefont {Ruskai}},\ and\ \bibinfo {author} {\bibfnamefont {R.}~\bibnamefont {Seiler}},\ }\bibfield  {title} {\bibinfo {title} {Bounds for the adiabatic approximation with applications to quantum computation},\ }\href {https://doi.org/10.1063/1.2798382} {\bibfield  {journal} {\bibinfo  {journal} {J. Math. Phys.}\ }\textbf {\bibinfo {volume} {48}},\ \bibinfo {pages} {102111} (\bibinfo {year} {2007})}\BibitemShut {NoStop}%
  \bibitem [{\citenamefont {Lidar}\ \emph {et~al.}(2009)\citenamefont {Lidar}, \citenamefont {Rezakhani},\ and\ \citenamefont {Hamma}}]{LidarRezakhaniHamma2009}%
    \BibitemOpen
    \bibfield  {author} {\bibinfo {author} {\bibfnamefont {D.~A.}\ \bibnamefont {Lidar}}, \bibinfo {author} {\bibfnamefont {A.~T.}\ \bibnamefont {Rezakhani}},\ and\ \bibinfo {author} {\bibfnamefont {A.}~\bibnamefont {Hamma}},\ }\bibfield  {title} {\bibinfo {title} {Adiabatic approximation with exponential accuracy for many-body systems and quantum computation},\ }\href {https://doi.org/10.1063/1.3236685} {\bibfield  {journal} {\bibinfo  {journal} {J. Math. Phys.}\ }\textbf {\bibinfo {volume} {50}},\ \bibinfo {pages} {102106} (\bibinfo {year} {2009})}\BibitemShut {NoStop}%
  \bibitem [{\citenamefont {Campos~Venuti}\ and\ \citenamefont {Lidar}(2018)}]{CamposVenutiLidar2018}%
    \BibitemOpen
    \bibfield  {author} {\bibinfo {author} {\bibfnamefont {L.}~\bibnamefont {Campos~Venuti}}\ and\ \bibinfo {author} {\bibfnamefont {D.~A.}\ \bibnamefont {Lidar}},\ }\bibfield  {title} {\bibinfo {title} {Error reduction in quantum annealing using boundary cancellation: {O}nly the end matters},\ }\href {https://doi.org/10.1103/PhysRevA.98.022315} {\bibfield  {journal} {\bibinfo  {journal} {Phys. Rev. A}\ }\textbf {\bibinfo {volume} {98}},\ \bibinfo {pages} {022315} (\bibinfo {year} {2018})}\BibitemShut {NoStop}%
  \bibitem [{\citenamefont {Cohen}\ and\ \citenamefont {Oh}(2025)}]{CohenOh2025}%
    \BibitemOpen
    \bibfield  {author} {\bibinfo {author} {\bibfnamefont {T.~D.}\ \bibnamefont {Cohen}}\ and\ \bibinfo {author} {\bibfnamefont {H.}~\bibnamefont {Oh}},\ }\bibfield  {title} {\bibinfo {title} {Asymptotic errors in adiabatic evolution},\ }\href {https://doi.org/10.1103/PhysRevA.111.042612} {\bibfield  {journal} {\bibinfo  {journal} {Phys. Rev. A}\ }\textbf {\bibinfo {volume} {111}},\ \bibinfo {pages} {042612} (\bibinfo {year} {2025})}\BibitemShut {NoStop}%
  \bibitem [{\citenamefont {H{\"o}rmander}(2003)}]{HormanderAnalysisI}%
    \BibitemOpen
    \bibfield  {author} {\bibinfo {author} {\bibfnamefont {L.}~\bibnamefont {H{\"o}rmander}},\ }\href {https://doi.org/10.1007/978-3-642-61497-2} {\emph {\bibinfo {title} {The Analysis of Linear Partial Differential Operators I: Distribution Theory and Fourier Analysis}}},\ \bibinfo {edition} {2nd}\ ed.,\ Classics in Mathematics\ (\bibinfo  {publisher} {Springer},\ \bibinfo {year} {2003})\BibitemShut {NoStop}%
  \bibitem [{\citenamefont {Aharonov}\ and\ \citenamefont {Anandan}(1987)}]{AharonovAnandan1987}%
    \BibitemOpen
    \bibfield  {author} {\bibinfo {author} {\bibfnamefont {Y.}~\bibnamefont {Aharonov}}\ and\ \bibinfo {author} {\bibfnamefont {J.}~\bibnamefont {Anandan}},\ }\bibfield  {title} {\bibinfo {title} {Phase change during a cyclic quantum evolution},\ }\href {https://doi.org/10.1103/PhysRevLett.58.1593} {\bibfield  {journal} {\bibinfo  {journal} {Phys. Rev. Lett.}\ }\textbf {\bibinfo {volume} {58}},\ \bibinfo {pages} {1593} (\bibinfo {year} {1987})}\BibitemShut {NoStop}%
  \bibitem [{\citenamefont {Samuel}\ and\ \citenamefont {Bhandari}(1988)}]{SamuelBhandari1988}%
    \BibitemOpen
    \bibfield  {author} {\bibinfo {author} {\bibfnamefont {J.}~\bibnamefont {Samuel}}\ and\ \bibinfo {author} {\bibfnamefont {R.}~\bibnamefont {Bhandari}},\ }\bibfield  {title} {\bibinfo {title} {General setting for {Berry's} phase},\ }\href {https://doi.org/10.1103/PhysRevLett.60.2339} {\bibfield  {journal} {\bibinfo  {journal} {Phys. Rev. Lett.}\ }\textbf {\bibinfo {volume} {60}},\ \bibinfo {pages} {2339} (\bibinfo {year} {1988})}\BibitemShut {NoStop}%
  \bibitem [{\citenamefont {Sj{\"o}qvist}\ \emph {et~al.}(2000)\citenamefont {Sj{\"o}qvist}, \citenamefont {Pati}, \citenamefont {Ekert}, \citenamefont {Anandan}, \citenamefont {Ericsson}, \citenamefont {Oi},\ and\ \citenamefont {Vedral}}]{Sjoqvist2000}%
    \BibitemOpen
    \bibfield  {author} {\bibinfo {author} {\bibfnamefont {E.}~\bibnamefont {Sj{\"o}qvist}}, \bibinfo {author} {\bibfnamefont {A.~K.}\ \bibnamefont {Pati}}, \bibinfo {author} {\bibfnamefont {A.}~\bibnamefont {Ekert}}, \bibinfo {author} {\bibfnamefont {J.~S.}\ \bibnamefont {Anandan}}, \bibinfo {author} {\bibfnamefont {M.}~\bibnamefont {Ericsson}}, \bibinfo {author} {\bibfnamefont {D.~K.~L.}\ \bibnamefont {Oi}},\ and\ \bibinfo {author} {\bibfnamefont {V.}~\bibnamefont {Vedral}},\ }\bibfield  {title} {\bibinfo {title} {Geometric phases for mixed states in interferometry},\ }\href {https://doi.org/10.1103/PhysRevLett.85.2845} {\bibfield  {journal} {\bibinfo  {journal} {Phys. Rev. Lett.}\ }\textbf {\bibinfo {volume} {85}},\ \bibinfo {pages} {2845} (\bibinfo {year} {2000})}\BibitemShut {NoStop}%
  \bibitem [{\citenamefont {Tong}\ \emph {et~al.}(2004)\citenamefont {Tong}, \citenamefont {Sj{\"o}qvist}, \citenamefont {Kwek},\ and\ \citenamefont {Oh}}]{Tong2004}%
    \BibitemOpen
    \bibfield  {author} {\bibinfo {author} {\bibfnamefont {D.~M.}\ \bibnamefont {Tong}}, \bibinfo {author} {\bibfnamefont {E.}~\bibnamefont {Sj{\"o}qvist}}, \bibinfo {author} {\bibfnamefont {L.~C.}\ \bibnamefont {Kwek}},\ and\ \bibinfo {author} {\bibfnamefont {C.~H.}\ \bibnamefont {Oh}},\ }\bibfield  {title} {\bibinfo {title} {Kinematic approach to the mixed state geometric phase in nonunitary evolution},\ }\href {https://doi.org/10.1103/PhysRevLett.93.080405} {\bibfield  {journal} {\bibinfo  {journal} {Phys. Rev. Lett.}\ }\textbf {\bibinfo {volume} {93}},\ \bibinfo {pages} {080405} (\bibinfo {year} {2004})}\BibitemShut {NoStop}%
  \bibitem [{\citenamefont {Carollo}\ \emph {et~al.}(2003)\citenamefont {Carollo}, \citenamefont {Fuentes-Guridi}, \citenamefont {Fran{\c{c}}a~Santos},\ and\ \citenamefont {Vedral}}]{Carollo2003}%
    \BibitemOpen
    \bibfield  {author} {\bibinfo {author} {\bibfnamefont {A.}~\bibnamefont {Carollo}}, \bibinfo {author} {\bibfnamefont {I.}~\bibnamefont {Fuentes-Guridi}}, \bibinfo {author} {\bibfnamefont {M.}~\bibnamefont {Fran{\c{c}}a~Santos}},\ and\ \bibinfo {author} {\bibfnamefont {V.}~\bibnamefont {Vedral}},\ }\bibfield  {title} {\bibinfo {title} {Geometric phase in open systems},\ }\href {https://doi.org/10.1103/PhysRevLett.90.160402} {\bibfield  {journal} {\bibinfo  {journal} {Phys. Rev. Lett.}\ }\textbf {\bibinfo {volume} {90}},\ \bibinfo {pages} {160402} (\bibinfo {year} {2003})}\BibitemShut {NoStop}%
  \bibitem [{\citenamefont {Sarandy}\ and\ \citenamefont {Lidar}(2006)}]{SarandyLidar2006}%
    \BibitemOpen
    \bibfield  {author} {\bibinfo {author} {\bibfnamefont {M.~S.}\ \bibnamefont {Sarandy}}\ and\ \bibinfo {author} {\bibfnamefont {D.~A.}\ \bibnamefont {Lidar}},\ }\bibfield  {title} {\bibinfo {title} {Abelian and non-{Abelian} geometric phases in adiabatic open quantum systems},\ }\href {https://doi.org/10.1103/PhysRevA.73.062101} {\bibfield  {journal} {\bibinfo  {journal} {Phys. Rev. A}\ }\textbf {\bibinfo {volume} {73}},\ \bibinfo {pages} {062101} (\bibinfo {year} {2006})}\BibitemShut {NoStop}%
  \bibitem [{\citenamefont {Whitney}\ \emph {et~al.}(2005)\citenamefont {Whitney}, \citenamefont {Makhlin}, \citenamefont {Shnirman},\ and\ \citenamefont {Gefen}}]{Whitney2005}%
    \BibitemOpen
    \bibfield  {author} {\bibinfo {author} {\bibfnamefont {R.~S.}\ \bibnamefont {Whitney}}, \bibinfo {author} {\bibfnamefont {Y.}~\bibnamefont {Makhlin}}, \bibinfo {author} {\bibfnamefont {A.}~\bibnamefont {Shnirman}},\ and\ \bibinfo {author} {\bibfnamefont {Y.}~\bibnamefont {Gefen}},\ }\bibfield  {title} {\bibinfo {title} {Geometric nature of the environment-induced {Berry} phase and geometric dephasing},\ }\href {https://doi.org/10.1103/PhysRevLett.94.070407} {\bibfield  {journal} {\bibinfo  {journal} {Phys. Rev. Lett.}\ }\textbf {\bibinfo {volume} {94}},\ \bibinfo {pages} {070407} (\bibinfo {year} {2005})}\BibitemShut {NoStop}%
  \bibitem [{\citenamefont {Uhlmann}(1986)}]{Uhlmann1986}%
    \BibitemOpen
    \bibfield  {author} {\bibinfo {author} {\bibfnamefont {A.}~\bibnamefont {Uhlmann}},\ }\bibfield  {title} {\bibinfo {title} {Parallel transport and quantum holonomy along density operators},\ }\href {https://doi.org/10.1016/0034-4877(86)90055-8} {\bibfield  {journal} {\bibinfo  {journal} {Rep. Math. Phys.}\ }\textbf {\bibinfo {volume} {24}},\ \bibinfo {pages} {229} (\bibinfo {year} {1986})}\BibitemShut {NoStop}%
  \bibitem [{\citenamefont {Wilczek}\ and\ \citenamefont {Zee}(1984)}]{WilczekZee1984}%
    \BibitemOpen
    \bibfield  {author} {\bibinfo {author} {\bibfnamefont {F.}~\bibnamefont {Wilczek}}\ and\ \bibinfo {author} {\bibfnamefont {A.}~\bibnamefont {Zee}},\ }\bibfield  {title} {\bibinfo {title} {Appearance of gauge structure in simple dynamical systems},\ }\href {https://doi.org/10.1103/PhysRevLett.52.2111} {\bibfield  {journal} {\bibinfo  {journal} {Phys. Rev. Lett.}\ }\textbf {\bibinfo {volume} {52}},\ \bibinfo {pages} {2111} (\bibinfo {year} {1984})}\BibitemShut {NoStop}%
  \bibitem [{\citenamefont {Duan}\ \emph {et~al.}(2001)\citenamefont {Duan}, \citenamefont {Cirac},\ and\ \citenamefont {Zoller}}]{DuanCiracZoller2001}%
    \BibitemOpen
    \bibfield  {author} {\bibinfo {author} {\bibfnamefont {L.-M.}\ \bibnamefont {Duan}}, \bibinfo {author} {\bibfnamefont {J.~I.}\ \bibnamefont {Cirac}},\ and\ \bibinfo {author} {\bibfnamefont {P.}~\bibnamefont {Zoller}},\ }\bibfield  {title} {\bibinfo {title} {Geometric manipulation of trapped ions for quantum computation},\ }\href {https://doi.org/10.1126/science.1058835} {\bibfield  {journal} {\bibinfo  {journal} {Science}\ }\textbf {\bibinfo {volume} {292}},\ \bibinfo {pages} {1695} (\bibinfo {year} {2001})}\BibitemShut {NoStop}%
  \bibitem [{\citenamefont {Thouless}(1983)}]{Thouless1983}%
    \BibitemOpen
    \bibfield  {author} {\bibinfo {author} {\bibfnamefont {D.~J.}\ \bibnamefont {Thouless}},\ }\bibfield  {title} {\bibinfo {title} {Quantization of particle transport},\ }\href {https://doi.org/10.1103/PhysRevB.27.6083} {\bibfield  {journal} {\bibinfo  {journal} {Phys. Rev. B}\ }\textbf {\bibinfo {volume} {27}},\ \bibinfo {pages} {6083} (\bibinfo {year} {1983})}\BibitemShut {NoStop}%
  \bibitem [{\citenamefont {Provost}\ and\ \citenamefont {Vall{\'e}e}(1980)}]{ProvostVallee1980}%
    \BibitemOpen
    \bibfield  {author} {\bibinfo {author} {\bibfnamefont {J.-P.}\ \bibnamefont {Provost}}\ and\ \bibinfo {author} {\bibfnamefont {G.}~\bibnamefont {Vall{\'e}e}},\ }\bibfield  {title} {\bibinfo {title} {Riemannian structure on manifolds of quantum states},\ }\href {https://doi.org/10.1007/BF02193559} {\bibfield  {journal} {\bibinfo  {journal} {Commun. Math. Phys.}\ }\textbf {\bibinfo {volume} {76}},\ \bibinfo {pages} {289} (\bibinfo {year} {1980})}\BibitemShut {NoStop}%
  \bibitem [{\citenamefont {Kolodrubetz}\ \emph {et~al.}(2017)\citenamefont {Kolodrubetz}, \citenamefont {Sels}, \citenamefont {Mehta},\ and\ \citenamefont {Polkovnikov}}]{Kolodrubetz2017}%
    \BibitemOpen
    \bibfield  {author} {\bibinfo {author} {\bibfnamefont {M.}~\bibnamefont {Kolodrubetz}}, \bibinfo {author} {\bibfnamefont {D.}~\bibnamefont {Sels}}, \bibinfo {author} {\bibfnamefont {P.}~\bibnamefont {Mehta}},\ and\ \bibinfo {author} {\bibfnamefont {A.}~\bibnamefont {Polkovnikov}},\ }\bibfield  {title} {\bibinfo {title} {Geometry and non-adiabatic response in quantum and classical systems},\ }\href {https://doi.org/10.1016/j.physrep.2017.07.001} {\bibfield  {journal} {\bibinfo  {journal} {Phys. Rep.}\ }\textbf {\bibinfo {volume} {697}},\ \bibinfo {pages} {1} (\bibinfo {year} {2017})}\BibitemShut {NoStop}%
  \bibitem [{\citenamefont {Ozawa}\ and\ \citenamefont {Goldman}(2018)}]{OzawaGoldman2018}%
    \BibitemOpen
    \bibfield  {author} {\bibinfo {author} {\bibfnamefont {T.}~\bibnamefont {Ozawa}}\ and\ \bibinfo {author} {\bibfnamefont {N.}~\bibnamefont {Goldman}},\ }\bibfield  {title} {\bibinfo {title} {Extracting the quantum metric tensor through periodic driving},\ }\href {https://doi.org/10.1103/PhysRevB.97.201117} {\bibfield  {journal} {\bibinfo  {journal} {Phys. Rev. B}\ }\textbf {\bibinfo {volume} {97}},\ \bibinfo {pages} {201117} (\bibinfo {year} {2018})}\BibitemShut {NoStop}%
  \bibitem [{\citenamefont {Childs}\ \emph {et~al.}(2021)\citenamefont {Childs}, \citenamefont {Su}, \citenamefont {Tran}, \citenamefont {Wiebe},\ and\ \citenamefont {Zhu}}]{ChildsSuTranWiebeZhu2021}%
    \BibitemOpen
    \bibfield  {author} {\bibinfo {author} {\bibfnamefont {A.~M.}\ \bibnamefont {Childs}}, \bibinfo {author} {\bibfnamefont {Y.}~\bibnamefont {Su}}, \bibinfo {author} {\bibfnamefont {M.~C.}\ \bibnamefont {Tran}}, \bibinfo {author} {\bibfnamefont {N.}~\bibnamefont {Wiebe}},\ and\ \bibinfo {author} {\bibfnamefont {S.}~\bibnamefont {Zhu}},\ }\bibfield  {title} {\bibinfo {title} {Theory of {Trotter} error with commutator scaling},\ }\href {https://doi.org/10.1103/PhysRevX.11.011020} {\bibfield  {journal} {\bibinfo  {journal} {Phys. Rev. X}\ }\textbf {\bibinfo {volume} {11}},\ \bibinfo {pages} {011020} (\bibinfo {year} {2021})}\BibitemShut {NoStop}%
  \bibitem [{\citenamefont {De~Chiara}\ and\ \citenamefont {Palma}(2003)}]{DeChiaraPalma2003}%
    \BibitemOpen
    \bibfield  {author} {\bibinfo {author} {\bibfnamefont {G.}~\bibnamefont {De~Chiara}}\ and\ \bibinfo {author} {\bibfnamefont {G.~M.}\ \bibnamefont {Palma}},\ }\bibfield  {title} {\bibinfo {title} {Berry phase for a spin $1/2$ particle in a classical fluctuating field},\ }\href {https://doi.org/10.1103/PhysRevLett.91.090404} {\bibfield  {journal} {\bibinfo  {journal} {Phys. Rev. Lett.}\ }\textbf {\bibinfo {volume} {91}},\ \bibinfo {pages} {090404} (\bibinfo {year} {2003})}\BibitemShut {NoStop}%
  \bibitem [{\citenamefont {Filipp}\ \emph {et~al.}(2009)\citenamefont {Filipp}, \citenamefont {Klepp}, \citenamefont {Hasegawa}, \citenamefont {Plonka-Spehr}, \citenamefont {Schmidt}, \citenamefont {Geltenbort},\ and\ \citenamefont {Rauch}}]{Filipp2009}%
    \BibitemOpen
    \bibfield  {author} {\bibinfo {author} {\bibfnamefont {S.}~\bibnamefont {Filipp}}, \bibinfo {author} {\bibfnamefont {J.}~\bibnamefont {Klepp}}, \bibinfo {author} {\bibfnamefont {Y.}~\bibnamefont {Hasegawa}}, \bibinfo {author} {\bibfnamefont {C.}~\bibnamefont {Plonka-Spehr}}, \bibinfo {author} {\bibfnamefont {U.}~\bibnamefont {Schmidt}}, \bibinfo {author} {\bibfnamefont {P.}~\bibnamefont {Geltenbort}},\ and\ \bibinfo {author} {\bibfnamefont {H.}~\bibnamefont {Rauch}},\ }\bibfield  {title} {\bibinfo {title} {Experimental demonstration of the stability of {Berry}'s phase for a spin-$1/2$ particle},\ }\href {https://doi.org/10.1103/PhysRevLett.102.030404} {\bibfield  {journal} {\bibinfo  {journal} {Phys. Rev. Lett.}\ }\textbf {\bibinfo {volume} {102}},\ \bibinfo {pages} {030404} (\bibinfo {year} {2009})}\BibitemShut {NoStop}%
  \bibitem [{\citenamefont {Zhu}\ and\ \citenamefont {Zanardi}(2005)}]{ZhuZanardi2005}%
    \BibitemOpen
    \bibfield  {author} {\bibinfo {author} {\bibfnamefont {S.-L.}\ \bibnamefont {Zhu}}\ and\ \bibinfo {author} {\bibfnamefont {P.}~\bibnamefont {Zanardi}},\ }\bibfield  {title} {\bibinfo {title} {Geometric quantum gates that are robust against stochastic control errors},\ }\href {https://doi.org/10.1103/PhysRevA.72.020301} {\bibfield  {journal} {\bibinfo  {journal} {Phys. Rev. A}\ }\textbf {\bibinfo {volume} {72}},\ \bibinfo {pages} {020301(R)} (\bibinfo {year} {2005})}\BibitemShut {NoStop}%
  \bibitem [{\citenamefont {Granet}\ and\ \citenamefont {Dreyer}(2024)}]{Granet2024}%
    \BibitemOpen
    \bibfield  {author} {\bibinfo {author} {\bibfnamefont {E.}~\bibnamefont {Granet}}\ and\ \bibinfo {author} {\bibfnamefont {H.}~\bibnamefont {Dreyer}},\ }\bibfield  {title} {\bibinfo {title} {Hamiltonian dynamics on digital quantum computers without discretization error},\ }\href {https://doi.org/10.1038/s41534-024-00877-y} {\bibfield  {journal} {\bibinfo  {journal} {npj Quantum Inf.}\ }\textbf {\bibinfo {volume} {10}},\ \bibinfo {pages} {82} (\bibinfo {year} {2024})}\BibitemShut {NoStop}%
  \bibitem [{\citenamefont {Kiumi}\ and\ \citenamefont {Koczor}(2025)}]{KiumiKoczor2025}%
    \BibitemOpen
    \bibfield  {author} {\bibinfo {author} {\bibfnamefont {C.}~\bibnamefont {Kiumi}}\ and\ \bibinfo {author} {\bibfnamefont {B.}~\bibnamefont {Koczor}},\ }\bibfield  {title} {\bibinfo {title} {{TE-PAI}: Exact time evolution by sampling random circuits},\ }\href {https://doi.org/10.1088/2058-9565/ae1160} {\bibfield  {journal} {\bibinfo  {journal} {Quantum Sci. Technol.}\ }\textbf {\bibinfo {volume} {10}},\ \bibinfo {pages} {045071} (\bibinfo {year} {2025})}\BibitemShut {NoStop}%
  \bibitem [{\citenamefont {Pages}\ \emph {et~al.}(2026)\citenamefont {Pages}, \citenamefont {Kiumi}, \citenamefont {Morohoshi}, \citenamefont {Koczor},\ and\ \citenamefont {Mitarai}}]{PagesKiumiMorohoshi2026}%
    \BibitemOpen
    \bibfield  {author} {\bibinfo {author} {\bibfnamefont {H.}~\bibnamefont {Pages}}, \bibinfo {author} {\bibfnamefont {C.}~\bibnamefont {Kiumi}}, \bibinfo {author} {\bibfnamefont {Y.}~\bibnamefont {Morohoshi}}, \bibinfo {author} {\bibfnamefont {B.}~\bibnamefont {Koczor}},\ and\ \bibinfo {author} {\bibfnamefont {K.}~\bibnamefont {Mitarai}},\ }\href {https://doi.org/10.48550/arXiv.2601.13881} {\bibinfo {title} {Low-resource quantum energy gap estimation via randomization}} (\bibinfo {year} {2026}),\ \Eprint {https://arxiv.org/abs/2601.13881} {arXiv:2601.13881 [quant-ph]} \BibitemShut {NoStop}%
  \bibitem [{\citenamefont {Hayata}\ and\ \citenamefont {Kikuchi}(2026)}]{HayataKikuchi2026}%
    \BibitemOpen
    \bibfield  {author} {\bibinfo {author} {\bibfnamefont {T.}~\bibnamefont {Hayata}}\ and\ \bibinfo {author} {\bibfnamefont {Y.}~\bibnamefont {Kikuchi}},\ }\href {https://doi.org/10.48550/ARXIV.2604.02854} {\bibinfo {title} {Continuous-time evolution via probabilistic angle interpolation and its applications}} (\bibinfo {year} {2026}),\ \Eprint {https://arxiv.org/abs/2604.02854} {arXiv:2604.02854 [quant-ph]} \BibitemShut {NoStop}%
  \bibitem [{\citenamefont {Carrera~Vazquez}\ \emph {et~al.}(2023)\citenamefont {Carrera~Vazquez}, \citenamefont {Egger}, \citenamefont {Ochsner},\ and\ \citenamefont {Woerner}}]{CarreraVazquezEggerOchsnerWoerner2023}%
    \BibitemOpen
    \bibfield  {author} {\bibinfo {author} {\bibfnamefont {A.}~\bibnamefont {Carrera~Vazquez}}, \bibinfo {author} {\bibfnamefont {D.~J.}\ \bibnamefont {Egger}}, \bibinfo {author} {\bibfnamefont {D.}~\bibnamefont {Ochsner}},\ and\ \bibinfo {author} {\bibfnamefont {S.}~\bibnamefont {Woerner}},\ }\bibfield  {title} {\bibinfo {title} {Well-conditioned multi-product formulas for hardware-friendly {Hamiltonian} simulation},\ }\href {https://doi.org/10.22331/q-2023-07-25-1067} {\bibfield  {journal} {\bibinfo  {journal} {Quantum}\ }\textbf {\bibinfo {volume} {7}},\ \bibinfo {pages} {1067} (\bibinfo {year} {2023})}\BibitemShut {NoStop}%
  \bibitem [{\citenamefont {Watson}\ and\ \citenamefont {Watkins}(2025)}]{Watson2025}%
    \BibitemOpen
    \bibfield  {author} {\bibinfo {author} {\bibfnamefont {J.~D.}\ \bibnamefont {Watson}}\ and\ \bibinfo {author} {\bibfnamefont {J.}~\bibnamefont {Watkins}},\ }\bibfield  {title} {\bibinfo {title} {Exponentially reduced circuit depths using {Trotter} error mitigation},\ }\href {https://doi.org/10.1103/kw39-yxq5} {\bibfield  {journal} {\bibinfo  {journal} {PRX Quantum}\ }\textbf {\bibinfo {volume} {6}},\ \bibinfo {pages} {030325} (\bibinfo {year} {2025})}\BibitemShut {NoStop}%
  \end{thebibliography}
\end{document}